\documentclass[aps,pra,twocolumn,superscriptaddress,floatfix, notitlepage]{revtex4-2}
\usepackage{amssymb}
\usepackage{amsmath}
\usepackage{bm}
\usepackage{dsfont}
\usepackage{graphicx}

\usepackage{float}
\usepackage{tabularx}
\usepackage{dcolumn}
\usepackage{isotope}
\usepackage{color}
\usepackage{xcolor}
\usepackage{caption}
\usepackage{subcaption}
\usepackage[normalem]{ulem}
\usepackage{lipsum}
\DeclareMathOperator{\sinc}{sinc}
\usepackage[colorlinks=true, allcolors=blue]{hyperref}
\usepackage{cleveref} 
\usepackage{physics}

\usepackage{amsfonts}
\usepackage{amsthm}

\usepackage{braket}

\usepackage{mathtools}
\usepackage{mhchem}

\newtheorem{theorem}{Theorem}
\newtheorem{lemma}{Lemma}
\newtheorem{corollary}{Corollary}

\newenvironment{relemma}[1]
  {\renewcommand{\thelemma}{\ref{#1}}\begin{lemma}}
  {\end{lemma}}

\newenvironment{recorollary}[1]
  {\renewcommand{\thecorollary}{\ref{#1}}\begin{corollary}}
  {\end{corollary}}

\DeclareMathOperator*{\Var}{Var}
\DeclareMathOperator\Vol{Vol}

\gdef\omegaEig{\lambda}

\begin{document}

\title{Classical feature map surrogates and metrics for quantum control landscapes}

\author{Martino~Calzavara}
\email[]{m.calzavara@fz-juelich.de}
\affiliation{Peter Grünberg Institute - Quantum Control (PGI-8), Forschungszentrum Jülich GmbH, Wilhelm-Johnen-Straße, 52428 Jülich, Germany}
\affiliation{Institute for Theoretical Physics, University of Cologne, Zülpicher Straße 77, 50937 Cologne, Germany}

\author{Tommaso~Calarco}
\affiliation{Peter Grünberg Institute - Quantum Control (PGI-8), Forschungszentrum Jülich GmbH, Wilhelm-Johnen-Straße, 52428 Jülich, Germany}
\affiliation{Institute for Theoretical Physics, University of Cologne, Zülpicher Straße 77, 50937 Cologne, Germany}
\affiliation{Dipartimento di Fisica e Astronomia, Università di Bologna, 40127 Bologna, Italy}

\author{Felix~Motzoi}
\email[]{f.motzoi@fz-juelich.de}
\affiliation{Peter Grünberg Institute - Quantum Control (PGI-8), Forschungszentrum Jülich GmbH, Wilhelm-Johnen-Straße, 52428 Jülich, Germany}
\affiliation{Institute for Theoretical Physics, University of Cologne, Zülpicher Straße 77, 50937 Cologne, Germany}

\date{\today}

\begin{abstract}
We derive and analyze three feature map representations of parametrized quantum dynamics, which generalize variational quantum circuits. These are (i) a Lie-Fourier partial sum, (ii) a Taylor expansion, and (iii) a finite-dimensional sinc kernel regression representation.  The Lie-Fourier representation is shown to have a dense spectrum with discrete peaks, that reflects control Hamiltonian properties, but that is also compressible in typically found symmetric systems. We prove boundedness in the spectrum and the cost function derivatives, and discrete symmetries of the coefficients, with implications for learning and simulation. We further show the landscape is Lipschitz continuous, linking global variation bounds to local Taylor approximation error—key for step size selection, convergence estimates, and stopping criteria in optimization. This also provides a new form of polynomial barren plateaux originating from the Lie-Fourier structure of the quantum dynamics. These results may find application in local and general surrogate model learning, which we benchmark numerically, in characterizations of hardness in the problem instances, and for meta-parameter heuristics in quantum optimizers.
\end{abstract}

\maketitle

\section{Introduction}\label{introduction}
Many functional processes in the natural and engineering sciences suffer from high complexity or limited observability. As a result, they are often treated and interrogated as \emph{black box oracles} -- systems with unknown structure or landscape. This challenge concerns a wide range of disparate fields, including condensed matter, control theory, complexity theory, machine learning, and quantum optimization. The respective approaches to modeling the landscapes reveal different but complementary methods with large overlaps and synergies. 

A central question across these domains is how to characterize the prescriptive performance of the system, that is, how well it performs under a given configuration, and, by extension, how accurately a (partial) model can predict or optimize that performance. In condensed matter physics, average behaviour is analysed to uncover landscape features such as phase transitions and the prevalence of metastable states (or traps). In contrast, complexity theory focuses on worst case performance, often emphasizing the difficulty of sampling in landscapes plagued by barren plateaus that hinder the search for global optima. Machine learning aims to approximate and reproduce the functional behavior of the black-box system, with a focus on model expressivity. Conversely, control theory seeks to directly optimize system outputs, raising questions about controllability in relation to high-dimensional input spaces.

In this work, we examine this multifaceted problem through the lens of quantum optimization and optimal control theory. Many classical concepts, such as phase transitions, barren plateaus, and landscape complexity, have quantum analogues, and corresponding methodologies have been adapted to the quantum setting. For instance, phase transitions are known to impact optimization hardness, and have been studied in quantum optimal control (QOC) \cite{Bukov18_2qb, Bukov19, Dalgaard22, Beato24}. Similarly, insights from classical machine learning have shaped our understanding of variational quantum algorithms (VQAs) and quantum machine learning (QML), often highlighting negative results due to barren plateaus \cite{McClean18, Holmes22, Larocca22}. The presence of traps and other obstructions has also been investigated in the QOC literature \cite{Rabitz04, Pechen11, DeFouquieres13, Zhdanov18}, alongside various metrics to quantify the hardness of the optimization landscape \cite{Larocca18, Shen06, Hsieh09, Nanduri13, Dalgaard22}.  Efforts to learn and predict the quantum cost landscape have employed techniques ranging from Fourier analysis \cite{Koczor22QAD, Landman22, Schreiber23, Fontana22, Rudolph23, Fontana23} to Gaussian process \cite{Kokail19, Sauvage20, Dalgaard22} and neural network \cite{Dalgaard22} Ans{\"a}tze.

Rather than relying on physically inspired Ans{\"a}tze, our approach derives first-principles constraints that characterize landscape structure and its surrogates under general assumptions. We analyze the quantum functional mapping controls to figure of merit (e.g. fidelity) using three representations: (i) a Fourier basis expansion, (ii) a Taylor expansion with respect to the controls, and (iii) a low-dimensional bandwidth-limited kernel. Notably, we show that VQA and QML landscapes can be viewed as special cases of the general QOC landscape, where controls act sequentially and time-ordering becomes trivial.

We derive upper bounds on gradient and higher-order derivatives, bounds relating control variations to cost changes, and bounds on the variance of the cost landscape, which plays a critical role in the emergence of barren plateaus. These derivative bounds enable a quantitative analysis of the Taylor expansion error, establishing its local efficiency and offering a global Lipschitz bound that can guide optimizer design, including stopping conditions and global sampling strategies. Conversely, our variance analysis indicates that, even under ideal conditions, the fraction of the landscape with meaningful variation shrinks rapidly with increasing problem dimensionality.

Finally, we compare the representational power of the three expansions. We find that the kernel representation is the most expressive in the limit of dense sampling, while the Fourier and Taylor representations are more effective in low-data regimes. Our results provide quantitative dependencies of landscape complexity on key parameters such as control operator spectral bandwidth, total evolution time (or circuit depth), and system size -- offering practical insights for algorithm design and parameter tuning.

The manuscript is organized as follows: in Sec.~\ref{technical_intro} we define the notion of quantum dynamical landscape and introduce the three representations that we will use to study the problem. These are later investigated in depth within dedicated sections: Sec.~\ref{lie_fourier} for the Fourier-, Sec.~\ref{taylor} for the Taylor- and Sec.~\ref{kernel} for the sinc representation. In Sec.~\ref{derivatives} we prove \textit{en passant} some properties of the landscape derivatives (using the results from Sec.~\ref{lie_fourier}). We then highlight the relevance of all previous results in the context of landscape optimization and of optimizer design respectively in Sec.~\ref{metrics} and Sec.~\ref{applications}. Finally, we draw our conclusions in Sec.~\ref{conclusions}. The Appendices are primarily devoted to the mathematical proofs of the main results (especially App.~\ref{appendix_c} and \ref{appendix_d}), and to related details like notation (App.~\ref{appendix_a}), problem standardization (App.~\ref{appendix_b}) and numerical examples (App.~\ref{appendix_e}).

\section{Problem setting and overview}\label{technical_intro}
The starting point of our discussion is a finite $D$-dimensional evolving quantum system, with states $\ket{\psi} \in \mathcal{H} \simeq \mathbb{C}^D$. To consider the general case for parametric unitary operations and extend the simple structure of layered gates, we model the circuit evolution through its generator given by a Hamiltonian $\hat{H}(t)$ that depends on time through a single bounded external control $u(t) \in [-u_{\mathrm{max}},u_{\mathrm{max}}]$. To perform time ordering, we discretize the control to be piecewise-constant on an $N$-timestep grid (with uniform step $\delta t = T/N$) from $t=0$ to $t=T$. Note that while $\delta t$ can be made arbitrarily small to retain precision, in quantum circuits one usually considers gates generated on large single timesteps, while in control theory the sampling rate of the control sets a practical upper bound on $\delta t$ \cite{motzoi2011optimal}. 

We consider the standard bilinear form of the controls given by 
\[ \hat{H}[u(t)] = \hat{H}_d + \hat{H}_c u(t)\]
which we write in vector form (see App.~\ref{appendix_a} for more details about notation) as
\[ \hat{H}(u_{\nu}) = \hat{H}_d + \hat{H}_c u_{\nu},\]
where $\boldsymbol{u} \in \mathcal{C}^N := [-u_{\mathrm{max}},u_{\mathrm{max}}]^N$, defining a region of interest as a hypercube. 
In order to obtain a minimally cumbersome treatment, all results in the main text are discussed within this single control setting. Nevertheless, most of them are proved for multiple controls in App.~\ref{appendix_c}, yielding the results for a single control as a special case.

Starting from the initial state $\ket{\psi(0)}$, the system will evolve to 
\[ \ket{\psi(T)} = \hat{U}(\boldsymbol{u})\ket{\psi(0)} \]
at time $T$, where the time-ordering is given straightforwardly by
\begin{equation}
    \hat{U}(\boldsymbol{u}) = \hat{U}(u_N) \cdots \hat{U}(u_1) = e^{-i\delta t \hat{H}(u_{N})} \cdots e^{-i\delta t \hat{H}(u_{1})}.
    \label{eq:timeordering}
\end{equation}
This unitary operator defines a generalized Parametrized Quantum Circuit (PQC), where the drift and control Hamiltonian $\hat{H}_d$ and $\hat{H}_c$ that generate the gates can in general act at the same time. This is in contrast with the typical VQA setting, where usually it is assumed that only one generator is acting within the gates defining each circuit layer.

\begin{figure*}[]
    \centering
    \includegraphics[width=\textwidth]{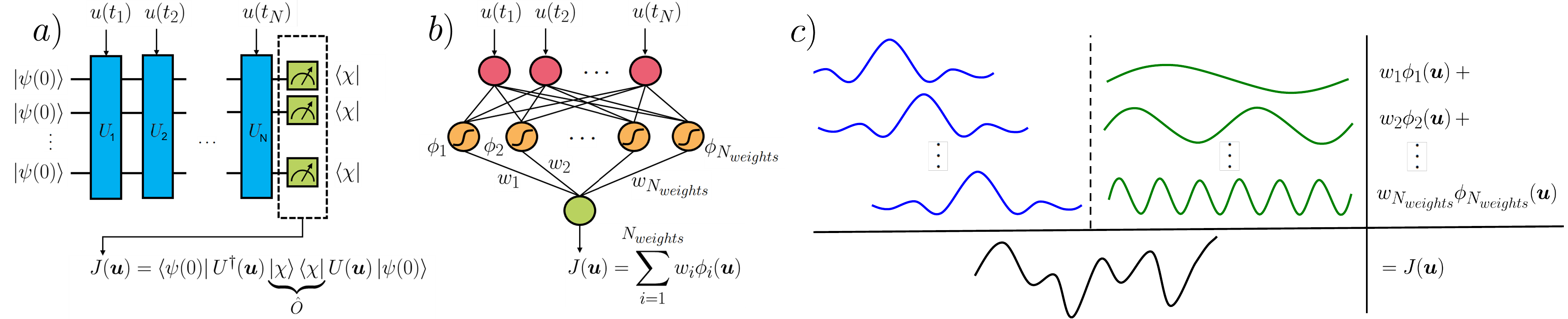}
    \caption{The time-evolved expectation value $\langle \hat{O}(\boldsymbol{u})\rangle$ of an observable (e.g. state fidelity $\hat{O} = \ket{\chi} \bra{\chi}$) for a system controlled through stepwise-constant pulses $\boldsymbol{u}$ is a Parametrized Quantum Circuit (Panel $(a)$). As such, it can be represented as a shallow computational network (Panel $(b)$), which consists in a linear combination of non-linear functions $\phi_i(\boldsymbol{u})$ called ``features", with weights $w_i$ (Panel $(c)$).}
    \label{fig:cl_sur}
\end{figure*}

We define a quantum dynamical landscape $J(\boldsymbol{u})$ as the expectation value of a Hermitian operator $\hat{O}$ over the output of the circuit $\ket{\psi(T)}$:
\begin{equation} 
J(\boldsymbol{u}) = \bra{\psi(T)}\hat{O}\ket{\psi(T)} = \bra{\psi(0)}\hat{U}^{\dagger}(\boldsymbol{u})\hat{O}\hat{U}(\boldsymbol{u})\ket{\psi(0)}. \label{eq:fid}
\end{equation}
In the VQA setting, $\hat{O}$ is typically expressed as a sum of observables that can be readily measured in a typical experimental setup (e.g. Pauli strings) \cite{Kokail19}. 
The case of optimal control for a state transfer problem corresponds instead to choosing $\ket{\psi} = \ket{\psi(0)}$ as initial state and $\hat{O} = \ket{\chi}\bra{\chi}$ as the density matrix of the (pure) target state. In this case, $J(\boldsymbol{u})$ corresponds to the state fidelity that we want to maximize over the controls $\boldsymbol{u}$. 
The average gate fidelity $F_{\mathrm{avg}}$ with respect to a gate $\hat{U}_{\mathrm{target}}$ can also be written as a sum of such landscapes. In fact, given a sample of states $\{\ket{\psi_i}\in \mathcal{H}\}_{i=1}^{N_s}$, we have
\begin{multline*}
F_{\mathrm{avg}} := \frac{1}{N_s}\sum_{i=1}^{N_s} |\bra{\psi_i}\hat{U}^{\dagger}_{target}\hat{U}(\boldsymbol{u})\ket{\psi_i}|^2 \\
=\frac{1}{N_s}\sum_{i=1}^{N_s} \bra{\psi_i}\hat{U}^{\dagger}(\boldsymbol{u})\hat{O}_i\hat{U}(\boldsymbol{u})\ket{\psi_i} = \frac{1}{N_s}\sum_{i=1}^{N_s} J_i(\boldsymbol{u}),
\end{multline*} 
where we defined the observables 
\[\hat{O}_i = \hat{U}_{\mathrm{target}} \ket{\psi_i}\bra{\psi_i}\hat{U}_{\mathrm{target}}^{\dagger}.\]
In the main text we will focus on the case of a state fidelity landscape, but we prove most of our results for a generic observable $\hat{O}$ in App.~\ref{appendix_c}.

The main aim of this paper is to study how to represent a quantum dynamical landscape as a linear combination of non-linear functions of the form
\begin{equation} 
J(\boldsymbol{u}) = \smashoperator[lr]{\sum_{i=1}^{N_{\mathrm{weights}}}} w_i \phi_i(\boldsymbol{u}),
\label{eq:general_model}
\end{equation}
where $i$ indexes the basis, specified depending on the choice of the functions $\phi_i(\boldsymbol{u})$.

In the machine learning context these are usually referred to as features, hence the name ``feature map representation". See Fig.~\ref{fig:cl_sur} for a graphical representation of this class of models. Compared to deep neural networks, for these shallow models it is easier to relate their structure both to the physical content of the system generating them and to the properties of their outputs. They can therefore be thought of as intermediate representations (or classical surrogates) of the quantum system, making them an attractive tool to analyze and even compute quantum cost landscapes \cite{Schuld21, Landman22, Schreiber23, Rudolph23, wiedmann24}.

A very natural and well studied choice for the feature map is the so-called Fourier representation, given by  
\[ \phi_{\boldsymbol{\omega}}(\boldsymbol{u}) = e^{-i \boldsymbol{\omega} \cdot \boldsymbol{u}}\]
where the weighting over $\boldsymbol{\omega}$ encodes the frequency spectrum of the landscape. We show in Sec.~\ref{lie_fourier} that in general this spectrum is dense, unlike in the VQA and QML setting where only discrete combinations of the $\omega$ appear, corresponding to the eigenspectrum of the Hamiltonian generators. Naturally, this poses a unique challenge for classical surrogate models.

We also prove that the presence of symmetries in the dynamics constrains these frequencies, determining selection rules that we will describe in detail. We provide in Sec.~\ref{numerics} a numerical example for the transverse field Ising model.

Another important representation is the polynomial feature map given by 
\[ \phi_{\boldsymbol{p}}(\boldsymbol{u})=\prod_{\nu=1}^N  ({u}_\nu - {u}^{(0)}_\nu)^{p_{\nu}}\]
with $\boldsymbol{p} \in \mathbb{N}^N$, $\boldsymbol{u}^{(0)}$ a reference pulse, and where the weights $w_{\boldsymbol{p}}$ can be computed from the derivatives by Taylor expansion, as shown in Sec.~\ref{taylor}. We bound these derivatives using the Fourier representation, establishing the Lipschitz continuity of the landscape and thereby providing an estimate of the error in this expansion. This enables a low-order $\mathcal{O}(poly(N))$ truncation under certain conditions relating to constraining the total time and energy of the external fields. 

The third representation is given by bandwidth-limited kernel function 
\[ \phi_{i}(\boldsymbol{u}) = \prod_{\nu=1}^{N} \frac{\sin{[ \omega_{\mathrm{max}}\delta t (u_{\nu} - u^{(i)}_{\nu})}]}{ \omega_{\mathrm{max}}\delta t(u_{\nu} - u^{(i)}_{\nu})},\]
which captures correlations between different sampled data points $\mathcal{D}_{\mathrm{train}} = \{ (\boldsymbol{u}^{(i)}, J(\boldsymbol{u}^{(i)})) \}_{i=1}^{N_{\mathrm{train}}}$. This regression model is derived in Sec.~\ref{kernel} by essentially integrating over the infinite Fourier spectrum, enhancing the tractability of the learning task. We further discuss efficiently training the model via RIDGE regression \cite{Bishop06, Landman22}.

Finally, we conclude by exploring the relevance of our findings in the context of landscape optimization.
In order to quantify the difficulties inherent to this problem, several landscape measures have been proposed in the literature (ruggedness, trap density, gradient variance and barren plateaux, etc.) \cite{Dalgaard22, McClean18, Holmes22}, each one highlighting a different way in which the optimization problem can be considered hard (or not). In Sec.~\ref{metrics}, we relate a selection of these measures to the representations that we have studied, establishing bounds for them whenever possible.
We also investigate the relevance of our findings in the tuning and design of optimization algorithms in Sec.~\ref{applications}.
\section{Lie-Fourier representation}\label{lie_fourier}
In order to derive a Fourier representation of the landscape $J$, we expand each one of the time-step unitaries in Eq.~\eqref{eq:timeordering} by means of the Lie-Trotter product formula~\cite{ReedBarry72, Trotter59, Lloyd96}
\begin{multline}
    \hat{U}(u) = e^{-i\delta t \hat{H}(u)} \\
    = \lim_{n \xrightarrow[]{} \infty} (e^{-\frac{i\delta t}{n}\hat{H}_d} e^{-\frac{i\delta t}{n}u \hat{H}_c })^n = \lim_{n \xrightarrow[]{} \infty} \hat{U}_n(u)
    \label{eq:lie_single}
\end{multline}
which converges to its limit, with an error of order $\mathcal{O}(n^{-1})$ that depends on  $\delta t u_{\mathrm{max}}||\hat{H}||_{\infty}$. 
By studying this expression for a generic $n$ we can find a sequence of approximations $J_n$ of the landscape $J$, whose convergence properties are discussed in detail in App.~\ref{appendix_c}. On the other hand, finite order expansions are also interesting on their own, and for example for $n=1$ we recover an interleaved circuit which is a common Ansatz for VQA \cite{Kokail19}.
More generally, for any $n$ and even for multiple controls (see App.~\ref{appendix_c}), the resulting circuits will be Periodic Structure Ans{\"a}tze, as introduced in \cite{Larocca22}, which consist of stacked copies of a predefined gate sequence.
In such a case, we can then apply a similar procedure to \cite{Schuld21} and obtain an expression of Eq.~\eqref{eq:lie_single} as a possibly infinite sum of complex exponentials. We will refer to the resulting partial Fourier sum as to the Lie-Fourier representation of the landscape. This stresses the dependence of the representation on the usage of the Lie product expansion as a preliminary step, which can be in principle substituted with higher order Suzuki-Trotter products \cite{Suzuki91} or other approximation schemes \cite{Schilling24}, giving rise to different frequencies and coefficients. 
\subsection{Commuting Hamiltonians $[\hat{H}_d,\hat{H}_c]=0$}
Let us first look at a base case of the problem. 
If $[\hat{H}_c,\hat{H}_d]=0$ the Lie-Trotter expansion stops at the first order, and we have
\begin{equation*} 
\hat{U}(\boldsymbol{u}) = e^{-iN\delta t \hat{H}_d} e^{-i\delta t \hat{H}_c \sum_{\nu=1}^N u_{\nu}}.
\end{equation*}
The unitary operator only depends on the ``effective" control $\bar{u} := N^{-1}\sum_{\nu=1}^N u_{\nu}$.
We can write this unitary operator as a finite sum of Fourier components by working in the eigenbasis of $\hat{H}_c = \hat{V}^{\dagger} \hat{\Lambda} \hat{V}$, with $\bra{i}\hat{\Lambda}\ket{j} = \delta_{ij}\lambda_i$:
\begin{eqnarray*}
\hat{U}(\bar{u}) &=& e^{-iN\delta t \hat{H}_d} \hat{V}^{\dagger} e^{-i T \hat{\Lambda} \bar{u}} \hat{V} \\
&=& e^{-iN\delta t \hat{H}_d} \sum_{j=1}^D \hat{V}^{\dagger} \ket{j}e^{-iT \lambda_j \bar{u}} \bra{j}  \hat{V} \\
&=:& \sum_{\omega \in \mathcal{S}} \hat{B}^{\omega}(T)e^{-iT \omega \bar{u}}, \label{eq:commutingU}
\end{eqnarray*}
where we defined the spectrum of the control Hamiltonian $\mathcal{S} = \{\lambda_j \}_{j=1,\dots,D}$. 
We can now write also the transfer fidelity with respect to a target state $\ket{\chi}$ from Eq.~\eqref{eq:fid} as a finite Fourier sum
\begin{multline*} 
J(\bar{u}) = \sum_{\omega, \omega' \in \mathcal{S}} \underbrace{\bra{\psi}  \hat{B}^{\omega \dagger}(T) \ket{\chi}}_{b^*_{\omega}} \underbrace{\bra{\chi} \hat{B}^{\omega'}(T)\ket{\psi}}_{b_{\omega'}} e^{iT(\omega - \omega') \bar{u}} \\=: \sum_{\omega \in \mathcal{S}^{\Delta}} c_{\omega}(T) e^{-i T \omega \bar{u}}, 
\end{multline*}
where we defined the Fourier spectrum of the fidelity as the set of all possible differences between frequencies in the control spectrum
\begin{equation}
\mathcal{S}^{\Delta} = \{ \omega = \omegaEig' - \omegaEig | \omegaEig,\omegaEig' \in \mathcal{S}\} 
\label{eq:spec_fid}
\end{equation}
and the fidelity Fourier coefficients as
\begin{equation*} c_{\omega}(T) = \smashoperator[lr]{\sum_{\omega',\omega'' \in \mathcal{S}}} \delta_{\omega''-\omega', \omega} b^*_{\omega'}(T) b_{\omega''}(T). 
\end{equation*}

\subsection{General case $[\hat{H}_d, \hat{H}_c] \neq 0$}
In general, the expansion in Eq.~\eqref{eq:lie_single} does not stop for any finite $n$,
and we obtain for the cost functional a sequence of partial Fourier sums $J_n(\boldsymbol{u})$ that converges to $J(\boldsymbol{u})$.
Once again, we write the control Hamiltonian in its eigenspectrum $\hat{\Lambda}$ and absorb the change of basis in the terms $\hat{W}$ which do not depend on the control $u$ and are defined implicitly via
\begin{equation*}
\hat{U}_n(u) = \hat{V}^{\dagger}\underbrace{\hat{V}e^{-\frac{i\delta t}{n}\hat{H}_d}\hat{V}^{\dagger}}_{\hat{W}(n^{-1}\delta t)} e^{-\frac{i\delta t}{n}u \hat{\Lambda}}\cdots e^{-\frac{i\delta t}{n}\hat{H}_d}\hat{V}^{\dagger}e^{-\frac{i\delta t}{n}u \hat{\Lambda}} \hat{V}
\end{equation*}
where we multiplied a $\hat{V}^{\dagger}\hat{V} = \hat{I}$ factor on the left. When the dependence on $n$ is not crucial, we omit it to simplify the notation. By expressing the result in terms of the matrix components, we finally obtain the Lie-Fourier representation of the unitary time-step evolution 
\begin{widetext}
\begin{multline}
\bra{i} \hat{U}_{n}(u) \ket{k} = \sum_{j_1 \dots j_{n}l=1}^D e^{-i\delta t\frac{u}{n}(\lambda_{j_1}+\dots+\lambda_{j_n})} V^{\dagger}_{il} W_{l j_{n}}(n^{-1}\delta t) \cdots W_{j_2 j_1}(n^{-1}\delta t) V_{j_1 k} \\
= \sum_{\boldsymbol{j}\in [D]^n} e^{-i\delta t u \omega(\boldsymbol{j})} A^{\boldsymbol{j}}_{ik}(n,\delta t) = \sum_{\omega \in \mathcal{S}_n}  e^{-i\delta t u \omega} B^{\omega}_{ik}(n,\delta t),
\label{eq:timestep_unitary}
\end{multline}
\end{widetext}
where $[D]^n \subset \mathbb{N}^n$ is the set of integer vectors with elements ranging from $1$ to $D$, and the sum over $l$ is absorbed into the definition of the coefficients $A^{\boldsymbol{j}}_{ik}$, since the frequency $\omega$ does not depend on $l$:
\[  A^{\boldsymbol{j}}_{ik}(n,\delta t) = \sum_{l=1}^D V^{\dagger}_{il} W_{l j_{n}}(n^{-1}\delta t) \cdots W_{j_2 j_1}(n^{-1}\delta t) V_{j_1 k}. \]
In the last step all the coefficients $A^{\boldsymbol{j}}_{ik}$ corresponding to the same frequency $\omega=\omega(\boldsymbol{j})$ have been summed up into the matrix with entries $B^{\omega}_{ik}$,
\[ B^{\omega}_{ik}(n,\delta t) = \sum_{\boldsymbol{j}\in \mathcal{J}_{\omega}} A^{\boldsymbol{j}}_{ik}(n,\delta t),\ \mathcal{J}_{\omega} = \{ \boldsymbol{j}\in [D]^n\ |\ \omega(\boldsymbol{j}) = \omega \} \]
which is only defined for values of $\omega$ belonging to the discrete Fourier spectrum $\mathcal{S}_n$ 
\[ \mathcal{S}_n = \{ \omega(\boldsymbol{j}) = \frac{1}{n}(\lambda_{j_1}+\dots+\lambda_{j_n})\  |\  \boldsymbol{j}\in [D]^n \}. \]
Now by layering the $N$ time-step unitaries we can build up the full unitary $\hat{U}_n(\boldsymbol{u})$:
\begin{multline}
\hat{U}_n(\boldsymbol{u}) = \hat{U}_{n}(u_N) \cdots \hat{U}_{n}(u_1) \\
= \sum_{\boldsymbol{j}_{1} \cdots \boldsymbol{j}_{N}} e^{-i\delta t \sum_{ \nu} u_{\nu}\omega(\boldsymbol{j}_{\nu})} \hat{A}^{\boldsymbol{j}_{N}} \cdots \hat{A}^{\boldsymbol{j}_{1}}\\
=  \sum_{{\omega}_{1} \dots {\omega}_{N}}  e^{-i\delta t \sum_{\nu} u_{\nu}\omega_{\nu}} \hat{B}^{{\omega}_{N}}\cdots  \hat{B}^{{\omega}_{1}} \\=  \sum_{\boldsymbol{\omega}\in \mathcal{S}_n^N}  e^{-i\delta t \boldsymbol{\omega} \cdot \boldsymbol{u}} \hat{B}^{\boldsymbol{\omega}}(n,\delta t),
\end{multline}
where we defined the frequency vector $\boldsymbol{\omega}\in \mathcal{S}^N_n$ and the operator $\hat{B}^{\boldsymbol{\omega}}$ as the product
\[
    \hat{B}^{\boldsymbol{\omega}} = \hat{B}^{{\omega}_{N}}\cdots  \hat{B}^{{\omega}_{1}}.
\]
By plugging this expression for the unitary operator back into Eq.~\eqref{eq:fid} we finally find the Lie-Fourier representation of the fidelity:
\begin{multline}
    J_n(\boldsymbol{u}) = \bra{\psi} \hat{U}_n^{\dagger}(\boldsymbol{u})\ket{\chi}\bra{\chi} \hat{U}_n(\boldsymbol{u}) \ket{\psi} \\
    = \sum_{\boldsymbol{\omega'},\boldsymbol{\omega''}\in \mathcal{S}^N_n} \underbrace{\bra{\psi} \hat{B}^{\boldsymbol{\omega''} \dagger} \ket{\chi}}_{b^*_{\boldsymbol{\omega''}}} \underbrace{\bra{\chi}   \hat{B}^{\boldsymbol{\omega'}} \ket{\psi}}_{b_{\boldsymbol{\omega'}}} e^{-i\delta t \boldsymbol{(\omega'-\omega'')} \cdot \boldsymbol{u}} \\ =  \sum_{\boldsymbol{\omega} \in (\mathcal{S}_n^{\Delta})^N} c_{\boldsymbol{\omega}}(n,\delta t)  e^{-i\delta t \boldsymbol{\omega} \cdot \boldsymbol{u}} 
    \label{eq:fid_diff}
\end{multline}
where $(\mathcal{S}_n^{\Delta})^N$ is the set of frequency differences within the Fourier spectrum of the unitary. The coefficients of the expansion $c_{\boldsymbol{\omega}}$ are given by
\begin{equation}
c_{\boldsymbol{\omega}} = \sum_{\boldsymbol{\omega'},\boldsymbol{\omega''}\in \mathcal{S}_n^N} \delta_{\boldsymbol{\omega}, \boldsymbol{\omega'}-\boldsymbol{\omega''}} b_{\boldsymbol{\omega''}}^*  b_{\boldsymbol{\omega'}}
\label{eq:coeff_fid}
\end{equation}
We conclude this derivation by pointing out that there is a useful trick we can use since in the main text we are working with just one control per time step.
In fact, we can equivalently write the fidelity as
\begin{equation}
J_n(\boldsymbol{u}) = \bra{\tilde{\psi}} \hat{\tilde{U}}_n^{\dagger}(\boldsymbol{u})\ket{\tilde{\chi}}\bra{\tilde{\chi}} \hat{\tilde{U}}_n(\boldsymbol{u}) \ket{\tilde{\psi}}
\end{equation}
by transforming the dynamics to a new unitary frame through $\hat{V}$, so that the initial state becomes $\ket{\tilde{\psi}} = \hat{V}\ket{\psi}$ and the propagator $\hat{\tilde{U}} = \hat{V} \hat{U} \hat{V}^{\dagger}$.
In this new frame, the coefficients $\tilde{b}_{\omega}=b_{\omega}$ stay the same (they are scalars), while $ \tilde{A}_{ik}^{\boldsymbol{j}}$ (and consequently $ {\tilde{B}}_{ik}^{\omega}$) take a simpler form  
\[  \tilde{A}^{\boldsymbol{j}}_{ik}(n,\delta t) = W_{i j_{n}}(n^{-1}\delta t) \cdots W_{j_2 j_1}(n^{-1}\delta t) \delta_{j_1 k}. \]

\subsection{Basic properties of the Lie-Fourier representation}
Let us now prove some properties of the Lie-Fourier representation. 
It is easy to see from Eq.~\eqref{eq:spec_fid}-\eqref{eq:coeff_fid} that, compatibly with $J$ and $J_n$ being real functions, both the fidelity spectrum $\mathcal{S}^{\Delta}_n$ and the coefficients $c_{\boldsymbol{\omega}}$ are symmetric in the following way
\begin{align*}
    \omega \in \mathcal{S}^{\Delta}_n \implies& -\omega \in \mathcal{S}^{\Delta}_n, \\
    c_{-\boldsymbol{\omega}} = c^*_{\boldsymbol{\omega}}&
\end{align*}
Another straightforward result is that all the approximating functions $J_n$ are bandwidth limited 
\begin{lemma}[Bandwidth limitation]
    $\mathcal{S}_n \in [\lambda_{\mathrm{min}},\lambda_{\mathrm{max}}]$ with $\lambda_{\mathrm{min}} (\lambda_{\mathrm{max}})$ the minimum (maximum) eigenvalue of $\hat{H}_c$.
    \label{lemma:bandwidth}
\end{lemma}
\begin{proof}
    \[
    \forall n\ \ \min_{\boldsymbol{j}\in [D]^n} \omega(\boldsymbol{j}) = \min_{\boldsymbol{j}\in [D]^n} \sum_{i=1}^n\frac{\lambda_{j_i}}{n} = \sum_{i=1}^n \min_{j \in [D]} \frac{\lambda_{j}}{n} = \lambda_{\mathrm{min}}, 
    \]
    and the same is true for the $\max$, which implies $\lambda_{\mathrm{min}} \leq \omega \leq \lambda_{\mathrm{max}}$.
\end{proof}
Moreover, let us consider the subset of frequencies $\mathcal{S}_n' \subset \mathcal{S}_n$ obtained by constraining the choice of eigenvalues to only $\lambda_{\mathrm{min}}, \lambda_{\mathrm{max}}$:
\[ \mathcal{S}_n'=\{ \omega = \frac{m}{n} \lambda_{\mathrm{max}} + \frac{n-m}{n} \lambda_{\mathrm{min}},\ m=0,\dots,n \}.\]
It easy to see that this set is a regular grid over $[\lambda_{\mathrm{min}},\lambda_{\mathrm{max}}]$ with a step of $n^{-1}(\lambda_{\mathrm{max}} - \lambda_{\mathrm{min}})$, which is also larger than the maximum distance between any point in $[\lambda_{\mathrm{min}},\lambda_{\mathrm{max}}]$ and $\mathcal{S}_n'$. But then as $n\to \infty$ the frequencies in $\mathcal{S}_n'$ (and therefore in $\mathcal{S}_n$) will fill up that interval densely.

Similar properties hold also for the fidelity spectrum  $\mathcal{S}^{\Delta}_n \subset [-\omega_{\mathrm{max}},\omega_{\mathrm{max}}]$ where we defined $\omega_{\mathrm{max}} = \lambda_{\mathrm{max}} - \lambda_{\mathrm{min}}$.
Physically, this means that the landscape displays a typical scale $\sim \omega^{-1}_{\mathrm{max}}$ in control space below which there are no new details within a given tolerance. This information is relevant, for instance, in the context of optimization, as can potentially be used to set step sizes, filter out noise below that length scale and avoid oversampling. We will explore this idea more in detail later in Sec.~\ref{metrics}.

To better understand the properties of the Fourier representation in the $n\xrightarrow[]{} \infty$ limit, we first state two results concerning the boundedness of the coefficients:
\begin{lemma}[$L_1$ Boundedness of the coefficients]
    \[ \exists r \in \mathbb{R},\ s.t.\ \   \forall n,\  \sum_{\boldsymbol{\omega} \in (\mathcal{S}^{\Delta}_n)^N} |c_{\boldsymbol{\omega}}(n,\delta t)| \leq r \]
    \label{lemma:coefficients_l1}
\end{lemma}
\begin{proof}
See App.~\ref{appendix_d}.
\end{proof}
\begin{lemma}[$L_2$ Boundedness of the coefficients]
    \[ \forall n,\  \sum_{\boldsymbol{\omega} \in (\mathcal{S}^{\Delta}_n)^N} |c_{\boldsymbol{\omega}}(n,\delta t)|^2 \leq 1 \]
    \label{lemma:coefficients_l2}
\end{lemma}
\begin{proof}
See App.~\ref{appendix_c} for a proof valid for multiple controls and a generic observable $\hat{O}$.
\end{proof}
Due to the exponential behaviour of $r$ with respect to the dimensionality of the Hilbert space $D$ (which is itself exponential in the number of subsystems) and the number of timesteps $N$, the $L_1$ bound is unlikely to be helpful in numerical predictions in general, given the results of this analysis. To obtain useful results concerning quantities that are linear in the coefficients (such as the landscape derivatives in Sec.~\ref{derivatives}), we will instead directly invoke the boundedness of $J$. Conversely, the $L_2$ bound is stronger because it is constant and, as we will see in Sec.~\ref{metrics}, it is related to the landscape variance, as it is quadratic in the coefficients.

These bounds also suggest that the coefficients are well behaved in the large $n$ limit. Following this intuition, we can show that the coefficients with frequencies corresponding to the eigenvalues of the control Hamiltonian ---that is, matching the trivial commuting case \eqref{eq:spec_fid}--- contain in general a finite contribution, so we will refer to them as ``resonant".
Note that these will correspond to resonances in the landscape, i.e.~generalized Rabi oscillations in the system output, which in general differ from the resonances in the underlying physical system because of the presence of the drift Hamiltonian.

If we fix $\boldsymbol{j}=(p,\dots,p)$ then clearly $\omega(\boldsymbol{j}) = \lambda_p$ and the corresponding contribution to the propagator coefficient
$B^{\lambda_p}_{ik}$ is therefore 
\begin{eqnarray*} 
A^{(p,\dots,p)}_{ik}(n,\delta t)&=&\lim_{n \xrightarrow[]{} \infty} \sum_{l=1}^d V^{\dagger}_{il} W_{l p} \cdots W_{p p} V_{p k} \\ &=&  V^{\dagger}_{ip} \lim_{n \xrightarrow[]{} \infty} (W_{p p})^n V_{p k},
\end{eqnarray*}
where the only term in the sum which is non-zero in the limit is $l=p$. In fact, all the other terms contain an off-diagonal element of $\hat{W}$, which
is vanishingly small for large $n$ (similarly to App.~\ref{appendix_d}). Moreover, we have the first-order expansion
\[ W_{ij} = [\hat{V} e^{-\frac{i\delta t}{n}\hat{H}_d} \hat{V}^{\dagger}]_{ij} =  \delta_{ij} - i \frac{\delta t}{n} [\hat{V} \hat{H}_d \hat{V}^{\dagger}]_{ij} + o\left( \frac{\delta t}{n} \right). \]
Thus, we can exploit the continuity of the logarithm to write
\begin{multline*}
\lim_{n \xrightarrow[]{} \infty} (W_{p p}(n^{-1}\delta t))^n = \exp \lim_{n \xrightarrow[]{} \infty} n \log (W_{p p}(n^{-1}\delta t)) \\ = \exp \lim_{n \xrightarrow[]{} \infty} n (- i \frac{\delta t}{n} [\hat{V} \hat{H}_d \hat{V}^{\dagger}]_{pp} + o\left( \frac{\delta t}{n} \right))\\
= e^{ - i \delta t [\hat{V} \hat{H}_d \hat{V}^{\dagger}]_{pp}},
\end{multline*} 
which gives us, as anticipated, a finite result for a resonant landscape frequency. 
Indeed, in the next section we will see, using numerics, that the fidelity coefficients corresponding to differences in eigenvalues are typically finite, unless discrete symmetries decide otherwise, while all the other coefficients converge to 0 as $n \to \infty$, often creating a mostly continuous spectrum in that limit.
We conclude by stating a result concerning the role of dynamical symmetries, which determine which frequencies can be non-zero:
\begin{lemma}[Symmetries and selection rules]
Let $\hat{\Gamma}$ be a symmetry of the system, that is $[\hat{\Gamma},\hat{H}_c]=[\hat{\Gamma},\hat{H}_d]=0$. Let $\ket{\gamma_{g(i)},\lambda_i}$ be the simultaneous eigenstates of $\hat{\Gamma}$ and $\hat{H}_c$ (the respective eigenvalues being $\gamma_{g}$ and $\lambda_i$).
Let $\mathcal{I}_g$ be the set of indices belonging to the $g$-th symmetry sector \[ \mathcal{I}_g = \{ i\in\{1,\dots,D\}\ |\ g(i)=g\}, \]
so that the projector $\hat{P}_g$ onto that sector is
\[ \hat{P}_g = \sum_{i \in \mathcal{I}_g} \ket{\gamma_{g},\lambda_i} \bra{\gamma_{g},\lambda_i},\] 
and let $\mathcal{G}$ be the set of sector indices with corresponding non-zero overlap on both the initial and target states
\[  \mathcal{G} = \{ g\in\{1,\dots,G\}\ |\ \hat{P}_g\ket{\psi} \neq 0 \wedge  \hat{P}_g\ket{\chi} \neq 0 \}. \]
Then, the coefficients $b_{\boldsymbol{\omega}}$ of the cost functional are zero unless $\boldsymbol{\omega} \in \ce{^{(\Gamma,N)}\mathcal{S}_n}$, with
\[\ce{^{(\Gamma,N)}\mathcal{S}_n} = \bigcup_{g \in \mathcal{G}} (\mathcal{S}^{(g)}_n)^N,\]
and $\mathcal{S}^{(g)}_n$ the set of frequencies built from the $g$-th symmetry sector
\[ \mathcal{S}^{(g)}_n = \{ \omega(\boldsymbol{j}) = \frac{1}{n}(\lambda_{j_1}+\dots+\lambda_{j_n})\  |\  \boldsymbol{j}\in \mathcal{I}_g^n \}. \]
The coefficients $c_{\boldsymbol{\omega}}$ are zero unless $\boldsymbol{\omega} \in \ce{^{(\Gamma,N)}\mathcal{S}^\Delta_n}$, with
\[ \ce{^{(\Gamma,N)}\mathcal{S}^\Delta_n} = \{ \boldsymbol{\omega} = \boldsymbol{\omega'} - \boldsymbol{\omega''} | \boldsymbol{\omega'},\boldsymbol{\omega''} \in \ce{^{(\Gamma,N)}\mathcal{S}_n} \} .\]
\label{lemma:symmetries}
\end{lemma}
\begin{proof}
See App.~\ref{appendix_d}.
\end{proof}
In essence, the lemma establishes that the eigenvalues that can appear in the Lie-Fourier representation of the cost landscape can only originate from the same symmetry sector at every time step, where the relevant sectors must contain a projection of the initial and final states.

This result relies on the decomposition of the dynamics into invariant subspaces, which causes the unitaries to have a block diagonal matrix representation. Another known consequence is that the dynamical Lie algebra of the generators factorizes, simplifying the dynamics and putting constraints on the appearance of barren plateaux and on classical non-computability of the dynamics \cite{kazi24, Larocca22}. In light of this, and given the relation between Fourier spectrum and landscape variance that we will explore in Sec.~\ref{metrics}, we interpret the frequency selection rules as one of the manifestations of this algebraic phenomenon in the context of the Fourier representation of the landscape. This provides an alternative viewpoint on the subject and informs the further development of classical simulation algorithms \cite{Rudolph23}.

\subsection{Numerical examples}\label{numerics}
We will now pick a specific system and see how the theory detailed earlier in this section comes into play.
The system under consideration is the one dimensional transverse field Ising model
\begin{equation}
    \hat{H}[u(t)] = \sum_{i=1}^Q \alpha_d( \hat{\sigma}^{(i)}_z\hat{\sigma}^{(i+1)}_z + h_z \hat{\sigma}^{(i)}_z ) + u(t)\hat{\sigma}^{(i)}_x
    \label{eq:ising}
\end{equation}
where the Hilbert space $\mathcal{H}$, with $D:=\dim{\mathcal{H}}=2^Q$ is the product of $Q$ qubits and $\hat{\sigma}^{(i)}_a$ is the Pauli operator for qubit $i$.
We use the transverse field along the $x$ direction as control, while the constant $\alpha_d$ plays the role of drift Hamiltonian strength. The parameter $h_z$ can be used to turn on the longitudinal field, and we fix it to zero unless specified otherwise.
We assume periodic boundary conditions, so that as far as the qubit indices are concerned $Q+1 \equiv 1$.
The Ising model is a paradigmatic example of quantum dynamics \cite{Kormos17}, and its state transfer fidelity landscape has already been object of study \cite{Bukov19, Dalgaard22}. Moreover, circuit Ans{\"a}tze featuring the generators in Eq.~\eqref{eq:ising} (or variations thereof) are often used within variational algorithms such as QAOA \cite{farhi14} and their dynamical symmetries have been studied in depth $\cite{Larocca22, kazi24, DAlessandro21}$.

\subsubsection{Ising model spectrum}
As we will see now, the Ising model also gives rise to a remarkably simple Lie-Fourier representation. In fact, as we see from the eigenstates of $\hat{H}_c$, namely
\begin{equation}
     \sum_{i=1} ^Q \hat{\sigma}^{(i)}_x \ket{-^{b_1} \cdots -^{b_Q}} = \sum_{i=1}^Q (-1)^{b_i} \ket{-^{b_1} \cdots -^{b_Q}},
     \label{eq:diag_Hc}
\end{equation}
where we adopt the convention $\ket{-^0}=\ket{+}, \ket{-^1}=\ket{-}$,
the spectrum $\mathcal{S}$ of the control Hamiltonian 
\[ \mathcal{S}=\{-Q, -Q+2, \dots ,Q-2, Q \} \]
is highly degenerate (the number of distinct eigenvalues increases linearly instead of exponentially in $Q$, as in the non-degenerate scenario). In turn this gives rise to the following Fourier spectra for the single-timestep unitary and fidelity
\begin{align*}
    \mathcal{S}_n &= \{-Q, -Q+\frac{2}{n}, \dots, Q-\frac{2}{n}, Q \}, \\
    \mathcal{S}^{\Delta}_n &= \{-2Q, -2Q+\frac{2}{n}, \dots, 2Q-\frac{2}{n}, 2Q \} \\
    &= \{ \omega = \omega_{\mathrm{max}}\frac{k}{k_{\mathrm{max}}}\ |\ k = -k_{\mathrm{max}}, \dots, k_{\mathrm{max}} \}
\end{align*}
with $k_{\mathrm{max}}=Qn$, so $\omega_{\mathrm{max}}= 2Q$, $\#\mathcal{S}_n = k_{\mathrm{max}} + 1$, $\#\mathcal{S}^{\Delta}_n = 2k_{\mathrm{max}}+1 =: n_{\Delta}$. We notice that for this model all the frequencies in $\mathcal{S}^{\Delta}_n$ are equally spaced, which allows us to compute numerically the $c_{\boldsymbol{\omega}}$ coefficients by means of the Discrete Fourier Transform (DFT) (refer to App.~\ref{appendix_e} for details).

By applying this numerical algorithm, we study the Lie-Fourier representation of several state transfer problems for the Ising model. More specifically, we consider the transfer between $\ket{{0}_Q}\mapsto \ket{{1}_Q}$ (also for non-zero longitudinal field $h_z \neq 0$), between eigenstates of the control Hamiltonian $\ket{{+}_Q}\mapsto \ket{{-}_Q}$ and between random states $\ket{r} \mapsto \ket{r'}$ \footnote{The random states $\ket{r},\ket{r'}$ are obtained by sampling independently real and imaginary part of each coefficient in the computational basis, using a uniform distribution. Multiplication by an overall factor ensures correct normalization.}. In Fig.~\ref{fig:DFT} we can see plotted the coefficients of the Lie-Fourier expansion for $N=1,2$ for each case.
\begin{figure} 
    \centering
    \includegraphics[width=0.48\textwidth]{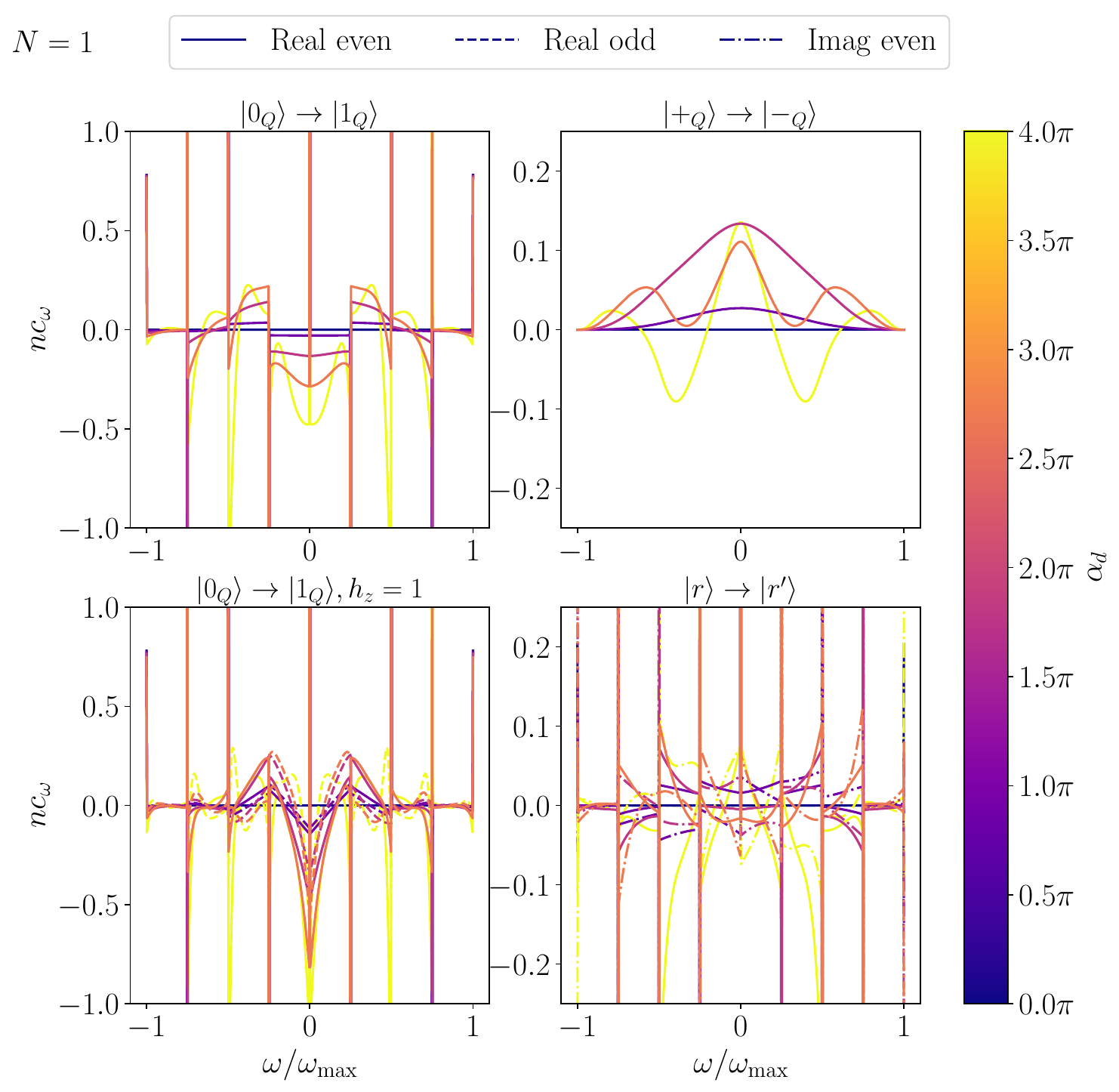}
    \includegraphics[width=0.48\textwidth]{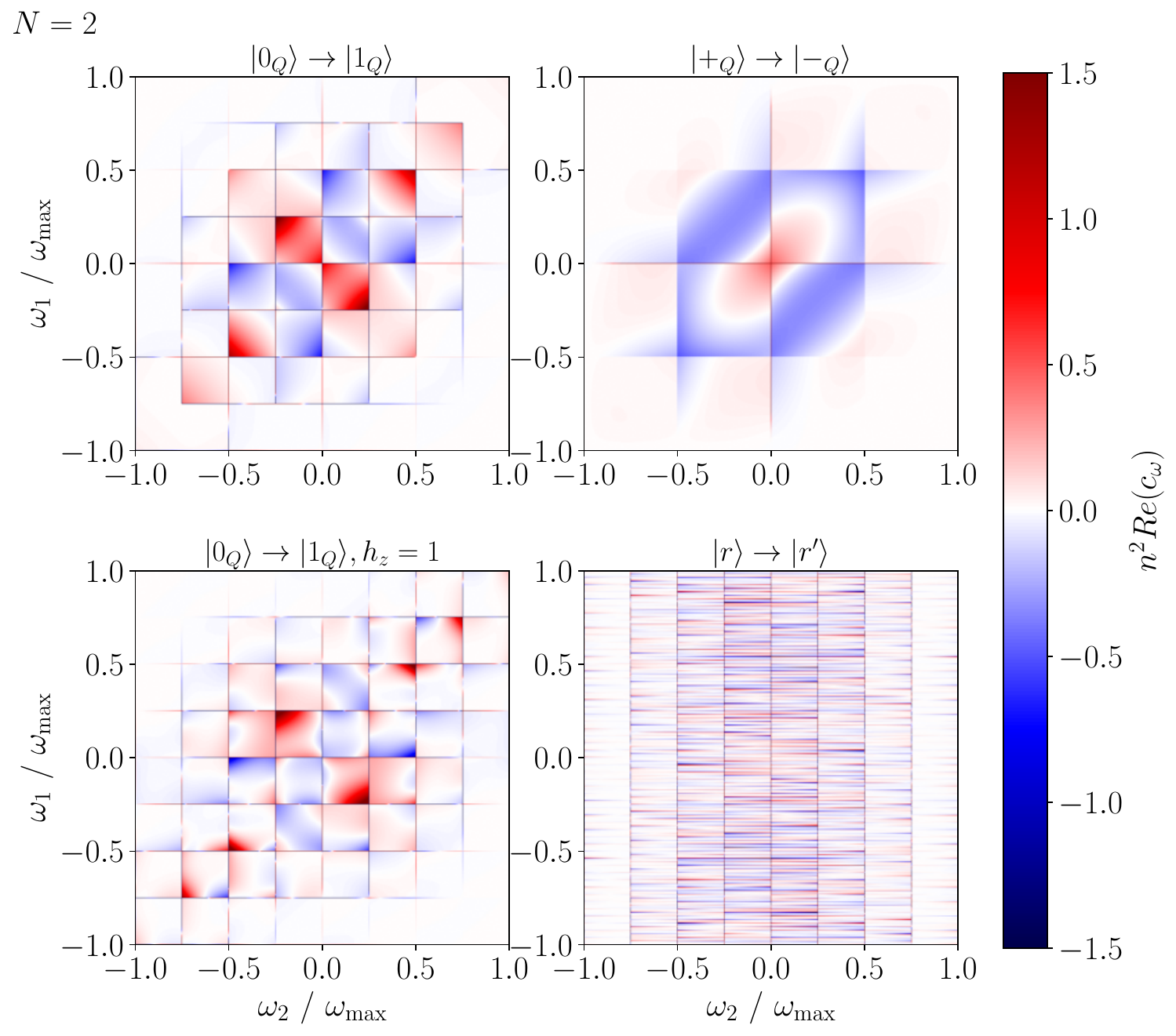}
    \caption{Lie-Fourier representation coefficients of the Ising model dynamical landscape for selected state transfer problems. The results are computed for $n=200$ using the DFT algorithm described in App.~\ref{appendix_e}. (Upper panel) The $c_{\omega}$ coefficients for a single timestep landscape $N=1$ are plotted for several values of drift strength $\alpha_d \pi^{-1} = 0.0, 0.89, 1.78, 2.67, 4.0$ (color scale). The line styles correspond to the real part of the even (solid) and odd (dashed) frequency index branches and to the imaginary part of the even branch (dot-dashed). The imaginary part of the odd branch is numerically zero in all cases. (Lower panel) The real parts of the $c_{\omega}$ coefficients for a double timestep landscape $N=2$ are plotted as a function of the vector frequencies $\boldsymbol{\omega}$ for $\alpha_d \pi^{-1} = 5.59$. In all cases but for random states $\ket{r} \to \ket{r'}$ the coefficients can be reorganized in two branches which exhibit continuous behaviour inside intervals defined by the resonant frequencies.}
    \label{fig:DFT}
\end{figure}

Let us start by looking at the case $\ket{{0}_Q}\mapsto \ket{{1}_Q}$.
As expected from the discussion in the previous subsection, we see that the coefficients corresponding to differences in eigenvalues correspond to finite contributions, which appear as sharp peaks, since we plotted $n^N c_{\omega}$. This plotting choice was adopted to emphasize the behaviour of the non-resonant coefficients, which we did not discuss extensively. 
We observe that numerically the coefficients give rise to an apparently smooth function within the intervals delimited by the resonant frequencies. This suggests that, at least in this particular case, the non-resonant part of the Lie-Fourier expansion approximates the Fourier Transform of a well behaved function. 
\subsubsection{Missing resonances}
If now instead we consider the state transfer problem $\ket{{+}_Q}\mapsto \ket{{-}_Q}$, we see that some of the frequencies that were resonant in the previous case do not display a sharp peak anymore.
The reason behind this lies in the discrete symmetries of the Hamiltonian in Eq.~\eqref{eq:ising}. Let us define the following operator(s): 
\begin{equation*}
    \hat{\Sigma}_{\alpha} = \bigotimes_{i=1}^Q \hat{\sigma}^{(i)}_{\alpha},\ \ \alpha=x,y,z,
\end{equation*}
which are Hermitian $\hat{\Sigma}_{\alpha}^{\dagger} = \hat{\Sigma}_{\alpha}$ and satisfy $\hat{\Sigma}_{\alpha}^2 = \hat{I}$.
One can easily check that $\hat{\Sigma}_x$ is a symmetry of both generators $[\hat{H}_d,\hat{\Sigma}_x]=[\hat{H}_c,\hat{\Sigma}_x]=0$ with eigenvalues $\pm 1$ and eigenvectors given by
\begin{equation*}
    \hat{\Sigma}_x \ket{-^{b_1} \cdots -^{b_Q}} = (-1)^{\sum_{i=1}^Q b_i} \ket{-^{b_1} \cdots -^{b_Q}}.
\end{equation*}
Comparing this expression with Eq.~\eqref{eq:diag_Hc}, we can categorize the eigenvalues in $\mathcal{S}$ according to the symmetry sector they belong to:
\begin{align*}
    \mathcal{S}^+&= 
    \begin{cases}
    \{Q, Q-4, \dots, -Q+4, -Q\}\ \mathrm{if}\ Q\ \mathrm{even}\\
    \{Q, Q-4, \dots, -Q+6, -Q+2\}\ \mathrm{if}\ Q\ \mathrm{odd} \\
    \end{cases}\\
    \mathcal{S}^- &=
    \begin{cases}
    \{Q-2, Q-6, \dots, -Q+6, -Q+2\}\ \mathrm{if}\ Q\ \mathrm{even}\\
    \{Q-2, Q-6, \dots, -Q+4, -Q\}\ \mathrm{if}\ Q\ \mathrm{odd}\\
    \end{cases}
\end{align*}
with $\mathcal{S} = \mathcal{S}^+ \cup \, \mathcal{S}^-$.
Moreover, the initial and target states lie inside a single symmetry sector:
\begin{align*}
    \hat{\Sigma}_x \ket{{+}_Q} &= \ket{{+}_Q} ,\\
    \hat{\Sigma}_x \ket{{-}_Q} &= (-1)^{Q}\ket{{-}_Q}.
\end{align*}
Then, because of Lemma \ref{lemma:symmetries}, if $Q$ is odd all
coefficients $b_{\boldsymbol{\omega}}$ and, therefore $c_{\boldsymbol{\omega}}$, are zero.
If instead $Q$ is even, the spectrum will be limited to the eigenvalues of the $+1$ symmetry sector, that is 
\[
    \mathcal{S}^+_n = \{Q, Q-\frac{4}{n}, \dots, -Q+\frac{4}{n}, -Q \}
\]
which gives rise to the fidelity spectrum
\[
    \mathcal{S}^{+ \Delta}_n = \{2Q, 2Q-\frac{4}{n}, \dots, -2Q+\frac{4}{n}, -2Q \}.
\]
This excludes from the spectrum roughly half of the frequencies, among which the resonant frequencies that do not appear in the $\ket{{+}_Q}\mapsto \ket{{-}_Q}$ case in Fig.~\ref{fig:DFT}.
Since this symmetry is broken in all the other cases, the resonances are in general not suppressed and can appear as sharp peaks.
\subsubsection{Additional spectrum features}
It is interesting to see that, in all cases except the transfer between random states, the coefficients $c_{\boldsymbol{\omega}} \in \mathbb{C}$ are real. As far as the $\ket{{0}_Q}\mapsto \ket{{1}_Q}$ case is concerned (also for $h_z \neq 0$), this fact can be explained by the existence of additional anti-symmetries of the Ising model. These are given by the operators $\hat{\Sigma}_{y,z}$ and have the property $\hat{\Sigma}_{y,z}\hat{H}(u)\hat{\Sigma}_{y,z} = \hat{H}(-u)$. But then we have the following:
\begin{multline*}
    J(\boldsymbol{u}) = |\bra{{1}_Q} \hat{U}(\boldsymbol{u}) \ket{{0}_Q}|^2 \\
    = |\bra{{1}_Q} \hat{\Sigma}_{z}^2 \hat{U}(u_N) \hat{\Sigma}_{z}^2 \cdots \hat{\Sigma}_{z}^2 \hat{U}(u_1)\hat{\Sigma}_{z}^2 \ket{{0}_Q}|^2 \\
    = |(-1)^Q \bra{{1}_Q}  \hat{U}(-u_N)  \cdots \hat{U}(-u_1) \ket{{0}_Q}|^2 = J(-\boldsymbol{u})
\end{multline*}
which together with $J \in \mathbb{R}$ implies $c_{\boldsymbol{\omega}} \in \mathbb{R}$. A similar reasoning using $\hat{\Sigma}_{x}$ shows another interesting property:
\begin{multline*}
    J(u_1, \dots, u_N) = |\bra{{1}_Q} \hat{U}(\boldsymbol{u}) \ket{{0}_Q}|^2 \\
    = |\bra{{1}_Q} \hat{\Sigma}_{x}^2 \hat{U}(u_N) \hat{\Sigma}_{x}^2 \cdots \hat{\Sigma}_{x}^2 \hat{U}(u_1)\hat{\Sigma}_{x}^2 \ket{{0}_Q}|^2 \\
    = |\bra{{0}_Q} \hat{U}(u_N)  \cdots \hat{U}(u_1) \ket{{1}_Q}|^2 \\
    = |\bra{{1}_Q} \hat{U}^{\dagger}(u_1)  \cdots \hat{U}^{\dagger}(u_N) \ket{{0}_Q}|^2 = J(u_N, \dots, u_1),
\end{multline*}
where the last step is justified since both $\hat{H}$ and the states $\ket{{0}_Q}, \ket{{1}_Q}$ only have real matrix elements and overlaps with the computational basis (see App.~\ref{appendix_e} for more details). This property shows that in this case the Lie-Fourier coefficients are approximately symmetric under reversal of the frequency order $c_{\omega_1 \cdots \omega_N} = c_{\omega_N \cdots \omega_1}$. The reason why the symmetry is not exact is that in the Lie-Fourier representation the unitary $\hat{U}$ is substituted with $\hat{U}_n$, which obeys the symmetry exactly only in the $n \to \infty$ limit.

Another phenomenon that we observe is the splitting of the coefficients into branches featuring apparently continuous behaviour inside the intervals described above. In all cases without longitudinal field $h_z =0$, the branch corresponding to frequencies in $\mathcal{S}^{\Delta}_n$ with odd indices are zero \footnote{Here we are thinking about $\mathcal{S}^{\Delta}_n$ as a vector of frequencies in ascending order, numbered from $0$ to $n_{\Delta}-1$. }. Moreover, the even/odd branches appear to be the only ones present in all cases except the random state transfer. In this last case, the structure looks much more complex, potentially featuring many branches and more discontinuities. We believe this to be related to the fact that randomly chosen states break not only the symmetries we discussed, but also the ones associated to cyclic qubit permutations \cite{DAlessandro21}.

The observed stepwise-continuous behaviour seems to suggest that at least in symmetric cases it would be advantageous to choose as feature maps a basis of e.g. polynomial functions in frequency space. This way large swaths of the continuous spectrum could be represented by a few basis functions, resulting in a more efficient representation. Even the discrete symmetries we highlighted could be explicitly encoded within such a feature map. While leaving this for future work, we notice that we can achieve a similar effect with the sinc-kernel representation as presented in Sec.~\ref{kernel}.

\section{Landscape derivatives} \label{derivatives}
The Lie-Fourier representation we discussed in Section \ref{lie_fourier} also allows us to prove some results concerning the derivatives of the landscape, which are independent of the representation itself. In fact we have that
\begin{lemma}[Boundedness of the derivatives]
The partial derivatives of $J(\boldsymbol{u})$ of any order $P$ are bounded by
\[ |\partial^{p_{1}}_{1}\cdots\partial^{p_{N}}_{N} J(\boldsymbol{u})|   \leq  \frac{(\omega_{\mathrm{max}} \delta t)^P}{2}\]
where $P = \sum_{\nu} p_{\nu} $  and $\omega_{\mathrm{max}} = |\lambda_{\mathrm{max}}-\lambda_{min}|$ is the maximum transition frequency in the control Hamiltonian.
\label{lemma:derivatives}
\end{lemma}
\begin{proof}
See App.~\ref{appendix_c} for a proof in the case of multiple controls and a generic observable $\hat{O}$.
\end{proof}
The proof first shows that the result holds for the Lie-Fourier approximants $J_n$, making use of the fact that they are bandwidth limited and bounded $0 \leq J_n \leq 1$. The result for the landscape $J$ is obtained by invoking uniform convergence of the functions $J_n$ and of their derivatives of any order.

As a direct consequence we notice that the $L_1$ norm of the gradient, Hessian and higher order derivative tensors are bounded by constants:
\[ \sum_{\boldsymbol{p} \in \mathbb{\mathbb{N}}^{N}} \delta_{||\boldsymbol{p}||_1,P}|\partial^{p_{1}}_{1}\cdots\partial^{p_{N}}_{N} J(\boldsymbol{u})| \leq \frac{(N \omega_{\mathrm{max}}\delta t)^P}{2} = \frac{L^P}{2}, \]
where we defined the nondimensional parameter $L=\omega_{\mathrm{max}}T$, with $T= N\delta t$ the total evolution time, and $\boldsymbol{p}=(p_{1},p_{2},\cdots,p_{N})^T$. Since the parameter $L$ does not depend on the circuit depth $N$ alone, this points to the bounds being relevant also in the case of continuous controls, which are approximated by the stepwise-constant pulses that we studied. 

Another consequence of the bound on the derivatives is given by the following result:
\begin{lemma}[Lipschitz continuity]
    The function $J(\boldsymbol{u})$ is Lipschitz continuous, that is
    \begin{equation*}
        |J(\boldsymbol{u}) - J(\boldsymbol{u'})| \leq K ||\boldsymbol{u} - \boldsymbol{u'}||_1
    \end{equation*}
    where the Lipschitz constant satisfies $0 \leq K \leq \omega_{\mathrm{max}}\delta t/2$. Moreover, if $\boldsymbol{u}$ is a critical point $\grad J(\boldsymbol{u})=\boldsymbol{0}$, then also the following inequality holds
    \begin{equation*}
        |J(\boldsymbol{u}) - J(\boldsymbol{u'})| \leq K_{c} ||\boldsymbol{u} - \boldsymbol{u'}||_1^2.
    \end{equation*}
    where $0 \leq K_c \leq (\omega_{\mathrm{max}} \delta t)^2 /2$.
    \label{lemma:lipschitz}
\end{lemma}
\begin{proof}
    See App.~\ref{appendix_c} for a proof valid in the case of multiple controls and a generic observable $\hat{O}$.
\end{proof}
In practical terms, this means that the maximum amount of variation in the landscape value is fixed through the constant rate $K$ to the distance between sampled points. This fact has important consequences for optimization that we will explore in Sec.~\ref{metrics}-\ref{applications}.
Besides computing an upper bound for $K$ in terms of the physical properties of the systems, this results clarifies what is the relevant notion of distance between controls that should be used in this case, namely the $L_1$ norm $||\cdot||_1$ (also known as ``taxicab" norm, see App.~\ref{appendix_a} for definitions). In order to study the scaling for large $N$, we can reabsorb the factor $\delta t$ inside the norm
\[ \sum_{\nu=1}^N \delta t|u(t_{\nu}) - u'(t_{\nu})|  \xrightarrow[]{N\xrightarrow[]{}\infty} \int_0^T dt |u(t) - u'(t)|. \]

We can also show that the properties we just discussed still apply if we work with parametrized controls, which is often the case in optimal control \cite{motzoi2011optimal, caneva2011chopped, Sherson18}. The following corollary to Lemma \ref{lemma:derivatives} shows that, provided the basis functions are appropriately normalized, the same bounds on the derivatives, and therefore on the Lipschitz constant, are satisfied:
\begin{corollary}[Linear parametrizations of controls] \label{corollary:parametrization}
    Let us consider a linear parametrization of the controls $\boldsymbol{u} = \boldsymbol{Rv}$, where $\boldsymbol{v}\in\mathbb{R}^{N_c}$ and $\boldsymbol{R}\in\mathbb{R}^{N \times N_c}$,
    which is normalized as follows:
    \[ \forall i=1,\dots,N_c\ \ ||\boldsymbol{r}_i||_1 := \sum_{\nu=1}^N |r_{\nu i}| = 1.\]
     where $\boldsymbol{r}_i$ are the columns of $\boldsymbol{R}$.
     
    Then, the derivatives of the parametrized landscape $\tilde{J}(\boldsymbol{v}) = J(\boldsymbol{Rv})$ obey the bound
    \[ |\tilde{\partial}^{p_{1}}_{1}\cdots\tilde{\partial}^{p_{N_c}}_{N_c} \tilde{J}(\boldsymbol{v})|   \leq  \frac{(\omega_{\mathrm{max}} \delta t )^P}{2}\prod_{i=1}^{N_c} ||\boldsymbol{r}_{i}||^{p_i}_1 = \frac{( \omega_{\mathrm{max}} \delta t)^P}{2}, \]
    where $P = \sum_{\nu} p_{\nu}$, $\tilde{\partial}_{\mu'\nu'} := \frac{\partial}{\partial v_{\mu'\nu'}}$ and $\partial_{\mu \nu} := \frac{\partial}{\partial u_{\mu \nu}}$
\end{corollary}
\begin{proof}
    See App.~\ref{appendix_c} for a  proof valid in the case of multiple controls and a generic observable $\hat{O}$.
\end{proof}
Once again, we can derive an appropriate formula for scaling in $N$ by including the factor $\delta t$ in the normalization of the basis functions.

\section{Taylor representation}\label{taylor}
The bound on the landscape derivatives that we found in Lemma \ref{lemma:derivatives} also gives important information regarding how efficiently the landscape can be locally represented by means of a Taylor expansion. By taking the reference pulse $\boldsymbol{u}_0$ as the expansion point, the Taylor representation up to order $P$ can be written in the following form \footnote{Typically the Taylor expansion is written using a multiindex that avoids multiple counting, but as Lemma \ref{lemma:taylor} shows, the two are equivalent.} \cite{Smirnov}
\begin{align*} 
J_P(\boldsymbol{u}) &= \sum_{p = 0}^P \sum_{\nu_1 \cdots \nu_p} a_{\boldsymbol{\nu}} (\boldsymbol{u} - \boldsymbol{u}_0)_{\nu_1}\cdots (\boldsymbol{u} - \boldsymbol{u}_0)_{\nu_p},  \\
a_{\boldsymbol{\nu}} &= \frac{1}{p!} \partial_{\nu_1}\cdots\partial_{\nu_p}J(\boldsymbol{u}_0),
\end{align*} 
In order to see that the landscape admits a representation in this form, we find an upper bound for the approximation error by means of Lemma \ref{lemma:derivatives} and of the following classical result due to Lagrange and Taylor:
\begin{lemma}[Taylor approximation error in the Lagrange form]
\[  |J(\boldsymbol{u}) - J_P(\boldsymbol{u})| \leq \frac{1}{2}\frac{( \omega_{\mathrm{max}} \delta t)^{P+1}}{(P+1)!}||\boldsymbol{u}-\boldsymbol{u}_0||_1^{P+1} \]
\label{lemma:taylor}
\end{lemma}
\begin{proof}
    See App.~\ref{appendix_c} for a proof valid for multiple controls and a generic observable $\hat{O}$.
\end{proof}
Inside the control region of interest $\mathcal{C}^{(N)}(\boldsymbol{u}_0)$ the bound can be relaxed to
\begin{equation} 
\sup_{\boldsymbol{u} \in \mathcal{C}^{(N)}(\boldsymbol{u}_0)} |J(\boldsymbol{u}) - J_P(\boldsymbol{u})| \leq \frac{({u_{\mathrm{max}}}L)^{P+1}}{2(P+1)!} =: \epsilon(P),
\label{eq:taylor_bound}
\end{equation}
where we recall that $L=\omega_{\mathrm{max}}N\delta t$.
Intuitively, this implies that for ${u_{\mathrm{max}}} < L^{-1}$, the error of the $P-$order expansion is suppressed by both the numerator (exponentially) and the denominator (factorially), defining a control region of improved convergence in $P$. On the other hand, for ${u_{\mathrm{max}}} \geq L^{-1}$ the numerator diverges exponentially, and the factorial suppression of the denominator must first kick in for the error to vanish. Arbitrary precision can be obtained in both cases by choosing $P$ to be large enough.

In Fig.~\ref{fig:taylor} we study what is the minimum required order of the Taylor expansion so that $\epsilon(P) = \epsilon$ for a given error threshold $\epsilon$.
As expected, we see that for ${u_{\mathrm{max}}}L < 1$ the we do not need very high order expansions to obtain high precision, whereas for ${u_{\mathrm{max}}}L > 1$ the order required for the expansion scales linearly with ${u_{\mathrm{max}}}L$.
The latter observation can be understood by developing the bound in Eq.~\eqref{eq:taylor_bound} using Stirling's approximation for the factorial under the assumption ${u_{\mathrm{max}}}L,P\gg1$:
\begin{multline*}
    0\approx\frac{\log(\epsilon)}{P} \approx \frac{1}{P}\log\left( \frac{({u_{\mathrm{max}}}L)^P}{2P!}\right) \approx \left( \log\frac{{u_{\mathrm{max}}}L}{P} +1 \right),
\end{multline*}
which implies the linear dependence $P \approx e {u_{\mathrm{max}}}L$.
\begin{figure}
    \centering
    \includegraphics[width=.5\textwidth]{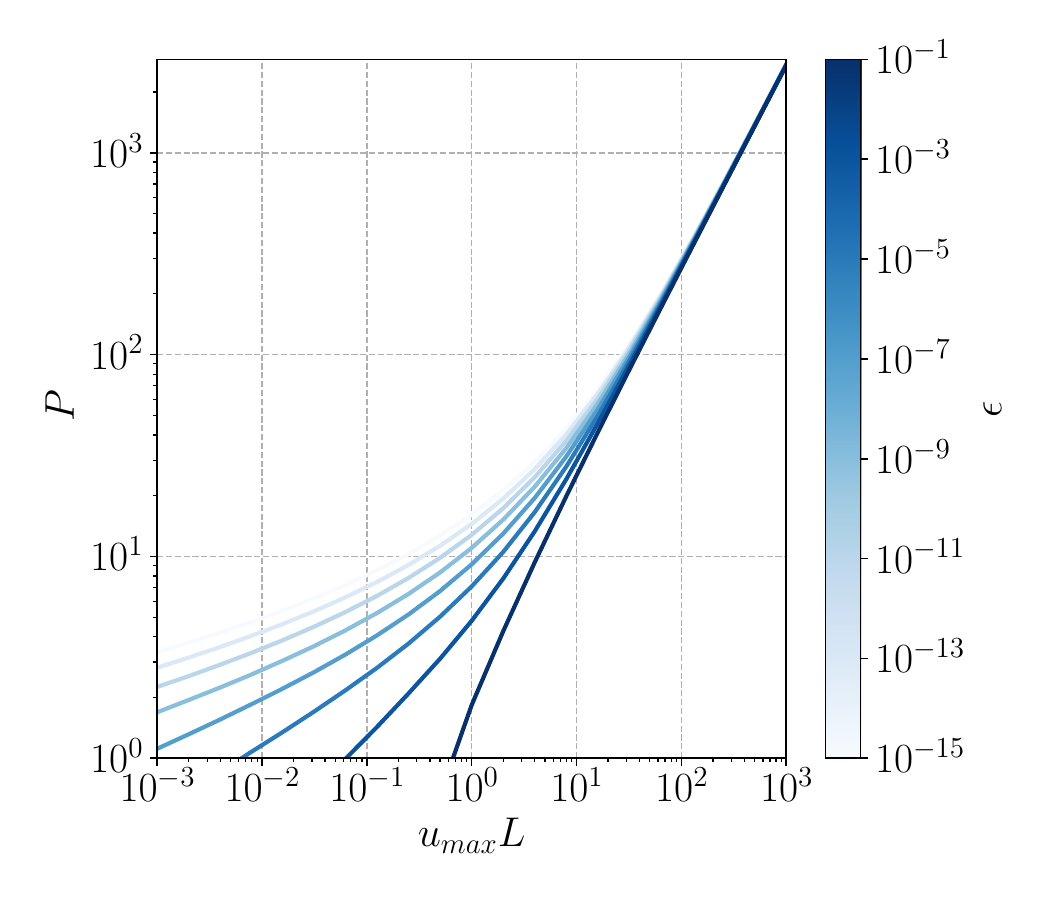}
    \caption{Given an error threshold $\epsilon$, we plot the solution for $P$ of the equation $\epsilon(P)=\epsilon$. This quantity represents the minimum order of Taylor expansion $P$ to represent $J$ up to an error $\epsilon$. The different lines show the results for several values for the error threshold $\epsilon=\{10^{-1}, 10^{-3}, 10^{-5}, 10^{-7}, 10^{-9}, 10^{-11}, 10^{-13}, 10^{-15}\}.$
    There is a crossover from ${u_{\mathrm{max}}}L < 1$, where $P\lesssim 10$ and it depends weakly on ${u_{\mathrm{max}}}L$, enabling an efficient local representation of the landscape, to the region ${u_{\mathrm{max}}}L \gg 1$ where the dependence becomes linear $P\sim e{u_{\mathrm{max}}}L$.}
    \label{fig:taylor}
\end{figure}

This error analysis shows that the landscape can be represented locally  (i.e. for ${u_{\mathrm{max}}}L \lesssim 1$, which is the case if the pulses amplitude ${u_{\mathrm{max}}}$, total time $T$ and/or maximum transition frequency $\omega_{\mathrm{max}}$ are sufficiently small) with a number of parameters which is polynomial in the number of time steps $N$, with a degree $P\lesssim10$.
As an example, Fig.~\ref{fig:taylor} shows that it is possible to represent the landscape with an error below $\epsilon \lesssim 10^{-3}$ for ${u_{\mathrm{max}}}L < 1$  with a $P=5$ representation, which contains $\mathcal{O}(N^5)$ coefficients or even more locally for ${u_{\mathrm{max}}}L \lesssim 0.25$ with a $P=2$ representation, which only contains $\mathcal{O}(N^2)$ terms (corresponding to the gradient and Hessian of the landscape at the expansion point). 

Such local representations provide relevant information during optimization and can be used to obtain noise-robust gradient estimates or to allow a direct jump to the minimum of the local landscape, speeding up convergence \cite{Goodwin16, Koczor22QAD}. In particular, for $P=2$ the latter approach is equivalent to applying Newton's method \cite{Goodwin16}, for which our results could help bound the convergence rate. Owing to its generality, the Taylor representation also serves as a performance baseline for any other local representation which exploits specific details of the system at hand \cite{Koczor22QAD}.

In the case of bang-bang controls $s_{\nu} = \pm 1$, the Taylor representation is very similar to the one presented in \cite{Bukov19}, although in that work the landscape is related to the logarithm of the infidelity.
Since the controls are chosen this way, the only allowed powers are $p_{\nu}=0,1$ and therefore the expansion is finite. The fact that the authors of \cite{Bukov19} could successfully represent the landscape for small values of $T$ using only the first few orders in their expansion is qualitatively similar to the picture that arises from our analysis. In our case we can roughly choose the timescale $T^*$ that defines the crossover as
\[T^* = \frac{1}{{u_{\mathrm{max}}}\omega_{\mathrm{max}}}.\]
As we will see in Sec.~\ref{metrics}, $T^*$ can be interpreted as a Quantum Speed Limit, so for $T/T^* \ll 1$ the infidelity cannot reach values arbitrarily close to $0$. But then the composition of $J$ with the logarithm should only add a prefactor close to $1$ to the bounds we discussed, leading to a qualitatively similar picture.
\section{Kernel representation and landscape learning}\label{kernel}
We now turn to the question regarding how it is possible, given a sample  $\mathcal{D}_{\mathrm{train}}$ of $N_{\mathrm{train}}$
control values $\boldsymbol{u}_i$ and the corresponding fidelities $J(\boldsymbol{u}_i) =: J_i$, namely
\[\mathcal{D}_{\mathrm{train}} = \{ (\boldsymbol{u}_i, J_i) \}_{i=1}^{N_{\mathrm{train}}}\]
to efficiently predict the fidelity for new control values.

This problem can be stated in the language of machine learning as a supervised learning (or regression) task, where given a family of models of the landscape $J_{\boldsymbol{w}}(\boldsymbol{u})$ parametrized by the weights $\boldsymbol{w} \in \mathbb{C}^{N_{\mathrm{weights}}}$, we look for the vector $\bar{\boldsymbol{w}}$ whose associated model best fits the points in the training data set.
This is quantified by means of a loss function $\mathfrak{L}$, that we choose to be the sum of the square residuals with RIDGE regularization \cite{Bishop06, Bukov19, Landman22}:
\[ \mathfrak{L}(\boldsymbol{w}, \boldsymbol{w}^{\dagger}, \mathcal{D}_{\mathrm{train}}) = \sum_{i=1}^{N_{\mathrm{train}}} |J_{\boldsymbol{w}}(\boldsymbol{u}_i) - J_i|^2 + \lambda_{R} \boldsymbol{w}^{\dagger} \boldsymbol{w} ,\]
so that $\bar{\boldsymbol{w}}$ will correspond to the minimum of $\mathfrak{L}$.
If the considered model has a non-linear dependence on the parameters, which is the case for deep learning \cite{Dalgaard22}, one usually has to resort to gradient-based optimization to find $\bar{\boldsymbol{w}}$. In the case of a linear model, that is
\begin{equation}
J_{\boldsymbol{w}}(\boldsymbol{u}) = \boldsymbol{w}^{\dagger} \boldsymbol{\phi}(\boldsymbol{u}),
\label{eq:lin_model}
\end{equation}
where the non-linear functions $\boldsymbol{\phi}(u)=(\phi_1(u),\phi_2(u),\cdots,\phi_{N_{\mathrm{weights}}}(u))^T: \mathbb{R}^{N} \xrightarrow{} \mathbb{C}^{N_{\mathrm{weights}}}$ are usually called features, the problem reduces to linear regression, which can be solved algebraically. This is possible since in this case the loss reduces to a quadratic form in the weights:
\[\mathfrak{L}(\boldsymbol{w}, \boldsymbol{w}^{\dagger}) =  \boldsymbol{w}^{\dagger} \boldsymbol{\Phi}^{\dagger} \boldsymbol{\Phi} \boldsymbol{w} + \boldsymbol{J} ^{\dagger}\boldsymbol{J} - \boldsymbol{J}^{\dagger}\boldsymbol{\Phi} \boldsymbol{w} - \boldsymbol{w}^{\dagger}\boldsymbol{\Phi}^{\dagger} \boldsymbol{J} + \lambda_R  \boldsymbol{w}^{\dagger} \boldsymbol{w}, \]
where we have defined the so-called feature matrix $\boldsymbol{\Phi}$ 
\[ \boldsymbol{\Phi}^{\dagger} = 
    \begin{pmatrix}
    \boldsymbol{\phi}(\boldsymbol{u}_1) & \dots & \boldsymbol{\phi}(\boldsymbol{u}_{N_{\mathrm{train}}}) \\
    \end{pmatrix}. \]
The stationary points $\grad_{\boldsymbol{w^{\dagger}}}\mathfrak{L} = \boldsymbol{0}$ can be found by solving the linear equation
\begin{equation}
\boldsymbol{w} =  -\frac{1}{\lambda_R} \boldsymbol{\Phi}^{\dagger}(\boldsymbol{\Phi}\boldsymbol{w} - \boldsymbol{J}).
\label{eq:L_stationary}
\end{equation}
Since the loss is real, differentiating with respect to $\boldsymbol{w}$ gives rise to the Hermitian conjugate of this equation, which is also satisfied when this equation is satisfied. The solution can be written explicitly via matrix inversion:
\[ \bar{\boldsymbol{w}} = \left( \boldsymbol{\Phi}^{\dagger}\boldsymbol{\Phi} + \lambda_R \mathds{1} \right)^{-1}\boldsymbol{\Phi}^{\dagger} \boldsymbol{J}. \]
We see from this expression that the RIDGE parameter $\lambda_R$ regularizes the inverse of the covariance matrix $\boldsymbol{\Phi}^{\dagger}\boldsymbol{\Phi} \in \mathbb{C}^{N_{\mathrm{weights}} \times N_{\mathrm{weights}}}$, avoiding potential problems arising from small eigenvalues of $\boldsymbol{\Phi}^{\dagger}\boldsymbol{\Phi}$, which can be caused by correlations in the data set.

In light of the Lie-Fourier representation we discussed in Sec.~\ref{lie_fourier}, a natural choice for the features is a complex exponential of the form
\[ \phi_{\boldsymbol{\omega}}(\boldsymbol{u}) = e^{-i \delta t \boldsymbol{\omega} \cdot \boldsymbol{u}}.\]
Since the frequencies fill up densely the hypercube $\boldsymbol{\omega} \in [-\omega_{\mathrm{max}}, \omega_{\mathrm{max}}]^N$, in the absence of a prior expectation on their distribution, we consider random Fourier features \cite{Landman22}, which we sample with uniform probability in the frequency hypercube. As we will see later, the boundedness of the control region we are considering results in a certain freedom to choose the frequencies in the model, so that they do not necessarily have to match the ones from the Lie-Fourier representation. This can potentially relieve us from having to deal with an infinite number of frequencies.

Another viable option is to use polynomial features of the form
\[\phi_{\boldsymbol{p}}(\boldsymbol{u}) = \prod_{\nu=1}^N (\boldsymbol{u} - \boldsymbol{u}_0)^{p_1} \cdots (\boldsymbol{u} - \boldsymbol{u}_0)^{p_N}, \]
which give rise to a Taylor representation like the one we studied in Sec.~\ref{taylor}, provided that the degree vectors $\boldsymbol{p}$ are chosen appropriately. Based on our discussion of the error, we would expect this representation to efficiently learn the landscape at least for $u_{\mathrm{max}}L < 1$.

These two options both suffer from the same problem, namely, that in order to solve the problem numerically they require subselecting a finite number of representative features $N_{\mathrm{weights}}$ from an infinite-dimensional feature space.  This space is in general needed in full to represent the landscape exactly, as we saw in detail in Sec.~\ref{lie_fourier} for the Lie-Fourier representation and in Sec.~\ref{taylor} for the Taylor representation. But then, we always have to neglect features that might be important, which is potentially problematic in the absence of guiding principles on how to choose the ones we keep.

There is a way around this problem that allows us to effectively perform infinite-dimensional linear regression, by means of the so-called kernel trick \cite{Bishop06, Landman22}.
The starting point of this method is to introduce the column vector $\boldsymbol{a}$ as follows:
\[ \boldsymbol{a} = -\frac{1}{\lambda_R}(\boldsymbol{\Phi}\boldsymbol{w} - \boldsymbol{J})\]
so that together with Eq.~\eqref{eq:L_stationary} we have
\[ \boldsymbol{w}(\boldsymbol{a}) = \boldsymbol{\Phi}^{\dagger} \boldsymbol{a}.\]
We can now rewrite the loss as a function of $\boldsymbol{a}$ alone:
\[  \mathfrak{L}(\boldsymbol{a}, \boldsymbol{a}^{\dagger}) = \boldsymbol{a}^{\dagger} \boldsymbol{K} \boldsymbol{K} \boldsymbol{a} - \boldsymbol{a}^{\dagger} \boldsymbol{K} \boldsymbol{J} - \boldsymbol{J}^{\dagger} \boldsymbol{K} \boldsymbol{a} + \boldsymbol{J}^{\dagger}\boldsymbol{J} + \lambda_{R} \boldsymbol{a}^{\dagger} \boldsymbol{K} \boldsymbol{a},\]
where we defined the kernel matrix $\boldsymbol{K} = \boldsymbol{\Phi}\boldsymbol{\Phi}^{\dagger}$, turning the original infinite-dimensional regression problem into a finite dimensional kernel regression. Since $\boldsymbol{K}^{\dagger} = \boldsymbol{K} \in \mathbb{C}^{N_{\mathrm{train}} \times N_{\mathrm{train}} }$, the dimensionality of this new regression problem is set by the number of data points $N_{\mathrm{train}}$ in the training dataset.
Once again, we have to minimize the loss by solving $\grad_{\boldsymbol{a}^{\dagger}}\mathfrak{L}=\boldsymbol{0}$, which yields
\[ \bar{\boldsymbol{a}} = (\boldsymbol{K} + \lambda_R \mathds{1})^{-1} \boldsymbol{J}. \]
Finally, inference on a new data point is to be carried out by substitution inside Eq.~\eqref{eq:lin_model}
\[J_{\boldsymbol{w}(\bar{\boldsymbol{a}})}(\boldsymbol{u}) = \bar{\boldsymbol{a}}^{\dagger} \boldsymbol{\Phi} \boldsymbol{\phi}(\boldsymbol{u}) = \sum_{i=1}^{N_{\mathrm{train}}} \bar{a}^*_i \kappa(\boldsymbol{u}_i,\boldsymbol{u}),\]%\bar{\boldsymbol{a}}^{\dagger} \boldsymbol{k}(\boldsymbol{u})\]
where we defined the kernel function $\kappa(\boldsymbol{u},\boldsymbol{u}') := \boldsymbol{\phi}^{\dagger}(\boldsymbol{u}) \boldsymbol{\phi}(\boldsymbol{u}') $, which also appears in the matrix elements of the kernel matrix $K_{ij} = \kappa(\boldsymbol{u}_i,\boldsymbol{u}_j)$.
Kernel methods like this one are categorized as instance-based learning methods, because the prediction for a new data point is obtained as a linear combinations of kernel functions evaluated on the training data points \cite{Kokail19, Sauvage20, Dalgaard22}.

In the case of complex exponential features, the calculation of the kernel functions can be carried out analytically, and is particularly straightforward in the case of equally spaced frequencies as they appear in the Ising model we discussed in Sec.~\ref{numerics}.
In fact, in that case we have that
\begin{widetext}
\begin{multline*}
\kappa(\boldsymbol{u},\boldsymbol{u}') = \boldsymbol{\phi}^{\dagger}(\boldsymbol{u}) \boldsymbol{\phi}(\boldsymbol{u}') = \lim_{n \xrightarrow[]{} \infty} \frac{1}{n_\Delta^N} \sum_{\boldsymbol{\omega} \in (\mathcal{S}^{\Delta}_n)^N}  e^{-i\delta t \boldsymbol{\omega}\cdot (\boldsymbol{u} - \boldsymbol{u}')} 
= \lim_{k_{\mathrm{max}} \xrightarrow[]{} \infty} \left(\frac{1}{2k_{\mathrm{max}}+1}\right)^N \sum_{\boldsymbol{k} = -\boldsymbol{k}_{\mathrm{max}}}^{\boldsymbol{k}_{\mathrm{max}} }  e^{-i \frac{\omega_{\mathrm{max}} \delta t }{ k_{\mathrm{max}} } \boldsymbol{k} \cdot (\boldsymbol{u} - \boldsymbol{u}')} =\\
= \prod_{i=1}^{N} \lim_{k_{\mathrm{max}} \xrightarrow[]{} \infty} \frac{1}{2k_{\mathrm{max}}+1} \left( \sum_{k_i = 0}^{k_{\mathrm{max}} }  e^{-i \frac{\omega_{\mathrm{max}} \delta t }{ k_{\mathrm{max}} } k_i (u_i - u_i')} + c.c. - 1 \right)  
= \prod_{i=1}^{N} \lim_{k_{\mathrm{max}} \xrightarrow[]{} \infty} \frac{1}{2k_{\mathrm{max}}+1}\left(\frac{e^{-i \omega_{\mathrm{max}} \delta t (u_i - u_i')}}{1-e^{+i \frac{\omega_{\mathrm{max}} \delta t }{ k_{\mathrm{max}} }(u_i - u_i')}} + c.c. \right) \\
= \prod_{i=1}^{N} \lim_{k_{\mathrm{max}} \xrightarrow[]{} \infty} \frac{\mathcal{O}(k_{\mathrm{max}})}{2k_{\mathrm{max}}+1}\left(\frac{e^{-i \omega_{\mathrm{max}} \delta t (u_i - u_i')}}{-i\omega_{\mathrm{max}} \delta t (u_i - u_i')} + c.c. \right) 
= \prod_{i=1}^{N} \frac{\sin{[ \omega_{\mathrm{max}} \delta t (u_i - u_i')}]}{\omega_{\mathrm{max}} \delta t (u_i - u_i')} = \kappa(\boldsymbol{u}-\boldsymbol{u}'),
\end{multline*}
\end{widetext}
where we remind that $n_\Delta = \#\mathcal{S}^{\Delta}_n$ is the number of frequencies in the $n$-th order Lie-Fourier representation, and the sums over $k_i$ are computed using the well-known formula for the geometric series.
The fact that we obtain a (multidimensional) $\text{sinc}$ kernel is consistent with the bandwidth-limited nature of the problem. Moreover, we find that for a data set $\mathcal{D}_{\mathrm{train}}$ obtained by sampling an infinite cubic hyperlattice $\boldsymbol{u}_i \in \pi(\omega_{\mathrm{max}} \delta t )^{-1}\mathbb{Z}^N$ this instance of kernel regression reduces to the well known Whittaker-Shannon interpolation formula \cite{Whittaker15}. In fact, since the $\sinc$ function reduces to a delta function on the hyperlattice points, we have that
\[ \kappa(\pi(\omega_{\mathrm{max}} \delta t )^{-1}(\boldsymbol{p}-\boldsymbol{q})) = \delta_{\boldsymbol{p},\boldsymbol{q}}\ \ \forall \boldsymbol{p},\boldsymbol{q} \in \mathbb{Z}^N, \]
therefore the kernel is the identity matrix $K_{ij}=\delta_{ij}$ and $\bar{\boldsymbol{a}} = \boldsymbol{J}$ by setting $\lambda_R=0$, which gives us the desired result
\[J_{\boldsymbol{w}(\bar{\boldsymbol{a}})}(\boldsymbol{u}) = \sum_{i=1}^{N_{\mathrm{train}}} J_i \kappa(\boldsymbol{u} - \boldsymbol{u}_i) = \sum_{\boldsymbol{p} \in \mathbb{Z}^N} J_i \kappa \left[ \boldsymbol{u} - \frac{\pi \boldsymbol{p}}{\omega_{\mathrm{max}} \delta t } \right].\]

According to this kernel representation, the estimated landscape is a linear superposition of sinc kernels with a given wavelength. Moreover, in the case of the Whittaker-Shannon formula, the individual kernels represent the solution for a single sampling point, playing the same role of a Green function for a classical linear field theory. This suggests that, similarly to classical optics, the landscape has a maximum resolution related to the kernel's wavelength, below which details cannot be distinguished.
We will make this argument more precise in the following section.
\subsection{Numerical examples}
We now benchmark the performance in the regression task of the various representations for the same quantum dynamical landscape we already studied numerically in Sec.~\ref{lie_fourier}. The system we consider is the transverse field Ising model in Eq.~\eqref{eq:ising}, with no longitudinal field $h_z=0$ and the state transfer problem  $\ket{{0}_Q}\mapsto \ket{{1}_Q}$. We generate a training dataset by randomly sampling a set of controls uniformly inside the hypercube $\boldsymbol{u}_i \in \mathcal{C}^N = [-u_{\mathrm{max}}, u_{\mathrm{max}}]$, and then we compute the state transfer fidelity $J_i$ for each one of the sampled controls. Then we train each regression model as described previously, and evaluate the prediction error on a new set data points $\mathcal{D}_{\mathrm{test}}=\{ (\boldsymbol{u}_i, J(\boldsymbol{u}_i)) \}_{i=1}^{N_{\mathrm{test}}}$, which we define as the root mean squared error: 
\[ \epsilon_{\mathrm{rms}} = \sqrt{\frac{1}{N_{\mathrm{test}}}\sum_{i=1}^{N_{\mathrm{test}}} (J_i- \smashoperator[lr]{\sum_{j=1}^{N_{\mathrm{weights}}}} w_j \phi_j(\boldsymbol{u}_i))^2} \]

In order to obtain a fair comparison between the three regression models, we need to take into account several aspects.
First of all, while the number of features $N_{\mathrm{weights}}$ can be chosen freely for the Fourier and Taylor features, it is fixed to $N_{\mathrm{weights}}=N_{\mathrm{train}}$ in the case of kernel regression, as it is an instance-based method. For this reason, we proceed as follows: we split the training data set into a set that we use for the training itself and a validation data set which we use to evaluate the error $\epsilon_{\mathrm{rms}}$, repeat the training for a range of values of $N_{\mathrm{weights}}$, and finally fix $N_{\mathrm{weights}}$ by choosing the value that gives the lowest validation error.
Another important aspect is the choice of the RIDGE regularization parameter $\lambda_R$. 
We empirically found that keeping $N_{\mathrm{weights}} < N_{\mathrm{train}}$ greatly simplifies the analysis in this regard, since then the best results are typically obtained by fixing $\lambda_R$ as small as it is permitted by numerical stability.
The results of this analysis are presented in Fig.~\ref{fig:ml}, where the prediction error $\epsilon_{\mathrm{rms}}$ is plotted as a function of $N_{\mathrm{train}}$ for the various regression models.
\begin{figure*}
    \centering
    \includegraphics[width=0.9\textwidth]{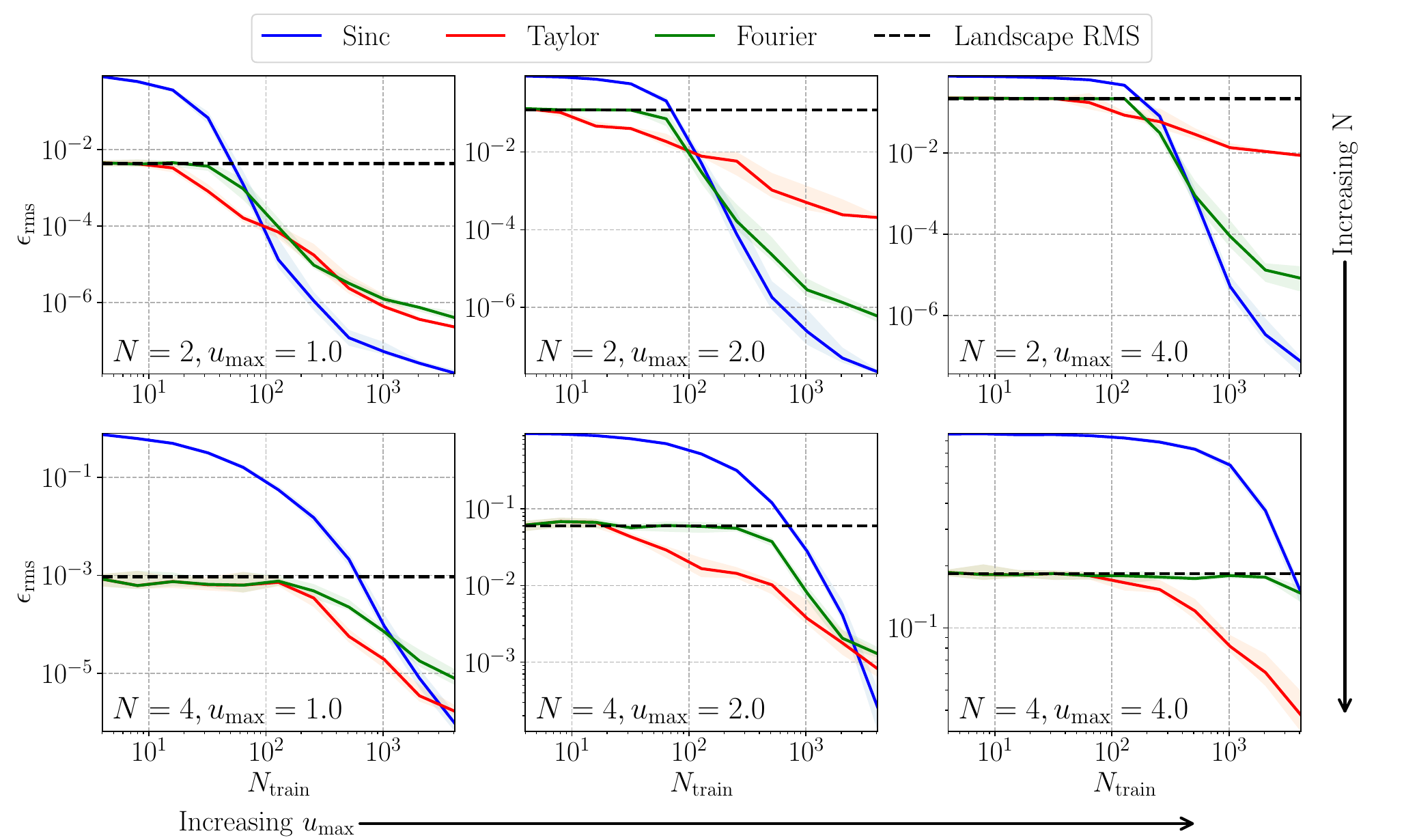}
    \caption{Prediction performance $\epsilon_{\mathrm{rms}}$ of surrogate models as a function of training dataset size $N_{\mathrm{train}}$ for the system given by Eq.~\eqref{eq:ising}. The results shown in the plot are obtained for $Q=5$ qubits, time $T=1.0$, where the colors correspond to the choice of feature map. Each of the six plots corresponds to a value of the parameters $(N, u_{\mathrm{max}})$  $\in \{2,4\} \times \{1.0,2.0,4.0\}$. For Taylor and Fourier features we have $\lambda_{R}=10^{-6}$, and only the result for the optimal value of $N_{\mathrm{weights}} \leq N_{\mathrm{train}}$ is shown, while for the sinc kernel we have $\lambda_{R}=10^{-12}$. The training datasets are sampled from a pool of $12672$ controls, while the test datasets with $N_{\mathrm{test}}=128$ are sampled from another pool of $3200$ and each training/test is repeated $32$ times. The solid lines correspond to the median of the prediction errors $\epsilon_{\mathrm{rms}}$ over the samples, while the shaded area corresponds to the $25-75$ percentile range. The dotted line shows the square root of the variance of the sampled landscape.  Compared to the other feature maps, the sinc kernel model typically shows lower values of $\epsilon_{\mathrm{rms}}$ for large values of $N_{\mathrm{train}}$, while the opposite is true for the Taylor representation.}
    \label{fig:ml}
\end{figure*}

The most significant finding is that the sinc kernel features appear to provide the best prediction performance if $N_{\mathrm{train}}$ is large enough. Even though its prediction error is larger than for the other feature maps when the training dataset is too small, the plots suggest that there is a threshold value for $N_{\mathrm{train}}$ above which the sinc kernel outperforms both the Taylor and Fourier feature maps. This threshold becomes larger as $N, u_{\mathrm{max}}$ increase in value. Meanwhile, for few training samples we see that the Taylor representation generally does best. In the intermediate range, a crossover takes place between Taylor being best for few time step layers, while Fourier works better for more time steps.

We also find qualitatively that the sinc kernel regression is generally more stable and can tolerate smaller values of $\lambda_R$ compared to the other choices. Finally, we find that the inferior performance of the sinc kernel for low $N_{\mathrm{train}}$ can be greatly improved by decreasing the kernel bandwidth to $\omega_{ker} \leq \omega_{max}$, giving rise to
\[ \kappa(\boldsymbol{u}-\boldsymbol{u}') = \prod_{i=1}^{N} \frac{\sin{[ \omega_{\mathrm{ker}} \delta t (u_i - u_i')}]}{\omega_{\mathrm{ker}} \delta t (u_i - u_i')}. \]
The results that we obtain using this strategy are showcased in Fig.~\ref{fig:ml_bandwidth}. There we see that a trade-off between precision on small and large data sets arises, so that $\omega_{ker}$ has to be tuned based on practical considerations regarding the problem and the data resources at hand.
\begin{figure}
    \centering
    \includegraphics[width=0.9\linewidth]{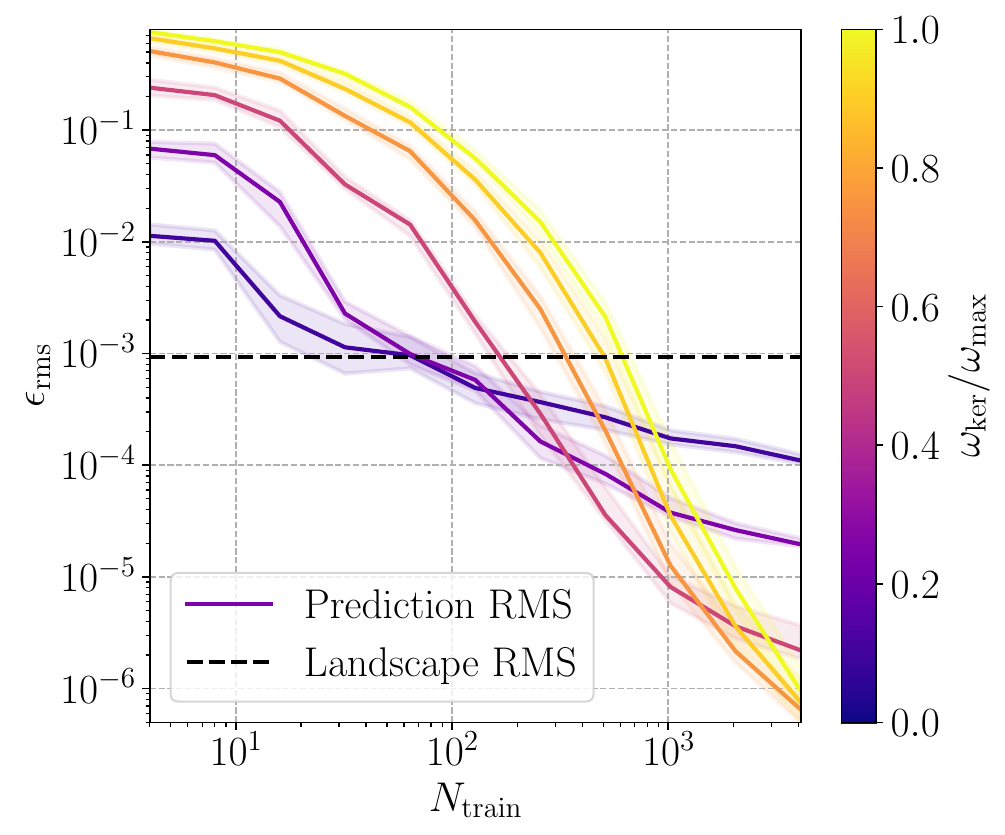}
    \caption{Sinc kernel regression with reduced kernel bandwidth $\omega_{\mathrm{ker}} \leq \omega_{\mathrm{max}}$ on the system given by Eq.~\eqref{eq:ising}. The results shown in the plot are obtained for $Q=5, T=1.0, N=4, u_{\mathrm{max}}=1.0$ and the different colors (from blue to yellow) corresponding to $\omega_{\mathrm{ker}} / \omega_{\mathrm{max}} = 0.1,0.3,0.5,0.8,0.9,1.0$. The training datasets are sampled from a pool of $12672$ points and each training is repeated $32$ times. The median (solid lines) and interquartile range (shaded area) of the prediction errors $\epsilon_{\mathrm{rms}}$ over the samples are shown. The dotted line shows the square root of the variance of the sampled landscape. Reducing $\omega_{\mathrm{ker}}$ considerably reduces prediction error for small training datasets, but also increases it for large datasets. }
    \label{fig:ml_bandwidth}
\end{figure}
\subsection{Analysis of Fourier regression} \label{fourier_regression}
Let us now go back to Fourier features regression and examine the problem more closely. A first insight that we can gain is that the regression problem for bounded sampled controls $\boldsymbol{u}_i \in \mathcal{C}^N$ does not require finding an exact representation of the landscape $J$ over $\boldsymbol{u} \in \mathbb{R}^N$, for which we would indeed need an infinite number of frequencies. All we need to do is to find an approximation that is good enough inside the hypercube $\mathcal{C}^N$. This has important consequences as far as the Fourier spectrum of the approximations is concerned. In order to understand this point, we reformulate the loss we have chosen (we fix $\lambda_R=0$ for simplicity) as the Monte Carlo sampling of an integral:
\[ \sum_{i=1}^{N_{\mathrm{train}}} \frac{|J_{\boldsymbol{w}}(\boldsymbol{u}_i) - J_i|^2}{N_{\mathrm{train}}}  \approx  \int_{\mathcal{C}^N} \frac{d^N\boldsymbol{u}}{(2{u_{\mathrm{max}}})^N} |J_{\boldsymbol{w}}(\boldsymbol{u}) - J(\boldsymbol{u})|^2 \]
where we assume the sampled controls to be drawn from the uniform probability distribution over $\mathcal{C}^N$.
As we saw in Sec.~\ref{lie_fourier}, $J$ can be approximated with arbitrarily low error over $\mathcal{C}^N$ with its Lie-Fourier representation $J_n$ for a large enough $n$.
We now try to solve the regression problem for the integral loss by choosing the model $J_{\boldsymbol{w}}$ as a sum of Fourier components picked from a set of frequencies $\mathcal{E}$ (that we can suppose to be finite, bounded and satisfying the same symmetries as the Lie-Fourier spectrum $(\mathcal{S}^{\Delta}_n)^N$). The integral loss can then be written directly in term of the weights as follows:
\begin{widetext}
\begin{multline} 
\tilde{\mathfrak{L}}(\boldsymbol{w}, \boldsymbol{w}^{\dagger}) = \int_{\mathcal{C}^N} \frac{d^N\boldsymbol{u}}{(2{u_{\mathrm{max}}})^N} \left|\sum_{\boldsymbol{\alpha} \in \mathcal{E}} w_{\boldsymbol{\alpha}}  e^{-i\delta t \boldsymbol{\alpha} \cdot \boldsymbol{u}} - \sum_{\boldsymbol{\omega} \in (\mathcal{S}^{\Delta}_n)^N} c_{\boldsymbol{\omega}}  e^{-i\delta t \boldsymbol{\omega} \cdot \boldsymbol{u}} \right|^2 = \\ \int_{\mathcal{C}^N} \frac{d^N\boldsymbol{u}}{(2{u_{\mathrm{max}}})^N}\left[\sum_{\boldsymbol{\alpha},\boldsymbol{\alpha}' \in \mathcal{E}} w^*_{\boldsymbol{\alpha}'} w_{\boldsymbol{\alpha}}  e^{i\delta t (\boldsymbol{\alpha}' - \boldsymbol{\alpha})\cdot \boldsymbol{u}} + \smashoperator[l]{\sum_{\boldsymbol{\omega},\boldsymbol{\omega}' \in (\mathcal{S}^{\Delta}_n)^N}} c^*_{\boldsymbol{\omega}'} c_{\boldsymbol{\omega}}  e^{i\delta t (\boldsymbol{\omega}' - \boldsymbol{\omega}) \cdot \boldsymbol{u}} - \smashoperator[l]{\sum_{\boldsymbol{\alpha}\in \mathcal{E},\boldsymbol{\omega} \in (\mathcal{S}^{\Delta}_n)^N}} [w^*_{\boldsymbol{\alpha}} c_{\boldsymbol{\omega}}  e^{i\delta t (\boldsymbol{\alpha} - \boldsymbol{\omega}) \cdot \boldsymbol{u}} + c^*_{\boldsymbol{\omega}} w_{\boldsymbol{\alpha}}  e^{i\delta t (\boldsymbol{\omega} - \boldsymbol{\alpha}) \cdot \boldsymbol{u}} ] \right] \\
= \sum_{\boldsymbol{\alpha},\boldsymbol{\alpha}' \in \mathcal{E}} w^*_{\boldsymbol{\alpha}'} w_{\boldsymbol{\alpha}} \tilde{\kappa}(\boldsymbol{\alpha}' - \boldsymbol{\alpha}) + \sum_{\boldsymbol{\omega},\boldsymbol{\omega}' \in (\mathcal{S}^{\Delta}_n)^N} c^*_{\boldsymbol{\omega}'} c_{\boldsymbol{\omega}}  \tilde{\kappa}(\boldsymbol{\omega}' - \boldsymbol{\omega}) - \sum_{\boldsymbol{\alpha}\in \mathcal{E},\boldsymbol{\omega} \in (\mathcal{S}^{\Delta}_n)^N} \left[ w^*_{\boldsymbol{\alpha}} c_{\boldsymbol{\omega}}  \tilde{\kappa} (\boldsymbol{\alpha} - \boldsymbol{\omega}) + c^*_{\boldsymbol{\omega}} w_{\boldsymbol{\alpha}}  \tilde{\kappa}(\boldsymbol{\omega} - \boldsymbol{\alpha}) \right] \label{eq:lossFourier}
\end{multline}
\end{widetext}
where we defined the frequency space kernel $\tilde{\kappa}$ as 
\[\tilde{\kappa}(\boldsymbol{\omega}) = \prod_{\nu=1}^{N} \frac{\sin{ \omega_{\nu} \delta t {u_{\mathrm{max}}}}}{\omega_{\nu} \delta t  {u_{\mathrm{max}}}}.\]

\subsubsection{Discrete frequencies approximation}
Since $\tilde{\kappa}$ has a finite resolution in Fourier space, the loss is to a first approximation only a function of a local average of the weights $w_{\alpha}$ within a certain volume in frequency space. We can study the consequences of this fact by expanding the kernel function for small arguments:
\[ \prod_{\nu=1}^{N} \frac{\sin{\omega_{\nu} \delta t  u_{\mathrm{max}}}}{\omega_{\nu} \delta t  u_{\mathrm{max}}} = 1 - \frac{(||\boldsymbol{\omega}||_2 \delta t {u_{\mathrm{max}}})^2}{6} + o((||\boldsymbol{\omega}||_2 \delta t {u_{\mathrm{max}}})^2) \]
This means that within a ball in frequency space defined by the Euclidean norm $||\cdot||_2$ and with radius $\Delta \omega$, we can approximate the kernel with a constant up to an error of order $(\Delta \omega\, \delta t\, {u_{\mathrm{max}}} )^2 $. Since the coefficients relating to frequencies closer than $\Delta \omega$ are then summed up in the loss function, the regression problem effectively depends only on a smaller subset of frequency modes than the original ones. In practice, this suggests that we can find approximate representations with fewer frequencies than the ones in the Lie-Fourier representation, as long as we are just interested in a bounded control region of interest. Estimating how many of these frequencies should be used is a hard problem because it is related to a $N$-dimensional sphere stacking problem \cite{HyperspherePacking}. 

\subsubsection{Flat landscape approximation}
It is perhaps even more interesting to take this reasoning to its extreme logical consequences and consider a regression model containing only the zero frequency mode $w_{\boldsymbol{0}}$. By substituting the small argument expansion for the kernel into Eq.~\eqref{eq:lossFourier} and neglecting the terms $o((\delta t {u_{\mathrm{max}}}||\boldsymbol{\omega}||_2)^2)$ \footnote{One can check that by carrying on this infinitesimal quantity in the calculations the result is the same}, we obtain the following
\begin{multline*} 
    \tilde{\mathfrak{L}}(w_{\boldsymbol{0}}, w_{\boldsymbol{0}}^*) =  |w_{\boldsymbol{0}}|^2 + \smashoperator[lr]{\sum_{\boldsymbol{\omega},\boldsymbol{\omega}' \in (\mathcal{S}^{\Delta}_n)^N}} c^*_{\boldsymbol{\omega}'} c_{\boldsymbol{\omega}}  (1 - \frac{\delta t^2 {u_{\mathrm{max}}}^2}{6} ||\boldsymbol{\omega}' - \boldsymbol{\omega}||_2^2) \\- \sum_{\boldsymbol{\omega} \in (\mathcal{S}^{\Delta}_n)^N} ( w^*_{\boldsymbol{0}} c_{\boldsymbol{\omega}}   + c^*_{\boldsymbol{\omega}} w_{\boldsymbol{0}} )(1 - \frac{\delta t^2 {u_{\mathrm{max}}}^2}{6}||\boldsymbol{\omega}||_2^2).
\end{multline*}
We can find the solution $\bar{w}_{\boldsymbol{0}}$ to the regression problem by solving the equation $\grad_{\boldsymbol{w}^{\dagger}} \mathfrak{L} = \boldsymbol{0}$. To the zeroth order approximation the solution is given (up to a phase) by 
\[ \bar{w}_{\boldsymbol{0}} = \sum_{\boldsymbol{\omega} \in (\mathcal{S}^{\Delta}_n)^N} c_{\boldsymbol{\omega}} = J_n(\boldsymbol{0}), \]
as defined by Eq.~\eqref{eq:fid_diff}. We can substitute again into the loss function to evaluate how well the constant model approximates the landscape. We obtain:
\begin{multline*} 
\mathfrak{L}(\bar{w}_{\boldsymbol{0}}, \bar{w}_{\boldsymbol{0}}^*) = \frac{\delta t ^2{u_{\mathrm{max}}^2}}{6} \smashoperator[lr]{\sum_{\boldsymbol{\omega},\boldsymbol{\omega}' \in (\mathcal{S}^{\Delta}_n)^N}} c^*_{\boldsymbol{\omega}'} c_{\boldsymbol{\omega}} (||\boldsymbol{\omega}||_2^2 + ||\boldsymbol{\omega}'||_2^2 - ||\boldsymbol{\omega}' - \boldsymbol{\omega}||_2^2  )  \\
\leq \frac{(\delta t {u_{\mathrm{max}}})^2}{6} \left| \sum_{\boldsymbol{\omega} \in (\mathcal{S}^{\Delta}_n)^N} c_{\boldsymbol{\omega}} \right|^2 2 N \omega_{\mathrm{max}}^2 \leq \frac{({u_{\mathrm{max}}} L)^2}{3N} \xrightarrow{N \xrightarrow{} \infty} 0,
\end{multline*}
where in the last step we made use of the boundedness of the landscape $|J_n(\boldsymbol{0})|^2 \leq 1$.

This results implies a remarkable property of any dynamical landscape with finite time-energy budget $L$: as the number of controls $N$ increase, the landscape becomes increasingly close to a flat landscape when distance between landscapes is measured using the $L_2$ norm. As we will see in the next Section, this is related to the appearance of (polynomial) barren plateaux. In relation to landscape learning, this suggests that the sum of square residuals for uniform finite samples is not a well defined loss function, unless $N$ is fixed, which in practice forces us to introduce a cutoff in control pulse discretization (or, equivalently, in frequency) in order to obtain a well defined regression problem.
Since numerical experiments necessarily deal with finite values of $N$, it is possible in practice to use the loss $L_2$ for regression, but these considerations cast a doubt on the relevance of scaling analysis for this kind of problems. 

In order to surpass the difficulties arising from this artificial cutoff, it is most likely to be necessary to employ stronger notions of distance (e.g. using the sup-norm $L_{\infty}$), non-uniform sampling strategies, or a combination of the two. Since the landscapes we are studying are, after all, objective functions to be optimized, sampling could also be provided by an optimizer \cite{Dalgaard22, Beato24}. This way, the focus of the problem would shift from learning the landscape itself to learning the landscape as it "appears" to the optimizer. 
\section{Landscape metrics} \label{metrics}
We now discuss the relevance of the results concerning quantum dynamical landscapes that we derived up to now in the context of optimization.
Finding the controls that minimize the landscape $J(\boldsymbol{u}) = \langle \hat{O}(\boldsymbol{u})\rangle$ is the problem we have to solve both in the context of optimal control (for which the landscape often corresponds to the infidelity, so that $\hat{O} = \hat{I} - \ket{\chi}\bra{\chi}$) and VQA. 

As we have seen in the previous section, the landscape can be represented by means of a kernel with a bandwidth which is set by the time and energy scales of the problem.
Therefore, analogously to classical optics, the finite bandwidth of the landscape should prevent us from distinguishing details below a certain length scale $l$ in control space, at least up to a certain tolerance $\delta J$. But then the landscape can be discretized up to an error $\delta J$ by sampling it on points standing at a distance $l$ from one another.
In order to make these considerations more precise, we can use the formula from Lemma \ref{lemma:lipschitz}. Then, we can see that for two controls $\boldsymbol{u},\boldsymbol{u'}$ whose infidelity differs by $\delta J$, the following holds true
\begin{equation}
||\boldsymbol{u} - \boldsymbol{u'}||_1 \geq \frac{2\delta J}{\omega_{\mathrm{max}} \delta t} =:l. 
\label{eq:min_dist}
\end{equation}
As anticipated, the two controls have to be separated by a certain minimum distance $l$ (which depends also on $\delta J$), that has to be measured using the taxicab norm $||\cdot||_1$.
This simple fact has multiple consequences on the landscape properties.

\subsection{General Quantum Speed Limit}

We first analyse the general structure of the Quantum Speed Limit (QSL), which can be defined as the minimum time $T_{QSL}$ needed to perform a control task, such as state transfer to the target state $\ket{\chi}$. We can obtain this from the infidelity landscape assuming the value $J=0$ at least once within the bounded control region of interest $\mathcal{C}^N = [-{u_{\mathrm{max}}},{u_{\mathrm{max}}}]^N$.
Even though we cannot say anything about the QSL for a single target state, we can put a lower bound on the time $T^{\perp}_{QSL}$ needed to reach both the original target state $\ket{\chi}$ and an orthogonal state $\ket{\chi_{\perp}}, \braket{\chi | \chi_{\perp}}=0$.
Although this definition of the QSL looks rather artificial, if the system exhibits full state controllability it will also satisfy this condition, so that $T_{QSL}^{full} \geq T^{\perp}_{QSL}$.
The bound on $T^{\perp}_{QSL}$ (and therefore on $T_{QSL}^{full}$) can be derived as follows: Since by definition there are controls $\boldsymbol{u}^*, \boldsymbol{u}^*_{\perp} \in \mathcal{C}^N$ such that $U(\boldsymbol{u}^*)\ket{\psi} = \ket{\chi}$ and $U(\boldsymbol{u}^*_{\perp})\ket{\psi} = \ket{\chi_{\perp}}$, the infidelity landscape will assume there the values $J(\boldsymbol{u}^*)= 1 -|\braket{\chi|\chi}|^2 = 0$ and $J(\boldsymbol{u}^*_{\perp})= 1-|\braket{\chi|\chi_{\perp}}|^2 = 1$. But then as we have already seen, Lemma \ref{lemma:lipschitz} implies that 
\[||\boldsymbol{u}^* - \boldsymbol{u}^*_{\perp}||_1 \geq |J(\boldsymbol{u}^*) - J(\boldsymbol{u}^*_{\perp})| K^{-1} \geq \frac{2}{\omega_{\mathrm{max}} \delta t} \]
The two controls inside the hypercube $\mathcal{C}^N$ which are furthest away from each other stand at opposite corners, so that we have  
\[ ||\boldsymbol{u}^* - \boldsymbol{u}^*_{\perp}||_1 \leq ||\boldsymbol{u}_{\mathrm{max}} - (-\boldsymbol{u}_{\mathrm{max}})||_1 = 2{u_{\mathrm{max}}}N,\]
which plugged back in the previous equation with $T=N\delta t$ gives the result:
\begin{equation}
    T^{\perp}_{QSL} \geq \frac{1}{\omega_{\mathrm{max}}u_{\mathrm{max}}} = \frac{1}{u_{\mathrm{max}}|\lambda_{\mathrm{max}} - \lambda_{\mathrm{min}}|},
\end{equation}
where we remind that $\omega_{\mathrm{max}} = |\lambda_{\mathrm{max}}-\lambda_{\mathrm{min}}|$ is the maximum transition angular frequency in the control Hamiltonian ($\hbar=1$). 
Note that since $T^{\perp}_{QSL}$ becomes vanishingly small as the bounds on the controls $u_{\mathrm{max}}$ are enlarged,
this lower bound for the Quantum Speed Limit is clearly not tight in the rather common cases in which the amplitude of the drift Hamiltonian $\hat{H}_d$ constitutes the bottleneck in time for achieving controllability.

\subsection{Trap separation and trap density}
\label{trap_density}
The notion of a minimum distance between distinguishable landscape points can also be used to infer the properties of the landscape around the local minima $\boldsymbol{u}^*$, which we define as global minima inside a neighbourhood
\begin{equation} \mathcal{M} = \{ \boldsymbol{u}^* \in \mathcal{C}^N |\ || \boldsymbol{u} - \boldsymbol{u}^*||_1 \leq \epsilon \Rightarrow J(\boldsymbol{u}) \geq J(\boldsymbol{u}^*) \}
\label{eq:minima}
\end{equation}
for some given $\epsilon >0$.

As a visualization, we can picture the quantum dynamical landscape $J(\boldsymbol{u})$ as a topographic landscape with mountains and valleys, filled up with water up until a certain level $\bar{J}$ which is the same across the region we are considering. 
The water will split up into disconnected water pockets forming a collection of lakes, as pictured in Fig.~\ref{fig:mountains_lakes}. We can interpret the water level as the current infidelity value and the lakes would be the control regions we have to explore to find better controls.
We define the depth of each lake as the difference in height between the water level and the lowest point of the landscape inside the lake (i.e., the local minimum).
\begin{figure}
    \centering
    \includegraphics[width=0.95\linewidth]{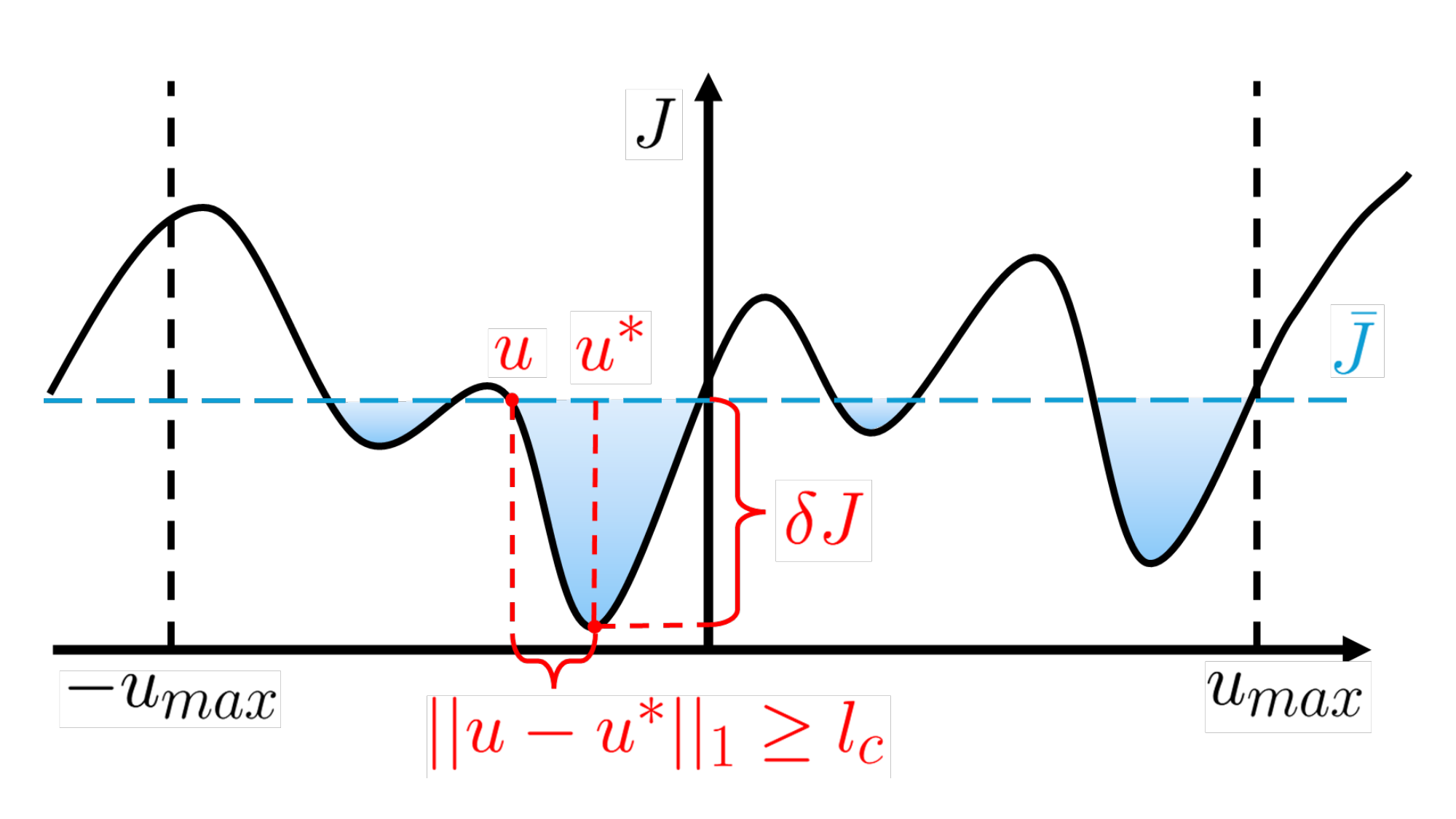}
    \caption{The quantum dynamical landscape $J$ as a topographical landscape filled with water up to $\bar{J}$. Climbing out of the water from a minimum of depth $\delta J$ requires moving in control space over a distance of at least $l_c$. Even though the size of a lake $l_c$ diverges with the control space dimensionality $N$, the underwater control volume fraction shrinks to zero (see Sec.~\ref{trap_density}).}
    \label{fig:mountains_lakes}
\end{figure}
But then, we can use the result on Lipschitz continuity to relate the extent of the lakes to their depth, i.e., the improvement in infidelity that can be achieved by exploring it.
By noticing that the local minima as we defined them are also critical points $\grad J(\boldsymbol{u}^*)=\boldsymbol{0}$, we can use the second version of the inequality in Lemma \ref{lemma:lipschitz}, that is
\begin{equation} 
||\boldsymbol{u} - \boldsymbol{u}^{\star}||_1 \geq \sqrt{\frac{2\delta J}{(\omega_{\mathrm{max}} \delta t)^2}} = \sqrt{2 \delta J}\frac{N}{L},
\end{equation}
which is a stronger condition than Eq.~\eqref{eq:min_dist} when $2J < 1$, and where as before $L=T\omega_{\mathrm{max}}$. We can then define a new control length scale $l_c$
\begin{equation}
\label{eq:neigh_min}
l_c = \beta \frac{N}{L}
\end{equation}
where $ \beta = \max\{\sqrt{2 \delta J}, 2 \delta J \}$.
Then, given a lake of depth $\delta J$, starting from its shore we need to travel a distance (measured using the taxicab norm $||\cdot||_1$) of at least $l_c$ in order to get to the bottom of the lake $\boldsymbol{u}^*$ (the local minimum). In other terms, Eq.~\eqref{eq:neigh_min} implies that the connected neighbourhood of the local minimum where $|J(\boldsymbol{u}) - J(\boldsymbol{u}^*)| < \delta J$ (corresponding to the surface of the lake in our picture) contains a taxicab ball $\mathcal{T}^{(N)}_{l_c}(\boldsymbol{u}^{\star})$ of radius $l_c$ centered around $\boldsymbol{u}^{\star}$:
\[ \mathcal{T}^{(N)}_{l_c}(\boldsymbol{u}^{\star}) = \{ \boldsymbol{u} \in \mathbb{R}^N : ||\boldsymbol{u} - \boldsymbol{u}^{\star}||_{1} \leq l_c \}.\]
This means that to every local minimum of depth $\delta J$ is associated a volume in control space of at least 
\[ \Vol\mathcal{T}^N_{l_c} = \frac{(2l_c)^N}{N!} = \left( \frac{2\beta}{L} \right)^N \frac{N^N}{N!}. \]
As we cannot fit more than $(\Vol\mathcal{T}^N_{l_c})^{-1}$ of these balls in the unit volume, this gives us an upper bound to their average density within the bulk volume $(2u_{max})^N$, this way neglecting surface effects determined by finite bounds on the control region of interest.
We can study the large $N$ behaviour of the density bound by means of the Stirling formula:
\begin{equation}
(\Vol\mathcal{T}^N_{l_c})^{-1} = \exp(-N \log(e \frac{2\beta}{L}) + \mathcal{O}(\log N)), 
\label{eq:trapdensity}
\end{equation}
so that the maximum density is exponentially suppressed if $2e\beta > L$ and grows exponentially for $2e\beta < L$
while the marginal case $2e\beta L^{-1} = 1$ requires further inspection of the remainder $\mathcal{O}(\log N)$.
On one hand, this suggests a way to bound the number of minima in a given volume but on the other hand, since the radius $l_c$ increases as $N$ is scaled up, Eq.~\eqref{eq:trapdensity} is only reliable for unbounded problems $u_{\mathrm{max}}=\infty$, which is a typical assumption in VQA settings
\footnote{In order to get around this difficulty, we could instead bound the number of approximate isolated critical points (without the need of specifying if they are minima or their depth as we defined it) by making use of the Taylor representation of the landscape, which we discussed in Sec.~\ref{taylor}. This problem is equivalent to finding the isolated zeros of a system of multivariate polynomials, whose number is bounded by Bezout's theorem and generalizations thereof \cite{Bernshtein75}. This would potentially be informative to bound the complexity of the landscape in situations when a large number of suboptimal solutions are expected, which typically happens for critically constrained systems \cite{Bukov18, Bukov19}.}.

\subsection{Landscape ruggedness}
Another relevant property related to landscape minima is the ruggedness, which instead is a local measure of their sharpness. It can be defined as the average diagonal element of the Hessian evaluated at the minima \cite{Dalgaard22}
\[ \rho = \frac{1}{N\#\mathcal{M}} \sum_{\boldsymbol{u} \in \mathcal{M}} \sum_{\nu=1}^{N} \partial^2_{\nu}J(\boldsymbol{u}), \]
where $\mathcal M$ is defined by Eq.~\eqref{eq:minima}. We can readily give an upper bound for this quantity by using the derivatives bound from Lemma \ref{lemma:derivatives} for $P=2$:
\[ |\rho| \leq \frac{1}{N\#\mathcal{M}} \sum_{\boldsymbol{u} \in \mathcal{M}} \sum_{\nu=1}^{N} |\partial^2_{\nu}J(\boldsymbol{u})| \leq \frac{(\omega_{\mathrm{max}} \delta t )^2}{2}. \]
It is worth noting that for fixed total time $T$, $\rho$ vanishes as $N$ increases.
Since the Hessian $\boldsymbol{H}$ is the leading contribution to the Taylor expansion of the landscape near a local minimum $\grad{J}(\boldsymbol{u}_0)=\boldsymbol{0}$, a large ruggedness implies a low robustness of the solution with respect to small perturbations.
Using again Lemma \ref{lemma:lipschitz}, we can upper bound this error. Given a small control deviation $||\boldsymbol{u}-\boldsymbol{u}^{\star}||_1\omega_{\mathrm{max}} \delta t \ll 1$, the stronger version of the inequality for critical points gives rise to
\[ |J(\boldsymbol{u}) - J(\boldsymbol{u}^{\star})| \leq  \frac{(  \omega_{\mathrm{max}}\delta t)^2}{2} ||\boldsymbol{u}-\boldsymbol{u}^{\star}||^2_{1},\]
which, unlike $\rho$, allows to quantitatively relate control errors  to fidelity variations.

\subsection{Barren plateaux}
Shifting the focus away from local minima, we can now consider a metric which is often studied in the context of landscape optimization (especially for VQA applications), namely the variance of the gradient over the controls, defined as follows
\[ \Var_{\boldsymbol{u} \in \mathcal{C}^N}[ \partial_{\nu} J(\boldsymbol{u}) ]= \mathbb{E}_{\boldsymbol{u} \in \mathcal{C}^N}[ |\partial_{\nu} J|^2 ] - |\mathbb{E}_{\boldsymbol{u} \in \mathcal{C}^N}[ \partial_{\nu} J ]|^2 \]
where expectation values are integral averages evaluated over the control hypercube $\mathcal{C}^N$ via
\[\mathbb{E}_{\boldsymbol{u} \in \mathcal{C}^N}[ g ] = \left(\frac{1}{2{u_{\mathrm{max}}}}\right)^N \int_{\mathcal{C}^N} c^N\boldsymbol{u}\ g(\boldsymbol{u}). \]
An exponential suppression of the latter with respect to some resource (number of qubits, circuit depth, etc.) is customarily referred to as the problem of barren plateaux (or, alternatively, of vanishing gradients). 
In the literature these averages are often computed over unbounded controls $\boldsymbol{u} \in \mathbb{R}^N$, which here corresponds to taking the limit $u_{\mathrm{max}} \to \infty$. Although this assumption simplifies the derivations, in practice the available controls are always bounded by experimental constraints. For this reason we try here to prove as much as possible for the case of bounded controls.

The Lie-Fourier representation of the landscape allows us to directly relate the variances of any derivative of $J_n$ to the representation coefficients. This way we can prove upper bounds which are also valid for the true landscape $J$:
\begin{lemma}[Variance over bounded controls]
\label{lemma:variance_bounded}
\begin{widetext}
For any integer $n \geq 1$, $P\geq 0$ and $1 \leq \nu_1,\dots,\nu_P \leq N$ we have:
\begin{multline*} 
\Var_{\boldsymbol{u} \in \mathcal{C}^N} \left[ \left( \prod_{p=1}^P \partial_{\nu_p} \right) J_n \right] = \delta t^{2P} \sum_{\mathclap{\boldsymbol{\omega}, \boldsymbol{\omega'} \in (\mathcal{S}^{\Delta}_n)^N\backslash \{ \boldsymbol{0}\}}} c^*_{\boldsymbol{\omega}}c_{\boldsymbol{\omega'}}(\prod_{p=1}^P \omega_{\nu_p}\omega_{\nu_p}') (\tilde{\kappa}(\boldsymbol{\omega} - \boldsymbol{\omega'}) - \tilde{\kappa}(\boldsymbol{\omega})\tilde{\kappa}(\boldsymbol{\omega'}) ),\ \ \tilde{\kappa}(\boldsymbol{\omega}) = \prod_{\nu =1}^{N} \frac{\sin{\delta t \omega_{\nu} u_{\mathrm{max}}}}{\delta t \omega_{\nu} u_{\mathrm{max}}}
\end{multline*}
Moreover, the following upper bounds hold:
\[
\Var_{\boldsymbol{u} \in \mathcal{C}^N}\left[\left(\prod_{p=1}^P \partial_{\nu_p}\right) J_n,J\right] 
\leq  \min \left[ \frac{L^{2P}}{4N^{2P}} , \frac{(L^{P+1} u_{\mathrm{max}})^2}{3N^{2P+1}} + o\left(\frac{L^{2P+2}}{N^{2P+1}}\right) \right] .
\]
\end{widetext}
\end{lemma}
\begin{proof}
See App.~\ref{appendix_c} for a proof valid in the case of multiple controls and a generic observable $\hat{O}$. 
\end{proof} 

It is interesting to note that, similarly to what was already observed in Sec.~\ref{fourier_regression}, the variance over the hypercube $\mathcal{C}^N$ could be expressed by means of a quadratic form of the Lie-Fourier coefficients featuring a modified version of the sinc kernel $\tilde{\kappa}(\boldsymbol{\omega})$ where the bandwidth $\delta t u_{\mathrm{max}}$ is now set by the size of the integration region. Since the formula holds in general for any set of frequencies and coefficients, it also proves that this quadratic form is semi-positive definite.

The upper bound in Lemma \ref{lemma:variance_bounded} shows that the variance for bounded controls increases with the time-energy budget $L$ and it decreases with the number of controls or circuit layers $N$, the dependence being polynomial in both cases, while in the same limit being exponentially suppressed with the order of the derivatives $P$. Since these upper bounds are independent on the details of the spectrum (other than $\omega_{\mathrm{max}}$), they can be understood as constraints on the best-case scenario for barren plateaux across quantum systems as set by time, depth and bandwidth limitation for their Lie-Fourier expansion.

The upper bound for the variance implies that even for $P=0$ (which corresponds to not taking any derivatives) this quantity goes to zero for large $N$ if $L=T \omega_{\mathrm{max}}$ is kept constant. As a consequence, the landscape becomes flatter and the control region $\mathcal{R}_{\overline{\delta J}}(N)$ where $J$ is smaller than its average by more than a finite amount $\overline{\delta J} > 0$, which we define as
\[ \mathcal{R}_{\overline{\delta J}}(N) = \{\boldsymbol{u} \in \mathcal{C}^N\ \text{s.t.}\ \mathbb{E}_{\boldsymbol{u} \in \mathcal{C}^N}[J] - J(\boldsymbol{u}) > \overline{\delta J} \}, \]
shrinks to zero in volume compared to the control hypercube. In fact we have that:
\begin{multline*}
    \Var_{\boldsymbol{u} \in \mathcal{C}^N} J =  \int_{\mathcal{C}^N}  \frac{d^N \boldsymbol{u}}{\Vol \mathcal{C}^N} |J(\boldsymbol{u}) - \mathbb{E}_{\boldsymbol{u} \in \mathcal{C}^N}[J]|^2 \geq \\
     \int_{\mathcal{C}^N}  \frac{d^N \boldsymbol{u}}{\Vol \mathcal{C}^N} |J(\boldsymbol{u}) - \mathbb{E}_{\boldsymbol{u} \in \mathcal{C}^N}[J]|^2 \Theta(\mathbb{E}_{\boldsymbol{u} \in \mathcal{C}^N}[J] - J(\boldsymbol{u}) - \overline{\delta J}) \\
    \geq  \frac{(\overline{\delta J})^2}{\Vol \mathcal{C}^N} \int_{\mathcal{C}^N}  d^N \boldsymbol{u}\  \Theta(\mathbb{E}_{\boldsymbol{u} \in \mathcal{C}^N}[J] - J(\boldsymbol{u}) - \overline{\delta J}) \\
    = (\overline{\delta J})^2 \frac{\Vol \mathcal{R}_{\overline{\delta J}}(N)}{\Vol \mathcal{C}^N}.
\end{multline*}
But then by making use of the asymptotic upper bound from Lemma~\ref{lemma:variance_bounded} for $P=0$, we conclude that
\[ \frac{\Vol \mathcal{R}_{\overline{\delta J}}(N)}{\Vol \mathcal{C}^N} \leq \mathcal{O}\left(\frac{1}{N}\right) \xrightarrow[]{N \to \infty}0.\]
This fact is closely related to the discussion in Sec.~\ref{fourier_regression} and also clarifies what happens to the volume surrounding the minima that we discussed previously (the overall lake surface). In fact, if we choose the water level to be below the landscape average $\bar{J} < \mathbb{E}_{\boldsymbol{u} \in \mathcal{C}^N}[J]$ and fix $\overline{\delta J} = \mathbb{E}_{\boldsymbol{u} \in \mathcal{C}^N}[J] - \bar{J}$, then $\mathcal{R}_{\overline{\delta J}}(N)$ is the part of the control landscape underwater. The fact that the volume fraction associated to this control region shrinks to zero as $N \to \infty$ could appear somewhat surprising, given that the minimum distance $l_c \propto N$ we have to travel to get from shore to bottom diverges. But the two facts are not necessarily mutually exclusive, as the balls $\mathcal{T}^N_{l_c}$ will eventually not completely fit inside $\mathcal{C}^N$ as $N$ increases, and the changes in volume determined by the growth in dimensionality are hard to picture intuitively. 

Unfortunately we cannot extend all the results for the true landscape $J$ from Lemma \ref{lemma:variance_bounded} to unbounded controls by just taking the limit $u_{\mathrm{max}} \to \infty$, because the limits in the Lie product expansion order $n$ and control region size $u_{\mathrm{max}}$ cannot be exchanged. 
Instead, if we work with a finite $n$, i.e. in the typical VQA/PQC setting, where the circuit has finite depth and therefore the Fourier representation is finite and exact, then that is possible. As noticed before, the $n=1$ case correspond physically to an interleaved circuit, and the frequencies in the Fourier representation $\mathcal{S}^{\Delta}_1=\mathcal{S}^{\Delta}$ correspond only to the differences in control Hamiltonian eigenvalues. Keeping this distinction in mind, we can prove similar results for unbounded controls.
\begin{lemma}[Variance over unbounded controls]
\label{lemma:variance_unbounded}
For any integer $n \geq 1$, $P\geq 0$ and $1 \leq \nu_1,\dots,\nu_P \leq N$ we have:
\[
\Var_{\boldsymbol{u} \in \mathbb{R}^{N}}\left[ \left( \prod_{p=1}^P \partial_{\nu_p} \right) J_n \right] = \delta t^{2P} \sum_{\mathclap{\boldsymbol{\omega} \in (\mathcal{S}^{\Delta}_n)^N\backslash \{ \boldsymbol{0}\}} } |c_{\boldsymbol{\omega}}|^2 \prod_{p=1}^P \omega_{\nu_p}^2
\]
Moreover, the following bounds hold:
\begin{multline*} 
\Var_{\boldsymbol{u} \in \mathbb{R}^{N}}\left[\left(\prod_{p=1}^P \partial_{\nu_p}\right) J_n, J\right] 
\leq \left( \frac{(\delta t \omega_{\mathrm{max}})^P}{2} \right)^2, \\
\sum_{\nu_1, \dots, \nu_P} \Var_{\boldsymbol{u} \in \mathbb{R}^{N}} \left[ \left( \prod_{p=1}^P \partial_{\nu_p} \right) J_n \right] \geq \frac{\Delta J^2 \delta t^{2P}}{4\sum_{\boldsymbol{\omega} \in (\mathcal{S}^{\Delta}_n)^N\backslash \{ \boldsymbol{0} \}} \frac{1}{||\boldsymbol{\omega}||^{2P}_2}} 
\end{multline*}
where we defined $\Delta J$ as the maximum variation in $J_n$,
\[\Delta J := \sup_{\boldsymbol{u} \in \mathbb{R}^{N}} J_n(\boldsymbol{u}) -\inf_{\boldsymbol{u} \in \mathbb{R}^{N}} J_n(\boldsymbol{u}).\]
\end{lemma}
\begin{proof}
See App.~\ref{appendix_c} for a proof in the case of multiple controls and a generic observable $\hat{O}$. 
\end{proof}
While the upper bound holds for both the Lie-Fourier representations $J_n$ and the landscape $J$ itself, the lower bound only holds for finite $n$. The latter can be derived thanks to the simpler structure of the kernel in the case of unbounded controls, but we do not rule out that it might be possible to derive an equivalent result for bounded controls. 
More specifically, the case $P=1$ (i.e. the gradient variance) gives us a worst-case scenario estimate for the barren plateaux even across different systems and different initial and target states (or observables). We can derive a simplified, looser bound in terms of the smallest non-zero frequency $\omega_{\mathrm{min}} \in \mathcal{S}^{\Delta}_n$, which assures $\forall \boldsymbol{\omega} \in (\mathcal{S}^{\Delta}_n)^N\backslash \{ \boldsymbol{0} \},\ ||\boldsymbol{\omega}||_2^2 \geq \omega^2_{\mathrm{min}}$, so that 
\begin{align*} 
\sum_{\nu=1}^N \Var_{\boldsymbol{u} \in \mathbb{R}^{N}} \left[ \partial_{\nu}  J_n \right] \geq& \frac{\Delta J^2 \delta t^{2}}{4\sum_{\boldsymbol{\omega} \in (\mathcal{S}^{\Delta}_n)^N\backslash \{ \boldsymbol{0} \}} \frac{1}{||\boldsymbol{\omega}||^{2}_2}} \\
\geq& \frac{\Delta J^2 \delta t^{2} \omega_{\mathrm{min}}^2}{4(\#\mathcal{S}^{\Delta}_n)^N}, 
\end{align*}
which recovers the exponential suppression in circuit depth $N$ that is typical in the case of barren plateaux. 
In the case $n=1$ of an interleaved circuit, $\mathcal{S}^{\Delta}_1=\mathcal{S}^{\Delta}$ can be of order $\mathcal{O}(D^2=2^{2Q})$ if $\hat{H}_c$ is non-degenerate, but as we saw in Sec.~\ref{numerics} it can be much smaller in physically relevant systems thanks to degeneracies. In the case of the $Q$-qubits Ising model (with or without longitudinal field), we have $\mathcal{S}^{\Delta} \sim \mathcal{O}(Q^2)$, so that even in the worst case scenario the gradient variance is only exponentially suppressed in the circuit depth $N$, but not in the number of qubits $Q$. This is in contrast to the case of a non-degenerate $\hat{H}_c$, where this analysis does not rule out the existence of fixed depth barren plateux when $Q$ increases. As a conclusive remark, we note that the dependence of gradient variance suppression on the growth of the dynamical Lie algebra \cite{Larocca22} implies that the latter must play a decisive role in the distribution of the Fourier coefficients, deciding how far away the system will be from this bound.
\section{Applications to optimizer design} \label{applications}
The considerations we made in the previous Section are also relevant for the design of optimization algorithms. First of all, the length scale $l$ we derived from Lipschitz continuity constrains how local the search needs to be.  

If we want to obtain an improvement of at least $\delta J$ to our current best control $\boldsymbol{u}$, Eq.~\eqref{eq:min_dist} assures us that such an improvement cannot be found within distance $l$. This is especially relevant to in-situ optimization, such as black-box optimizers like DCRAB \cite{caneva2011chopped,Rach15} and SOMA \cite{preti2022continuous}. For example, taking into account shot noise or other uncertainties, it may be meaningless to try to pick points closer than $l$. For instances of DCRAB based on a simplex search, the length scale can specifically be used to constrain the size of the simplex, for example at the start of the optimization. Alternatively, given the finite resolution of arbitrary waveform generators, one can constrain the minimum vertical resolution so as to attain a given change in the cost function landscape.

Going further still, one can consider a general parameterization of the controls given by a linear transformation of the piece-wise constant basis as prescribed by Ref.~\cite{motzoi2011optimal}, such that 
\[ u(t) = \sum_{i=1}^{N_c}r_i(t) v_i. \]
As we have seen in Corollary \ref{corollary:parametrization}, the derivative properties that we proved for the landscape $J(\boldsymbol{u})$ are still true for the parametrized landscape $\tilde{J}(\boldsymbol{v})$ provided that the basis functions are properly normalized (which we are always free to achieve). Hence, the notion of the length scale $l$ also applies to the new control space.

The convergence rate of the optimizers can also be affected by characterizing landscape properties, especially for model-based control. Gradient optimizers \cite{Khaneja05, DeFouquieres13} in particular will have to estimate step sizes which can be aided by notions of minimum distance as above, as well as bounds on first and second derivatives given by Lemma~\ref{lemma:derivatives}. Similarly quantum Newton search \cite{dalgaard2020hessian} and quasi-Newton search \cite{DeFouquieres13,Goodwin16} can also benefit from bounds on the Hessian. Likewise, the minimum length scale can be helpful to inform stopping conditions, especially in the presence of shot noise or decoherence.

The nature of the landscape is also especially important for sampling based methods such as those using machine learning, e.g.~the aforementioned meta-learning optimizer SOMA \cite{preti2022continuous}, since it can help to set a minimum distance between search points. This minimum distance is especially relevant to discrete point sampling, e.g. within a Monte-Carlo tree search, as used in global reinforcement learning \cite{dalgaard2020global}, or in grid search methods such as DIRECT \cite{Jones93}. In fact, the Lipschitz continuity of the landscape calls for a wider adoption of global optimizers that are specific to this class of functions, improving convergence with respect to more naive approaches \cite{Shubert72, Jones93, Nicholas15, Liu15}.

As an illustrative example, we can apply the upper bound for the Lipschitz constant to the DIRECT algorithm $\cite{Jones93}$. Appropriate bounds for extensions of this algorithm \cite{Liu15, Nicholas15, Kokail19} can be derived in a similar way. Since DIRECT samples $J$ on a rectangular grid, we can assume $\boldsymbol{u'} -\boldsymbol{u} = h \boldsymbol{e}_{\nu}$ for some canonical basis vector $\boldsymbol{e}_{\nu}$ and $h > 0$ \footnote{In the original paper the quantity $h$ is called $\delta$}. But then we have the following result
\[ |J(\boldsymbol{u}) - J(\boldsymbol{u'})| \leq \frac{\omega_{\mathrm{max}} \delta t}{2} ||\boldsymbol{u'} -\boldsymbol{u}||_1 \leq \frac{\omega_{\mathrm{max}} \delta t h}{2}. \]
This means that we can use the upper bound $K \leq \omega_{\mathrm{max}} \delta t/2$ within DIRECT as an additional condition during the selection of potentially optimal hyperrectangles to decrease the number of function evaluations.
\section{Conclusions}\label{conclusions}
In this paper we derived from first principles and under very general assumptions the main properties of three different feature map representations of a quantum dynamical landscape. The latter is the expectation value of an observable over the output state of a generalized Parametrized Quantum Circuit, which is equivalent to a controlled quantum system with stepwise-constant controls, and generalizes Variational Quantum Algorithm circuits. The feature map representations are approximations of the landscape, which we use to study the properties of the landscape itself, and as  physics-informed models for supervised learning, with the aim of informing the development of methods for quantum cost function optimization.

First, we obtained a Fourier representation \cite{Schuld21} by means of a Lie-Trotter approximation of the dynamics. We showed that the resulting frequency spectrum fills up densely a finite hypercube, with size given by the maximum transition frequency in the control Hamiltonian. We proved analytically some important properties of the representation coefficients which are related to boundedness and discrete symmetries, which are useful for the further development of simulation and learning algorithms \cite{Rudolph23}. 
We found numerically in the case of the Ising Model that owing to the high degree of symmetry of the Hamiltonian and initial and target states, the spectrum is a stepwise-continuous function, opening up the possibility of further compression by using e.g. polynomial bases in frequency space.
These numerical results were obtained thanks to an algorithm that can be applied to any model with equally spaced control eigenvalues. This allows further landscape exploration of Pauli controlled systems, which are a common Ansatz for ease of experimental implementation and theoretical investigation \cite{motzoi17, Rudolph23, Kokail19, Koczor22QAD}.

We then showed that the bounds in absolute value and bandwidth of the landscape cause it to be a Lipschitz continuous function. This means that there is a global maximum ratio between change in cost function and traveled length in control space.
We related this property to a local upper bound on the error of a Taylor expansion, which is therefore an efficient representation when the overall time-energy budget is limited. Local models are often useful for optimization, and this result can inform further study in this regard \cite{Goodwin16, Koczor22QAD}.

Since the dense Fourier spectrum represents a challenge for supervised learning with a finite dimensional feature space, we derived analytically an equivalent kernel regression problem that gives rise to the sinc kernel representation. Here, the landscape is represented as a linear superposition of finite bandwidth kernels in a way that generalizes the Green's function expansion from linear optics. In numerical benchmarks against the other two representations, the sinc kernel showed to be more efficient if the training data set is large enough, while its inferior performance on small data sets can be greatly improved by reducing the kernel bandwidth. We also commented on the limits of random sampling and of the $L_2$ distance to distinguish different landscapes, arguing that stronger notions of distance, such as $L_{\infty}$, or optimization-driven sampling strategies \cite{Dalgaard22, Beato24,dalgaard2020global} are necessary to obtain a meaningful regression problem when the number of time steps $N$ is not fixed or is too large.

We then discussed the consequences of our findings on some landscape metrics which are relevant for optimization. In particular, we showed that the aforementioned Lipschitz constant implies a minimum granted robustness of local optima, together with a relation between the volume of an unexplored region in control space and the improvement in the cost function value we can obtain by exploring that region. We also related the Fourier representation to the variance of the landscape and its derivatives over the controls. This allowed us to prove a set of upper bounds that constrain the best case scenario for this quantities, purely based on time, energy and depth limitations and largely independently on the specific Hamiltonian. This way, we showed that for constant final time $T$ and bounded controls, the non-trivial regions of the landscape have a vanishingly small volume compared to a fixed control region of interest as the number of control parameters increases. 
Overall, owing to the generality of the properties and the bounds we found, we deem them to be useful as universal baselines for the considered metrics in more specific cases. 

Finally, we explored how these landscape properties inform the design and tuning of optimizers. In this sense, the generalization of our proofs to parametrized landscapes is instrumental, allowing for the application to popular Quantum Optimal Control algorithms such as GRAPE \cite{Khaneja05, DeFouquieres13}, DCRAB \cite{Rach15}, Krotov \cite{goerz2019krotov}, SOMA \cite{preti2022continuous}, and SPINACH \cite{Goodwin16}, where our results inform the choice of stopping conditions and the estimate of convergence ratios. They also provide estimates for appropriate stepping sizes, bounds on the Hessian, and inform the choice of vertical resolution in the controls, which can all greatly benefit hyperparameter choices in such algorithms. We also argued in favour of Lipschitz-aware optimizers such as DIRECT \cite{Jones93, Kokail19}, for which we derive a more specific Lipschitz constant. The information about distance between sampled points is especially relevant to global optimizers and sampling based strategies such as machine learning, for example dictating the branching in Monte-Carlo tree searches \cite{dalgaard2020global}. That is, this can be used to decrease the number of calls to the quantum circuit by discarding sampling in areas which cannot contain better quantum cost function values because of Lipschitz continuity.

\begin{acknowledgements}
We thank Marin Bukov, Nicol\`o Beato and Mogens Dalgaard for all the insightful conversations, and especially Mogens Dalgaard for sharing his code for the simulation and optimal control of spin chains. We also thank Jan Reuter, Phila Rembold, Robert Zeier, Matthias M\"uller, Matteo Rizzi, Markus Schmitt, Michael Schilling and Francesco Preti for related discussions.
This work was funded by Horizon Europe programme HORIZON-CL4-2022-QUANTUM-02-SGA via
the project 101113690 (PASQuanS2.1), programme HORIZON-CL4-2022-QUANTUM-01-SGA via project 101113946 (OpenSuperQPlus100), and by Horizon Europe programme (HORIZON-CL4-2021-DIGITALEMERGING-02-10) Grant Agreement 101080085 QCFD, and it was supported from the Jülich Supercomputing Center through the JUWELS and JURECA clusters.
\end{acknowledgements}

\appendix
\section{Symbols, notation and useful formulas}\label{appendix_a}
\subsubsection{Sets and vectors}
We indicate sets with calligraphic capital letters like $\mathcal{H}, \mathcal{F}, \mathcal{D}$. Given a set $\mathcal{F}$, we use the exponential notation $\mathcal{F}^N$ with $N \in \mathbb{N}$ for the $N$-fold cartesian product of the set with itself and the symbol $\# \mathcal{F}$ for its cardinality.

Throughout the paper we often switch between vector and coordinate notation. We use boldface lower case letters for (column) vectors $\boldsymbol{v} \in \mathbb{C}^N$, always implying that $\boldsymbol{v} = (v_1,\dots,v_N)^T$, and boldface upper case letters for matrices $\boldsymbol{A} \in \mathbb{C}^{M\times N}$, with
\[ \boldsymbol{A} = \begin{pmatrix}
    a_{11} & \dots & a_{1N} \\
    \vdots & \ddots & \vdots \\
    a_{M1} &\dots &  a_{MN}
    \end{pmatrix} = 
    \begin{pmatrix}
    \boldsymbol{a}_{1} & \dots & \boldsymbol{a}_{N} \\
    \end{pmatrix} 
    \]
We use vectorized shorthand notations for fixed scalar quantities, e.g.:
\[
\boldsymbol{1} = (1,\dots,1)^T,\ \ \boldsymbol{\omega}_{\mathrm{max}} = (\omega_{\mathrm{max}},\dots,\omega_{\mathrm{max}})^T.
\]
When dealing with sums over a vector index $\boldsymbol{k}\in \mathbb{Z}^N$, where each element $k_i \in \mathbb{Z}$ of the vector spans the integers between $a_i \leq k_i \leq b_i$ we sometimes use the following shorthand notation:
\[ \sum_{\boldsymbol{k}=\boldsymbol{a}}^{\boldsymbol{b}} := \sum_{k_1=a_1}^{b_1} \cdots \sum_{k_N=a_N}^{b_N}.\]
\subsubsection{Norms and balls}
We use the following notation for norms of vectors in $\mathbb{C}^N$:
\begin{multline*}
    ||\boldsymbol{v}||_1 = \sum_{i=1}^N |v_i| \\
    ||\boldsymbol{v}||_2 = \sqrt{\sum_{i=1}^N |v_i|^2} \\
    ||\boldsymbol{v}||_{\infty}= \max_{i=1,\dots,N} |v_i|
\end{multline*}
and equivalent definitions hold for matrices.
We indicate with $\{\boldsymbol{e}_i\}_{i=1,\dots,N}$ the canonical basis. We call $\mathcal{C}^{(N)}_A(\boldsymbol{s}_0)$ the real ball defined by the sup-norm $||\cdot||_{\infty}$ of radius $A$ centered around $\boldsymbol{s}_0$ (i.e. the translated hypercube):
\[ \mathcal{C}^{(N)}_A(\boldsymbol{s}_0) = \{ \boldsymbol{s} \in \mathbb{R}^N : ||\boldsymbol{s} - \boldsymbol{s}_0||_{\infty} \leq A \}\]
and $\mathcal{T}^{(N)}_A(\boldsymbol{s}_0)$ the real ball defined by the $L_1$ (also known as ``taxicab") norm $||\cdot||_{1}$ of radius $A$ centered around $\boldsymbol{s}_0$:
\[ \mathcal{T}^{(N)}_A(\boldsymbol{s}_0) = \{ \boldsymbol{s} \in \mathbb{R}^N : ||\boldsymbol{s} - \boldsymbol{s}_0||_{1} \leq A \}.\]
In both cases, we drop the $\boldsymbol{s}_0$ argument when considering a ball centered at the origin $\boldsymbol{s}_0=\boldsymbol{0}$.

We use the braket notation for quantum states $\ket{\psi}, \ket{\chi}$ and operators $\hat{U}$. Since we work with finite dimensional systems with Hilbert space $\mathcal{H}\simeq \mathbb{C}^D$, we can choose a finite orthonormal basis $\{\boldsymbol{\beta}_i\}_{i=1,\dots,D}$ with $\braket{\beta_i|\beta_j}=\delta_{ij}$, $\sum_{i=1}^d\ket{\beta_i}\bra{\beta_i} = \text{id}(\mathcal{H}) =: \hat{I}$ and represent quantum operators as matrices with (capital letter) entries $U_{ij}$
\[ U_{ij} = \bra{\beta_i}\hat{U}\ket{\beta_j}, \ \hat{U} = \sum_{ij} U_{ij}\ket{\beta_i}\bra{\beta_j}  \]
and quantum states as vectors with entries $\psi_i$
\[ \psi_i = \braket{\beta_i|\psi}, \ \ket{\psi} = \sum_{i=1}^D \psi_i \ket{\beta_i}. \]
When dealing with product Hilbert spaces $\mathcal{H}^{\otimes Q}$, given $\ket{\psi} \in \mathcal{H}$ we use the notation
\[ \ket{{\psi}_Q} := \ket{\psi}^{\otimes Q}. \]
We indicate with 
\[ ||\hat{O}||_{\infty} = \max_{\ket{\psi} \in \mathcal{H}\backslash \{\ket{0}\} } \sqrt{\frac{\braket{\hat{O}\psi|\hat{O}\psi}}{\braket{\psi | \psi}}} \]
the sup-norm of the operator $\hat{O}$. The sup-norm is sub-additive and sub-multiplicative \cite{wiki_opnorm}.
We use the convention $\hbar=1$ throughout the text.

\subsubsection{Linear expansion of exponentials}
Since we are going to use this result inside several proofs, we show here a formula concerning the linear expansion of the exponential. Given a (finite dimensional) operator $\hat{X} \in \mathbb{C}^{D \times D}$, we have
\begin{multline*} 
||e^{\hat{X}} - \hat{I}||_{\infty} = ||\sum_{n=1}^{\infty} \frac{\hat{X}^n}{n!}||_{\infty} \leq \sum_{n=1}^{\infty} \frac{||\hat{X}||_{\infty}^n}{n!} = e^{||\hat{X}||_{\infty}}-1 ,\\
||e^{\hat{X}} - \hat{I} -\hat{X}||_{\infty} = ||\sum_{n=2}^{\infty} \frac{\hat{X}^n}{n!}||_{\infty} \leq \\ \leq \sum_{n=2}^{\infty} \frac{||\hat{X}||_{\infty}^n}{n!} = e^{||\hat{X}||_{\infty}}-1 -||\hat{X}||_{\infty} ,
\end{multline*}
where we made use of the continuity, subadditivity and submultiplicativity of the sup-norm $||\cdot||_{\infty}$.
Because of the mean value theorem (see e.g. Theorem 4.3.1 and 4.3.2 in \cite{DeMarco}), we have for $x\geq 0 \in \mathbb{R}, \exists y \in[0,x]$
\[ e^x = 1 + xe^y \]
which implies
\[ e^x - 1 - x = x(e^y -1).\]
By applying again the same theorem to the right hand side we have $\exists z \in[0,y]$ such that
\[e^x - 1 - x = x(ye^z) \leq x^2e^x.\] 
Using these inequalities we can conclude:
\begin{equation} 
\label{eq:exp_ineq1}
||e^{\hat{X}} - \hat{I}||_{\infty} \leq e^{||\hat{X}||_{\infty}}-1  \leq ||\hat{X}||_{\infty}e^{||\hat{X}||_{\infty}},
\end{equation}
\begin{multline}
\label{eq:exp_ineq2}
||e^{\hat{X}} - \hat{I} -\hat{X}||_{\infty} \leq e^{||\hat{X}||_{\infty}}-1 -||\hat{X}||_{\infty} \\
\leq ||\hat{X}||^2_{\infty}e^{||\hat{X}||_{\infty}}.
\end{multline}
\section{Reduction to standard form}\label{appendix_b}
We consider the time evolution of a quantum system of finite dimension $D$ under the time dependent Hamiltonian 
\[ \hat{H}(t) = \hat{H}^{(0)} + \sum_{\mu=1}^M \hat{H}^{(\mu)} u_{\mu}(t)\]
where the controls $u_{\mu}(t)$ are bound to an interval $u_{\mu} \in [u_{\mu}^{\mathrm{min}},u_{\mu}^{\mathrm{max}}]$. In general, only $M'$ out of $M$ control Hamiltonians $\hat{H}^{(\mu)}$
are linearly independent (let us say they are $\mu = 1,\dots,M'$). Then we can write the remaining $M-M'$ terms as a linear combination of the first $M'$ terms
\begin{multline*}
    \hat{H}(t) = \hat{H}^{(0)} + \sum_{\mu'=1}^{M'} \hat{H}^{(\mu')} u_{\mu'}(t) + \sum_{\mu=M'+1}^M \hat{H}^{(\mu)} u_{\mu}(t) \\
    = \hat{H}^{(0)} + \sum_{\mu'=1}^{M'} \hat{H}^{(\mu')} u_{\mu'}(t) + \sum_{\mu=M'+1}^M \sum_{\mu'=1}^{M'} a^{(\mu')}_{\mu} \hat{H}^{(\mu')} u_{\mu}(t) = \\
    = \hat{H}^{(0)} + \sum_{\mu'=1}^{M'} \underbrace{\left( u_{\mu'}(t) + \sum_{\mu=M'+1}^M  a^{(\mu')}_{\mu} u_{\mu}(t)  \right)}_{u'_{\mu'}(t)} \hat{H}^{(\mu')},
\end{multline*}
where the linearly independent controls $u'_{\mu}(t)$ are bound to another interval $u'_{\mu} \in [u_{\mu}^{'min},u_{\mu}^{'max}]$. Since $\hat{H}(t)$ only depends on $u'_{\mu}(t)$, we can suppose without loss of generality that the control Hamiltonians are linearly independent from the start and drop the apex notation.

We will now show that by a further reparametrization we can also fix the energy scales across different controls to a common value $\omega_{\mathrm{max}}$ and (optionally) center the control bounds around the origin. Let us define the maximum transition frequency $\omega^{(\mu)}_{\mathrm{max}}$ in the control Hamiltonian $\hat{H}^{(\mu)}$ as
\[ \omega^{(\mu)}_{\mathrm{max}} = |\lambda^{(\mu)}_{\mathrm{max}}-\lambda^{(\mu)}_{\mathrm{min}}|\]
where $\lambda^{(\mu)}_{\mathrm{max}} (\lambda^{(\mu)}_{\mathrm{min}})$ is the maximum (minimum) eigenvalue of $\hat{H}^{(\mu)}$. We define new controls $\tilde{u}_{\mu}(t)$ as
\[u_{\mu}(t) = m_{\mu} + \frac{\omega_{\mathrm{max}}}{\omega_{\mathrm{max}}^{(\mu)}}\tilde{u}_{\mu}(t),\ \ m_{\mu} = \frac{u_{\mu}^{\mathrm{min}} + u_{\mu}^{\mathrm{max}}}{2}\]
Then we can rewrite the Hamiltonian as
\begin{multline*}
\hat{H}(t) = \hat{H}^{(0)} + \sum_{\mu=1}^{M} \hat{H}^{(\mu)} u_{\mu}(t) = \\ \hat{H}^{(0)} + \sum_{\mu=1}^{M} \hat{H}^{(\mu)} m_{\mu} + \sum_{\mu=1}^{M} \hat{H}^{(\mu)}  \frac{\omega_{\mathrm{max}}}{\omega_{\mathrm{max}}^{(\mu)}}\tilde{u}_{\mu}(t)  = \\ \hat{\tilde{H}}^{(0)} + \sum_{\mu=1}^{M} \hat{\tilde{H}}^{(\mu)} \tilde{u}_{\mu}(t) 
\end{multline*}
so that now we have $\forall \mu\ \tilde{\omega}_{\mathrm{max}}^{(\mu)} = \omega_{\mathrm{max}}$ and the rescaled controls $\tilde{u}_{\mu}(t)$ are bound to the symmetric interval $[-\widetilde{\Delta u}_{\mu}/2,\widetilde{\Delta u}_{\mu}/2]$ given by
\[\widetilde{\Delta u}_{\mu} = \frac{\omega_{\mathrm{max}}^{(\mu)}}{\omega_{\mathrm{max}}}(u_{\mu}^{\mathrm{max}} - u_{\mu}^{\mathrm{min}})\]
Although the price to pay lies in modifying the drift Hamiltonian $\hat{H}^{(0)}$, we notice that loosening the bounds symmetrically does not change $m_{\mu}$ nor the drift term:
\begin{multline*}
 u_{\mu}^{\mathrm{min}} \mapsto u_{\mu}^{\mathrm{min}} - \delta,\ \ u_{\mu}^{\mathrm{max}} \mapsto u_{\mu}^{\mathrm{max}} + \delta,\\
 \widetilde{\Delta u}_{\mu} \mapsto \widetilde{\Delta u}_{\mu} + \frac{\omega_{\mathrm{max}}^{(\mu)}}{\omega_{\mathrm{max}}}2\delta,\ \ m_{\mu} \mapsto m_{\mu}.
\end{multline*}
We can then study the problem as the control bounds are changed for fixed drift by simply changing $\widetilde{\Delta u}_{\mu}$. 
Finally, we notice that adding a term proportional to the identity $\propto \hat{I}$ to the Hamiltonian only changes the overall phase of the time evolution, and therefore does not affect observation values, hence the landscape. But then we can use this freedom to shift the controls $\hat{H}^{(\mu)}$ so that $|\lambda^{(\mu)}_{\mathrm{max}}| = ||\hat{H}^{(\mu)}||_{\infty} = \omega_{\mathrm{max}}/2$.
\section{Multiple controls}\label{appendix_c}
We now generalize the Lie-Fourier representation that was discussed in the main text to the case of multiple controls. More specifically, we consider the following Hamiltonian
\[ \hat{H}(t) = \hat{H}^{(0)} + \sum_{\mu=1}^M \hat{H}^{(\mu)} u_{\mu}(t),\]
where we assume without loss of generality that $u_{\mu}(t) \in [-u_{\mathrm{max}}^{(\mu)}, u_{\mathrm{max}}^{(\mu)}]$ and $\omega^{(\mu)}_{\mathrm{max}} = \omega_{\mathrm{max}}$ (see App.~\ref{appendix_b} for details), while we keep the same time discretization convention as in the single control case. We will call $u_{\mathrm{max}} = \max_{\mu}u^{(\mu)}_{\mathrm{max}}$. When discretizing  the control pulses $u_{\mu \nu} := u_{\mu}(t_{\nu})$, we use the following matrix notation:
\[ \boldsymbol{U} = \begin{pmatrix}
    u_{11} & \dots & u_{1N} \\
    \vdots & \ddots & \vdots \\
    u_{M1} &\dots &  u_{MN}
    \end{pmatrix} = 
    \begin{pmatrix}
    \boldsymbol{u}_{1} & \dots & \boldsymbol{u}_{N} \\
    \end{pmatrix} 
    \]
where now the control vector $\boldsymbol{u}_{\nu} \in \mathbb{R}^M$ represents the value of the $M$ controls at the $\nu$-th timestep.
Correspondingly, the control region of interest $\mathcal{C}$ for the discretized pulse is
\[ \boldsymbol{U} \in \mathcal{C} := \prod_{\mu,\nu=1}^{M,N} [-u_{\mathrm{max}}^{(\mu)}, u_{\mathrm{max}}^{(\mu)}],\ \Vol (\mathcal{C}) = \prod_{\mu,\nu=1}^{M,N} 2 u_{\mathrm{max}}^{(\mu)} \]
We then start by expanding the timestep unitary $\hat{U}(\boldsymbol{u})$, defined by
\[\hat{U}(\boldsymbol{u}) = e^{-i\delta t(H^{(0)} + \sum_{\mu=1}^M u_{\mu} H^{(\mu)})}\]
as a sum of complex exponentials by means of the Lie-Trotter product formula Eq.~\eqref{eq:lie_single} \cite{Lloyd96}. In order to achieve that, we express the control Hamiltonians in their eigenbasis $\hat{H}^{(\mu)} =  \hat{V}^{(\mu)\dagger} \hat{\Lambda}^{(\mu)} \hat{V}^{(\mu)}$, and absorb the change of basis in the terms $\hat{W}$ which do not depend on the controls
\begin{widetext}
\begin{multline*}
 \hat{U}_n(\boldsymbol{u})= (e^{-\frac{i\delta t}{n}\hat{H}^{(0)}} e^{-\frac{i\delta t}{n}u_{1}\hat{H}^{(1)}} \cdots e^{-\frac{i\delta t}{n}u_{M}\hat{H}^{(M)}}  )^n  \\  = \hat{V}^{(M)\dagger} \underbrace{\hat{V}^{(M)} e^{-\frac{i\delta t}{n}\hat{H}^{(0)}}\hat{V}^{(1)\dagger}}_{\hat{W}^{(1)}(n^{-1}\delta t)} e^{-\frac{i\delta t}{n}u_{1} \hat{\Lambda}^{(1)}} \underbrace{\hat{V}^{(1)}\hat{V}^{(2)\dagger}}_{\hat{W}^{(2)}} \cdots \underbrace{\hat{V}^{(M-1)}\hat{V}^{(M)\dagger}}_{\hat{W}^{(M)}}e^{-\frac{i\delta t}{n}u_{M} \hat{\Lambda}^{(M)}} \times \\\times \underbrace{\hat{V}^{(M)}e^{-\frac{i\delta t}{n}\hat{H}^{(0)}}\hat{V}^{(1)\dagger}}_{\hat{W}^{(1)}(n^{-1}\delta t)} e^{-\frac{i\delta t}{n}u_{1} \hat{\Lambda}^{(1)}} \cdots \hat{V}^{(M)}.
\end{multline*}
\end{widetext}
By generalizing the notation from the $M=1$ case we can define the $M$ multiindices $\boldsymbol{J} = (\boldsymbol{j}^{(1)}, \dots, \boldsymbol{j}^{(M)})$ and write
\begin{widetext}
\begin{multline*}
[\hat{U}_n(\boldsymbol{u})]_{ik}=( \sum_{\boldsymbol{j}^{(1)} \dots \boldsymbol{j}^{(M)}l} e^{-i\delta t \sum_{\mu=1}^M u_{\mu}\omega(\boldsymbol{j}^{(\mu)})} V^{(M)\dagger}_{il} W^{(1)}_{l j^{(1)}_{1}}(n^{-1}\delta t)W^{(2)}_{j^{(1)}_1 j^{(2)}_1} \cdots W^{(M)}_{j^{(M-1)}_{1} j^{(M)}_{1}} \cdots \\
\cdots W^{(1)}_{j^{(M)}_{n-1} j^{(1)}_{n}}(n^{-1}\delta t)W^{(2)}_{j^{(1)}_n j^{(2)}_n} \cdots W^{(M)}_{j^{(M-1)}_{n} j^{(M)}_{n}} V_{j^{(M)}_{n} k} ) \\
= \sum_{\boldsymbol{J} \in [D]^{n\times M}} e^{-i\delta t \sum_{\mu=1}^M u_{\mu}\omega(\boldsymbol{j}^{(\mu)})} A^{\boldsymbol{J}}_{ik}(n,\delta t) = \sum_{\boldsymbol{\omega}\in \mathcal{S}_n}  e^{-i\delta t \sum_{\mu=1}^M u_{\mu}\omega_{\mu}} B^{\boldsymbol{\omega}}_{ik}(n,\delta t) = \sum_{\boldsymbol{\omega}\in \mathcal{S}_n}  e^{-i\delta t \boldsymbol{\omega}\cdot \boldsymbol{u}} B^{\boldsymbol{\omega}}_{ik}(n,\delta t),
\end{multline*}
\end{widetext}
where we defined $\boldsymbol{\omega} = (\omega_1,\dots,\omega_M)^T \in \mathcal{S}_n $ and we generalized the symbol $\mathcal{S}_n := \mathcal{S}_n^{(1)}\times\cdots\times\mathcal{S}_n^{(M)}$ compared to the single control case. Anyways the two definitions coincide for $M=1$. 
Now by stacking the $N$ time-step unitaries we can build up the full unitary $\hat{U}_n(\boldsymbol{U})$. We can define the multiindices $\boldsymbol{J}^{(\nu)} = (\boldsymbol{j}^{(1 \nu)}, \dots, \boldsymbol{j}^{(M \nu)})$
\begin{multline*}
\hat{U}_n(\boldsymbol{U}) = \hat{U}_n(\boldsymbol{u}_N) \cdots \hat{U}_n(\boldsymbol{u}_1) \\
=  \sum_{\boldsymbol{J}^{(1)} \cdots \boldsymbol{J}^{(N)}} e^{-i\delta t \sum_{\mu \nu} u_{\mu \nu}\omega(\boldsymbol{j}^{(\mu \nu)})} \hat{A}^{\boldsymbol{J}^{(N)}} \cdots \hat{A}^{\boldsymbol{J}^{(1)}}\\
= \sum_{\boldsymbol{\omega}^{(1)} \dots \boldsymbol{\omega}^{(N)}}  e^{-i\delta t \sum_{\mu \nu} u_{\mu \nu}\omega_{\mu \nu}} \hat{B}^{\boldsymbol{\omega}^{(N)}}\cdots  \hat{B}^{\boldsymbol{\omega}^{(1)}} \\ =  \sum_{\boldsymbol{\Omega}\in \mathcal{S}_n^N}  e^{-i\delta t \text{Tr}(\boldsymbol{\Omega}^T \boldsymbol{U})} \hat{B}^{\boldsymbol{\Omega}}(n,\delta t)
\end{multline*}
where we defined the frequency matrix $\boldsymbol{\Omega}$
\[ \boldsymbol{\Omega} = \begin{pmatrix}
    \omega_{11} & \dots & \omega_{1N} \\
    \vdots & \ddots & \vdots \\
    \omega_{M1} &\dots &  \omega_{MN}
    \end{pmatrix}
\]
which is an element of the set $\mathcal{S}_n^N = \prod_{\mu \nu=1} ^{MN}\mathcal{S}^{(\mu)}_n$. By plugging this expression for the unitary operator back into Eq.~\eqref{eq:fid} we finally find the Lie-Fourier representation $J_n$ of the quantum dynamical landscape $J$ associated to the observable $\hat{O}$
\begin{widetext}
\begin{multline*}
    J_n(\boldsymbol{U}) = \bra{\psi} \hat{U}^{\dagger}(\boldsymbol{U}) \hat{O} \hat{U}(\boldsymbol{U}) \ket{\psi} 
    =  \bra{\psi} \left( \sum_{\boldsymbol{\Omega''}\in \mathcal{S}_n^N}  e^{i\delta t \text{Tr}(\boldsymbol{\Omega''}^T \boldsymbol{U})} \hat{B}^{\boldsymbol{\Omega''} \dagger}(n,\delta t) \right) \hat{O} \left( \sum_{\boldsymbol{\Omega'}\in \mathcal{S}_n^N}  e^{-i\delta t \text{Tr}(\boldsymbol{\Omega'}^T \boldsymbol{U})} \hat{B}^{\boldsymbol{\Omega'}}(n,\delta t) \right) \ket{\psi} = \\
    =  \sum_{\boldsymbol{\Omega'},\boldsymbol{\Omega''}\in \mathcal{S}_n^N} \bra{\psi} \hat{B}^{\boldsymbol{\Omega''} \dagger}(n,\delta t) \hat{O} \hat{B}^{\boldsymbol{\Omega'}}(n,\delta t) \ket{\psi} e^{-i\delta t \text{Tr}(\boldsymbol{(\Omega'-\Omega'')}^T \boldsymbol{U})} =  \sum_{\boldsymbol{\Omega} \in (\mathcal{S}^\Delta_n)^N}  c_{\boldsymbol{\Omega}}(n,\delta t)  e^{-i\delta t \text{Tr}(\boldsymbol{\Omega}^T \boldsymbol{U})} 
\end{multline*}
\end{widetext}
where $(\mathcal{S}^\Delta_n)^N = \prod_{\mu \nu=1}^{MN}\mathcal{S}^{(\mu)\Delta}_n$ and $\mathcal{S}^{(\mu)\Delta}_n$ is again the set of frequency differences within the Fourier spectrum of the timestep unitaries, for each control
\[ \mathcal{S}^{(\mu)\Delta}_n = \{ \omega = \omega' - \omega''\ |\ \omega',\omega'' \in \mathcal{S}^{(\mu)}_n \}, \]
and the coefficients of the expansion $c_{\boldsymbol{\Omega}}$ are given by
\begin{multline*}
c_{\boldsymbol{\Omega}}(n,\delta t) = \sum_{\boldsymbol{\Omega'},\boldsymbol{\Omega''}\in \mathcal{S}_n^N} \delta_{\Omega, \Omega'-\Omega''} \bra{\psi} \hat{B}^{\boldsymbol{\Omega''} \dagger} \hat{O}  \hat{B}^{\boldsymbol{\Omega'}} \ket{\psi}  \\
= \sum_{\boldsymbol{\Omega'}\in \mathcal{S}_n^N} \bra{\psi} \hat{B}^{(\boldsymbol{\Omega'-\Omega}) \dagger} \hat{O}  \hat{B}^{\boldsymbol{\Omega'}} \ket{\psi}
\end{multline*}
\subsection{Proofs valid for multiple controls}
\begin{relemma}{lemma:coefficients_l2}[$L_2$ Boundedness of the coefficients]
    \[ \forall n\ \sum_{\boldsymbol{\Omega} \in (\mathcal{S}^\Delta_n)^N} |c_{\boldsymbol{\Omega}}(n, \delta t)|^2 \leq ||\hat{O}||^2_{\infty} \]
    \label{lemma:coefficients_l2_2}
\end{relemma}
\begin{proof}
We first notice the following:
    \begin{multline*}
    \mathbb{E}_{\boldsymbol{U} \in \mathcal{C}}[ |J_n|^2 ] 
    = \frac{1}{\Vol(\mathcal{C})}\int_{\mathcal{C}} d^{MN}\boldsymbol{U}\  J^*_n(\boldsymbol{U}) J_n(\boldsymbol{U}) =\\
    \sum_{\boldsymbol{\Omega}, \boldsymbol{\Omega'} \in (\mathcal{S}^\Delta_n)^N}  c^*_{\boldsymbol{\Omega}}c_{\boldsymbol{\Omega'}}\prod_{\mu, \nu =1}^{M,N} \int_{-u_{\mathrm{max}}^{(\mu)}}^{u_{\mathrm{max}}^{(\mu)}} \frac{du_{\mu \nu}}{2u_{\mathrm{max}}^{(\mu)}}\  e^{-i\delta t(\omega'_{\mu \nu} - \omega_{\mu \nu}) u_{\mu \nu}}  \\
    = \sum_{\boldsymbol{\Omega}, \boldsymbol{\Omega'} \in (\mathcal{S}^\Delta_n)^N} c^*_{\boldsymbol{\Omega}}c_{\boldsymbol{\Omega'}} \prod_{\mu, \nu =1}^{M,N} \frac{\sin{\delta t (\omega'_{\mu \nu} - \omega_{\mu \nu}) u_{\mathrm{max}}^{(\mu)}}}{\delta t (\omega'_{\mu \nu} - \omega_{\mu \nu}) u_{\mathrm{max}}^{(\mu)}} 
    \end{multline*}
    Now we fix $\forall \mu\ u_{\mathrm{max}}^{(\mu)} = u_{\mathrm{max}}$ and take the following limit:
    \begin{multline*}
    \lim_{u_{\mathrm{max}} \to \infty} \mathbb{E}_{\boldsymbol{U} \in \mathcal{C}_{u_{\mathrm{max}}}}[ |J_n|^2 ] =  \smashoperator[r]{\sum_{\boldsymbol{\Omega}, \boldsymbol{\Omega'} \in (\mathcal{S}^\Delta_n)^N}} c^*_{\boldsymbol{\Omega}}c_{\boldsymbol{\Omega'}} \prod_{\mu, \nu =1}^{M,N} \delta_{\omega_{\mu \nu},\omega'_{\mu \nu}} \\
    = \sum_{\boldsymbol{\Omega} \in (\mathcal{S}^\Delta_n)^N} |c_{\boldsymbol{\Omega}}|^2.
    \end{multline*}
But then we can make use of the following inequality:
    \begin{multline*}
    |J_n(\boldsymbol{U})| = |\bra{\psi} \hat{U}_n^{\dagger}(\boldsymbol{U}) \hat{O} \hat{U}_n(\boldsymbol{U})\ket{\psi}| \\
    \leq ||\hat{U}_n^{\dagger}(\boldsymbol{U}) \hat{O} \hat{U}_n(\boldsymbol{U})||_{\infty} \leq ||\hat{O}||_{\infty}
    \end{multline*}
where we made use of the sub-multiplicativity of the sup-norm and of the fact that for any unitary $||\hat{U}||_{\infty}=1$. 
This implies a bound on the expectation value as
\begin{multline*}
    \mathbb{E}_{\boldsymbol{U} \in \mathcal{C}_{u_{\mathrm{max}}}}[ |J_n|^2 ] 
    \leq \frac{1}{\Vol(\mathcal{C}_{u_{\mathrm{max}}})}\smashoperator[lr]{\int_{\mathcal{C}_{u_{\mathrm{max}}}}} d^{MN}\boldsymbol{U}\  ||\hat{O}||^2_{\infty} = ||\hat{O}||^2_{\infty}.
\end{multline*}
This bound will therefore be also valid in the limit $u_{\mathrm{max}} \xrightarrow{} \infty$, allowing us to conclude.
The result in the main text follows by fixing $M=1$ and noticing that $||\ket{\chi}\bra{\chi}||_{\infty}=1$.
\end{proof}

\begin{lemma}[Uniform convergence over compact sets]
\label{lemma:unif_conv}
$J_n$ converges uniformly to $J$ over $\boldsymbol{U} \in [-u_{\mathrm{max}},u_{\mathrm{max}}]^{M \times N}$ for any $u_{\mathrm{max}} \geq 0$, and the same is true for the derivatives of any order:
\[\sup_{\boldsymbol{U} \in [-u_{\mathrm{max}},u_{\mathrm{max}}]^{M \times N}} \left|\prod_{p=1}^P \partial_{\mu_p \nu_p} (J(\boldsymbol{U})-J_n(\boldsymbol{U}))\right| \xrightarrow[]{n \xrightarrow{} \infty}0 \]
\end{lemma}
\begin{proof}
Let us call $\mathcal{B}_{u_{\mathrm{max}}} = \{z \in \mathbb{C}\ : |z| < u_{\mathrm{max}}\}, \ \bar{\mathcal{B}}_{u_{\mathrm{max}}} = \{z \in \mathbb{C}\ : |z| \leq u_{\mathrm{max}}\}$ and consider the holomorphic extensions of $J_n,J$. It is easy to see that they are given by
\begin{multline*}
    J_n(\boldsymbol{U}) := \bra{\psi} [\hat{U}_n(\boldsymbol{U}^*)]^\dagger \hat{O} \hat{U}_n(\boldsymbol{U})\ket{\psi} \\
    J(\boldsymbol{U}) := \bra{\psi} [\hat{U}(\boldsymbol{U}^*)]^\dagger \hat{O} \hat{U}(\boldsymbol{U}) \ket{\psi}
\end{multline*}
for any $\boldsymbol{U} \in \mathbb{C}^{M\times N}$.
These functions are holomorphic as they are combinations of exponentials, finite sums and products, and they are only functions of $\boldsymbol{U}$ (and not of $\boldsymbol{U}^*$). Moreover, they coincide with the original functions when evaluated for $\boldsymbol{U} \in \mathbb{R}^{M\times N}$.
The fact that the Lie product expansion of $\hat{U}(\boldsymbol{u})$ converges uniformly over $\boldsymbol{u} \in \bar{\mathcal{B}}_{u_{\mathrm{max}}}^{M}$ is a slight modification of the classic result due to Lie. Following Theorem VIII.29 from \cite{ReedBarry72}, we can define 
\[ \hat{S}_n = \exp(-\frac{i\delta t}{n}\hat{H}(\boldsymbol{u})),\ \  \hat{T}_n = e^{-\frac{i\delta t}{n}\hat{H}^{(0)}}\prod_{\mu=1}^Me^{-\frac{i\delta t}{n}u_{\mu}\hat{H}^{(\mu)}}, \]
so that we have $\hat{U}(\boldsymbol{u}) = \hat{S}_n^n, \hat{U}_n(\boldsymbol{u}) = \hat{T}_n^n$ and
\[ ||\hat{U}(\boldsymbol{u}) - \hat{U}_n(\boldsymbol{u})||_{\infty} \leq n ||\hat{S}_n - \hat{T}_n|| e^{\zeta}
\]
with
\[ \zeta := \delta t \left(||\hat{H}^{(0)}||_{\infty} + \frac{u_{\mathrm{max}}M\omega_{\mathrm{max}}}{2}\right), \]
where we fixed $||\hat{H}^{(\mu)}||_{\infty} = \omega_{\mathrm{max}}/2$ as discussed in App.~\ref{appendix_b}.
We can estimate $||\hat{S}_n - \hat{T}_n||$ using the triangular inequality
\[ ||\hat{S}_n - \hat{T}_n|| \leq ||\hat{S}_n - \hat{I} - \frac{i\delta t}{n}\hat{H}(\boldsymbol{u})|| + ||\hat{T}_n - \hat{I} - \frac{i\delta t}{n}\hat{H}(\boldsymbol{u})||\]
and the results on exponential approximations shown in App.~\ref{appendix_a}. Then, a straightforward calculation shows that both terms are upper bounded by a finite sum of terms of order $\mathcal{O}((\zeta n^{-1})^m)$ with $m \geq 2$, so that
\[ ||\hat{U}(\boldsymbol{u}) - \hat{U}_n(\boldsymbol{u})||_{\infty} \sim n  \mathcal{O}\left(\frac{\zeta^2} {n^2}\right) e^{\zeta} \xrightarrow[]{n \to \infty}0. \]
This property trivially extends to finite products, in fact given two parametrized operator sequences $\hat{X}_n(\boldsymbol{u}), \hat{Y}_n(\boldsymbol{u'})$ with $\boldsymbol{u},\boldsymbol{u'}\in \bar{\mathcal{B}}_{u_{\mathrm{max}}}^M$, we have
\begin{multline*} 
||\hat{X}_n(\boldsymbol{u}) \hat{Y}_n(\boldsymbol{u'}) - \hat{X}(\boldsymbol{u})\hat{Y}(\boldsymbol{u'}) ||_{\infty} \leq \\
|| \hat{X}_n(\boldsymbol{u}) ||_{\infty} || \hat{Y}_n(\boldsymbol{u'}) - \hat{Y}(\boldsymbol{u'}) ||_{\infty} + \\
+|| \hat{X}_n(\boldsymbol{u}) - \hat{X}(\boldsymbol{u}) ||_{\infty} || \hat{Y}(\boldsymbol{u'}) ||_{\infty}.
\end{multline*}
But then if the two sequences converge uniformly and are uniformly bounded,
\begin{multline*}
\sup_{\boldsymbol{u}\in \bar{\mathcal{B}}_{u_{\mathrm{max}}}^M} ||\hat{X}_n(\boldsymbol{u}) - \hat{X}(\boldsymbol{u})|| \xrightarrow[]{n \to \infty}0, \\
\sup_{\boldsymbol{u}\in \bar{\mathcal{B}}_{u_{\mathrm{max}}}^M} ||\hat{Y}_n(\boldsymbol{u}) - \hat{Y}(\boldsymbol{u})|| \xrightarrow[]{n \to \infty}0, \\
\sup_{\boldsymbol{u}\in \bar{\mathcal{B}}_{u_{\mathrm{max}}}^M} ||\hat{X}_n(\boldsymbol{u})|| \leq X_{\mathrm{max}}, \sup_{\boldsymbol{u}\in \bar{\mathcal{B}}_{u_{\mathrm{max}}}^M} ||\hat{Y}_n(\boldsymbol{u})|| \leq Y_{\mathrm{max}} 
\end{multline*}
which is satisfied by the timestep unitaries, as we have
\[ ||\hat{U}(\boldsymbol{u})||_{\infty}, ||\hat{U}_n(\boldsymbol{u})||_{\infty} \leq e^{\zeta}, \]
then also the product converges uniformly. But then also the approximants $J_n$ converge uniformly to $J$, in fact 
\begin{multline*}
\sup_{\boldsymbol{U}\in \bar{\mathcal{B}}_{u_{\mathrm{max}}}^{M\times N}}|J_n(\boldsymbol{U})-J(\boldsymbol{U})|=  \\\sup_{\boldsymbol{U}\in \bar{\mathcal{B}}_{u_{\mathrm{max}}}^{M\times N}}|\bra{\psi} [\hat{U}_n(\boldsymbol{U}^*)]^{\dagger} \hat{O} \hat{U}_n(\boldsymbol{U}) - [\hat{U}(\boldsymbol{U}^*)]^{\dagger} \hat{O} \hat{U}(\boldsymbol{U}) \ket{\psi}|  \\ \leq \sup_{\boldsymbol{U}\in \bar{\mathcal{B}}_{u_{\mathrm{max}}}^{M\times N}}|| [\hat{U}_n(\boldsymbol{U}^*)]^{\dagger} \hat{O} \hat{U}_n(\boldsymbol{U}) - [\hat{U}(\boldsymbol{U}^*)]^{\dagger} \hat{O} \hat{U}(\boldsymbol{U})||_{\infty}  \\ \leq
\epsilon_n \xrightarrow[]{n \xrightarrow{} \infty}0.
\end{multline*}
for some infinitesimal sequence $\epsilon_n \geq 0$. Now we just need to show that this notion of convergence extends to the derivatives $\partial^{p_{11}}_{11}\cdots\partial^{p_{MN}}_{MN} J_n(\boldsymbol{U})$.
Since as we have already seen $J_n, J$ are holomorphic functions, we can write the partial derivatives as contour integrals by means of Cauchy's theorem \cite{Smirnov}, from which follows easily the uniform convergence on all compacts of the partial derivatives. In fact, for all $\boldsymbol{U} \in \mathcal{B}_{u_{\mathrm{max}}}$ we have
\begin{widetext}
\begin{multline*}
    |\partial^{p_{11}}_{11}\cdots\partial^{p_{MN}}_{MN} (J - J_n)(\boldsymbol{U})| = \left|\frac{p_{11}!\cdots p_{MN}!}{(2\pi i)^{MN}} \int_{\partial \bar{\mathcal{B}}_{u_{\mathrm{max}}}} du'_{11} \int_{\partial \bar{\mathcal{B}}_{u_{\mathrm{max}}}} du'_{MN} \frac{J(\boldsymbol{U}') - J_n(\boldsymbol{U}')}{(u'_{11}-u_{11})^{p_{11}+1}\cdots (u'_{MN}-u_{MN})^{p_{MN}+1}}\right|  \\
    \leq \frac{p_{11}!\cdots p_{MN}!}{(2\pi)^{MN}} \int_{0}^{2\pi} d\phi'_{11} u_{\mathrm{max}} \int_{0}^{2\pi} d\phi'_{MN} u_{\mathrm{max}} \sup_{\boldsymbol{U}'\in \partial\bar{\mathcal{B}}_{u_{\mathrm{max}}}^{M\times N}}\left(\frac{|J(\boldsymbol{U}') - J_n(\boldsymbol{U}')|}{|u'_{11}-u_{11}|^{p_{11}+1}\cdots |u'_{MN}-u_{MN}|^{p_{MN}+1}}\right) \\
    \leq \epsilon_n \frac{p_{11}!\cdots p_{MN}!}{(2\pi)^{MN}}  \frac{2 \pi u_{\mathrm{max}}}{a_{11}^{p_{11}+1}} \cdots \frac{2 \pi u_{\mathrm{max}}}{a_{MN}^{p_{MN}+1}} \xrightarrow[]{n \xrightarrow{} \infty}0,
\end{multline*}
\end{widetext}
where used the substitution $u'_{\mu \nu} = u_{\mathrm{max}}e^{i\phi'_{\mu \nu}}$ and we defined $a_{\mu \nu} = \inf_{u'\in {\partial\bar{\mathcal{B}}}_{u_{\mathrm{max}}}} |u_{\mu \nu} - u'_{\mu \nu}|>0$ which is strictly larger than zero, as $u_{\mu \nu}$ lies in the open ball $\mathcal{B}_{u_{\mathrm{max}}}$, while $u'_{\mu \nu}$ lies on its boundary $\partial \bar{\mathcal{B}}_{u_{\mathrm{max}}}$, the circle of radius $u_{\mathrm{max}}$.
So we proved uniform convergence also for the derivatives on the sets $\mathcal{B}_{u_{\mathrm{max}}}^{M \times N}$ for any $u_{\mathrm{max}}\geq 0$. But then the same result clearly holds for the restrictions to real controls $\boldsymbol{U} \in [-u_{\mathrm{max}},u_{\mathrm{max}}]^{M\times N} \subset \mathcal{B}_{u_{\mathrm{max}}+\delta}^{M \times N}$ for $\delta >0$, since the complex derivatives reduce in this case to the usual real derivatives.
\end{proof}
We can prove a generalized version of the result presented in the main text concerning the partial derivatives of the landscape $J$.
The result is first shown to be valid for each of the approximants $J_n$ and then uniform convergence is used to conclude that the result is also valid for $J$.
\begin{relemma}{lemma:derivatives}[Boundedness of the derivatives]
\label{lemma:derivatives2}
The derivatives of any order $P \geq 1$ of $J(\boldsymbol{U})$ are bounded by
\[ \left|\prod_{p=1}^P \partial_{\mu_p \nu_p} J(\boldsymbol{U})\right| \leq  \left(\delta t \omega_{\mathrm{max}}\right)^P \frac{\Delta O}{2}\]
where $\Delta O = (\max_{\ket{\psi} \in \mathcal{H}}  - \min_{\ket{\psi} \in \mathcal{H}}) \bra{\psi} \hat{O} \ket{\psi} $.
\end{relemma}
\begin{proof}
Since $J_n$ is the expectation value of $\hat{O}$ on a quantum state generated by the action of a unitary on the initial state, namely
\[ J_n(\boldsymbol{U}) = \bra{\psi}\hat{U}_n^{\dagger}(\boldsymbol{U})\hat{O}\hat{U}_n(\boldsymbol{U})\ket{\psi}\]
we have 
\[ \forall \boldsymbol{U},n\ \  \min_{\ket{\psi} \in \mathcal{H}} \bra{\psi} \hat{O} \ket{\psi} \leq J_n(\boldsymbol{U}) \leq \max_{\ket{\psi} \in \mathcal{H}} \bra{\psi} \hat{O} \ket{\psi}.\]
But then if we define \[\bar{O} = \frac{(\max_{\ket{\psi} \in \mathcal{H}} + \min_{\ket{\psi} \in \mathcal{H}}) \bra{\psi} \hat{O} \ket{\psi}}{2}\]
the function $\tilde{J}_n(\boldsymbol{U}) = J_n(\boldsymbol{U}) - \bar{O}$
has the same derivatives of order $P \geq 1$ as $J_n$ and satisfies $\tilde{J}_n(\boldsymbol{U}) \in [-\Delta O/2,\Delta O/2]$.

For any given $n, \boldsymbol{U}, \mu, \nu$ we can define the restriction
\[ F(u) = \tilde{J}_n(u_{11},\dots,\underbrace{u}_{\mu \nu},\dots,u_{MN}),\]
so that by the definition of partial derivative we get
\[ \frac{d}{du} F(u) = \partial_{\mu \nu}J_n(\boldsymbol{U}).\]
Clearly, since $F$ is a restriction of $\tilde{J}_n$, we also have $F(u) \in [-F_{\mathrm{max}},F_{\mathrm{max}}]$, with $F_{\mathrm{max}}= \Delta O/2$.

By rearranging the order of the frequency sums in the Lie-Fourier representation $J_n$, we can see that $F$ is a linear combination of complex exponentials in the variable $u$:
\begin{widetext}
\begin{equation*} 
F(u) = \left.\sum_{\boldsymbol{\Omega} \in (\mathcal{S}^\Delta_n)^N}  c_{\boldsymbol{\Omega}}(n,\delta t)  e^{-i\delta t \text{Tr}(\boldsymbol{\Omega}^T \boldsymbol{U})}\right\vert_{u_{\mu \nu}=u}= \sum_{\omega_{\mu \nu} \in \mathcal{S}^{(\mu)\Delta}_n} \underbrace{\sum_{\omega_{11} \in \mathcal{S}^{(1)\Delta}_n} \cdots \sum_{ \omega_{MN} \in \mathcal{S}^{(M)\Delta}_n}  c_{\boldsymbol{\Omega}}  e^{-i\delta t \sum_{\alpha \beta \neq \mu \nu}\omega_{\alpha \beta} u_{\alpha \beta}}}_{c'_{\omega_{\mu \nu}}(n,\delta t, u_{\alpha \beta \neq \mu \nu})}  e^{-i\delta t \omega_{\mu \nu} u}
\end{equation*}
\end{widetext}
This implies several properties: first of all $\forall n,\boldsymbol{U},\ F$ is a  holomorphic (or ``regular") function when we extend the variable $u \in \mathbb{C}$ to complex values. $F$ is also bandwidth limited, namely $\omega_{\mu \nu} \in [-\omega_{\mathrm{max}}, \omega_{\mathrm{max}}]$, which implies that $F=\mathcal{O}(e^{a|u|})$ with $a \leq a_{\mathrm{max}} = \delta t \omega_{\mathrm{max}}$:
\begin{multline*}
    |F(u)| = \left|\sum_{\omega \in \mathcal{S}^{(\mu)\Delta}_n} c'_{\omega}  e^{-i\delta t \omega u}\right| \leq \sum_{\omega \in \mathcal{S}^{(\mu)\Delta}_n} |c'_{\omega}|  |e^{-i\delta t \omega u}| \\
    \leq e^{\delta t \omega_{\mathrm{max}} |u|} \sum_{\omega \in \mathcal{S}^{(\mu)\Delta}_n} |c'_{\omega}|. 
\end{multline*}
Then, the conditions for Bernstein's Inequality, as proved in Theorem 7.24 in \cite{Zygmund}, are all satisfied, which in our case implies that $\forall n,\boldsymbol{U}$
\begin{equation*}
|\partial_{\mu \nu} J_n(\boldsymbol{U})| = \left| \frac{d}{du} F(u) \right| \leq a_{\mathrm{max}}F_{\mathrm{max}} = \delta t \omega_{\mathrm{max}} \frac{\Delta O}{2}.
\end{equation*}
An alternative proof (although of a version of the inequality which is looser by a factor of $2$) can be found in \cite{Pinsky2009}, Sec. 2.3.8.
So the case $P=1$ is settled. If now we suppose the bound
\[ \left|\prod_{p=1}^{P-1} \partial_{\mu_p \nu_p} J_n(\boldsymbol{U})\right| \leq  \left(\delta t \omega_{\mathrm{max}}\right)^{P-1} \frac{\Delta O}{2}\]
to be true, we can now show by the same argument as before that it also holds for $P$. In fact, by defining the restriction $F(u)$ as
\[ F(u) = \prod_{p=1}^{P-1} \partial_{\mu_p \nu_p} \tilde{J}_n(u_{11},\dots,\underbrace{u}_{\mu \nu},\dots,u_{MN}),\]
we have that $F(u)$ satisfies all the conditions as in the $P=1$ case, with $a_{\mathrm{max}}=\delta t \omega_{\mathrm{max}}$ and $F_{\mathrm{max}}=\left(\delta t \omega_{\mathrm{max}}\right)^{P-1} \Delta O/2$. But then by Bernstein's Inequality we have again $\forall n,\boldsymbol{U}$
\begin{multline*}
\left|\prod_{p=1}^P \partial_{\mu_p \nu_p} J_n(\boldsymbol{U})\right| = \left| \frac{d}{du} F(u) \right| \\ \leq a_{\mathrm{max}}F_{\mathrm{max}} = \delta t \omega_{\mathrm{max}} (\delta t \omega_{\mathrm{max}})^{P-1} \frac{\Delta O}{2}.
\end{multline*}
We can then conclude by induction that the bound on the derivatives of $J_n$ is true for all $P \geq 1$.

Finally, we show that the same bounds hold for the function $J$ itself by means of Lemma \ref{lemma:unif_conv}. In fact, for any $u_{\mathrm{max}} \geq 0$ and $\boldsymbol{U} \in [-u_{\mathrm{max}},u_{\mathrm{max}}]^{M \times N}$ we have
\begin{multline*} 
\left|\prod_{p=1}^P \partial_{\mu_p \nu_p} J(\boldsymbol{U})\right| = \left|\prod_{p=1}^P \partial_{\mu_p \nu_p} (J_n(\boldsymbol{U}) + J(\boldsymbol{U}) - J_n(\boldsymbol{U}))  \right| \\
\leq \left|\prod_{p=1}^P \partial_{\mu_p \nu_p} J_n(\boldsymbol{U})\right| + \left|\prod_{p=1}^P \partial_{\mu_p \nu_p}(J(\boldsymbol{U}) - J_n(\boldsymbol{U})) \right| \\
\leq \left(\delta t \omega_{\mathrm{max}}\right)^P \frac{\Delta O}{2} + \epsilon_n(u_{\mathrm{max}}) \xrightarrow[]{n \xrightarrow{} \infty} \left(\delta t \omega_{\mathrm{max}}\right)^P \frac{\Delta O}{2}.
\end{multline*}

The case treated in the main text is easily obtained by setting $M=1$ and noting that in this case $E = \max_{\ket{\psi} \in \mathcal{H}} |\bra{\psi} (\ket{\chi} \bra{\chi}) \ket{\psi}| = |\braket{\chi|\chi}|^2 = 1 $.
\end{proof}

We point out that while we derived this bound on the derivatives as a consequence of the Fourier representation, it is likely not the only way to prove a similar result. Other proofs might be possible e.g. by using the standard formula \cite{Khaneja05}
\[ \frac{d}{ds} \left. e^{A+sB} \right|_{s=0} = e^A \int_0^1 e^{At} B e^{-At} dt \]
or a variation thereof.

Since its derivatives are bounded globally, the function $J$ is Lipschitz continuous. This is a well known result from basic calculus (see e.g. \cite{DeMarco}) that we show here for completeness.
\begin{relemma}{lemma:lipschitz}[Lipschitz continuity]\label{lemma:lipschitz2}
    The function $J(\boldsymbol{U})$ is Lipschitz continuous, that is: $\forall \boldsymbol{U},\boldsymbol{U'} \in \mathbb{R}^{M\times N}$
    \begin{equation*}
        |J(\boldsymbol{U}) - J(\boldsymbol{U'})| \leq K ||\boldsymbol{U} - \boldsymbol{U'}||_1
    \end{equation*}
    where the Lipschitz constant $K \geq 0$ is upper bounded by 
    $\omega_{\mathrm{max}} \delta t \Delta O/2$. Moreover, if $\boldsymbol{U}$ is a critical point, so that $\forall \mu,\nu \ \partial_{\mu \nu}J(\boldsymbol{U})=0$, then the following inequality holds
    \begin{equation*}
        |J(\boldsymbol{U}) - J(\boldsymbol{U'})| \leq K_{c} ||\boldsymbol{U} - \boldsymbol{U'}||_1^2.
    \end{equation*}
    where $0 \leq K_c \leq (\omega_{\mathrm{max}} \delta t)^2 \Delta O/2$.
\end{relemma}
\begin{proof}
    As we already discussed, $J:\mathbb{R}^{M\times N} \xrightarrow{} \mathbb{R}$ is a continuously differentiable function (as it is the restriction to real controls of a holomorphic function) and has bounded partial derivatives. Let us consider its restriction to the straight line connecting $\boldsymbol{U}$ to $\boldsymbol{U'}$
    \[\tilde{J}:\mathbb{R} \xrightarrow{} \mathbb{R}, \tilde{J}(l)= J(\boldsymbol{U}(l)) := J( (\boldsymbol{U} - \boldsymbol{U'})l + \boldsymbol{U'}),\]
    so that $\tilde{J}(1) = J(\boldsymbol{U})$ and $\tilde{J}(0) = J(\boldsymbol{U'})$, and the restricted function is also continuously differentiable, with derivatives given by the chain rule. Because of the mean value theorem (see e.g. Theorem 4.3.1 and 4.3.2 in \cite{DeMarco}), we have
    \[ \tilde{J}(1) - \tilde{J}(0) = \left.\frac{d\tilde{J}}{dl}\right\vert_{l=c} = \sum_{\mu, \nu=1} ^{M,N}\partial_{\mu \nu}J(\boldsymbol{U}(c)) (\boldsymbol{U} - \boldsymbol{U'})_{\mu \nu} \]
    for some $0 \leq c \leq 1$. But then by using the bound from Lemma \ref{lemma:derivatives2} we have the result
    \begin{multline*} 
    |J(\boldsymbol{U}) - J(\boldsymbol{U'})| \leq \sum_{\mu, \nu=1} ^{M,N} |\partial_{\mu \nu}J(\boldsymbol{U}(c))| |\boldsymbol{U} - \boldsymbol{U'}|_{\mu \nu} \\
    \leq  \frac{\omega_{\mathrm{max}} \delta t \Delta O}{2} ||\boldsymbol{U} - \boldsymbol{U'}||_1.
    \end{multline*}
    If $\boldsymbol{U}$ is a critical point $\forall \mu,\nu \ \partial_{\mu \nu}J(\boldsymbol{U})=0$, we have that 
    \[ \left.\frac{d\tilde{J}}{dl}\right\vert_{l=0} = 0, \]
    so we can apply the mean value theorem twice, obtaining:
    \begin{multline*} 
    \tilde{J}(1) - \tilde{J}(0) = \left.\frac{d\tilde{J}}{dl}\right\vert_{l=c} = \left.\frac{d^2\tilde{J}}{dl^2}\right\vert_{l=d} d = \\
    d \sum_{\mu, \nu=1} ^{M,N} \sum_{\mu', \nu'=1} ^{M,N} \partial_{\mu \nu}\partial_{\mu' \nu'}J(\boldsymbol{U}(d)) (\boldsymbol{U} - \boldsymbol{U'})_{\mu \nu} (\boldsymbol{U} - \boldsymbol{U'})_{\mu' \nu'}
    \end{multline*}
    for some $0 \leq d \leq c$. But then by using the bound from Lemma \ref{lemma:derivatives2} we have the result
    \begin{multline*} 
    |J(\boldsymbol{U}) - J(\boldsymbol{U'})| \leq \\
    \sum_{\mu, \nu=1} ^{M,N} \sum_{\mu', \nu'=1} ^{M,N} |\partial_{\mu \nu}\partial_{\mu' \nu'}J(\boldsymbol{U}(d))| |\boldsymbol{U} - \boldsymbol{U'}|_{\mu \nu} |\boldsymbol{U} - \boldsymbol{U'}|_{\mu' \nu'} \\
    \leq \frac{(\omega_{\mathrm{max}} \delta t)^2 \Delta O}{2} ||\boldsymbol{U} - \boldsymbol{U'}||^2_1.
    \end{multline*}
    By setting $M=1$ and $\Delta O=1$, consistently with a fidelity landscape for a single control, we obtain the result from the main text.
\end{proof}

\begin{recorollary}{corollary:parametrization}[Linear parametrizations of controls]
    Let $\boldsymbol{U}(\boldsymbol{V})$ be a linear parametrization of the controls given by $u_{\mu \nu} = \sum_{\mu' \nu'} r_{\mu \nu \mu' \nu'} v_{\mu' \nu'}$, where $\mu' = 1,\dots,M_c$ and $\nu' = 1,\dots,N_c$, which is normalized as follows:
    \[ \forall \mu',\nu'\ \  \sum_{\mu=1}^{M_c} \sum_{\nu=1}^{N_c}   |r_{\mu \nu \mu' \nu'} | = 1.\]
    Then, the partial derivatives of the parametrized landscape $\tilde{J}(\boldsymbol{V}) = J(\boldsymbol{U}(\boldsymbol{V}))$ obey the bounds
    \begin{multline*}
    \left|\prod_{p=1}^P \tilde{\partial}_{\mu'_p \nu'_p} \tilde{J}(\boldsymbol{V})\right|
         \leq (\delta t \omega_{\mathrm{max}})^{P} \frac{\Delta O}{2} \prod_{p=1}^P \sum_{\mu_p,\nu_p=1}^{M_c,N_c} |r_{\mu_p \nu_p \mu'_p \nu'_p}| \\
         = (\delta t \omega_{\mathrm{max}})^{P} \frac{\Delta O}{2}.
    \end{multline*}
    where $\tilde{\partial}_{\mu'\nu'} := \frac{\partial}{\partial v_{\mu'\nu'}}$ and $\partial_{\mu \nu} := \frac{\partial}{\partial u_{\mu \nu}}$.
\end{recorollary}
\begin{proof}
    Let us start with a preliminary step by computing
    \[ \tilde{\partial}_{\mu' \nu'} u_{\mu \nu}(\boldsymbol{V}) = \tilde{\partial}_{\mu' \nu'} \sum_{\alpha=1}^{M_c} \sum_{\beta=1}^{N_c} r_{\mu \nu \alpha \beta} v_{\alpha \beta} = r_{\mu \nu \mu' \nu'}, \]
    with the higher order derivatives are zero because $\boldsymbol{U}(\boldsymbol{V})$ is a linear function. For this reason, by applying the chain rule for differentiation $p$ times and using the bounds from Lemma \ref{lemma:derivatives2} we obtain the first result:
     \begin{multline*}
        \left|\prod_{p=1}^P \tilde{\partial}_{\mu'_p \nu'_p} \tilde{J}(\boldsymbol{V})\right| = \\
        = \left| \prod_{p=1}^P \left( \sum_{\mu_p=1}^{M_c} \sum_{\nu_p=1}^{N_c}   r_{\mu_p \nu_p \mu'_p \nu'_p} \partial_{\mu_p \nu_p} \right) J(\boldsymbol{U}(\boldsymbol{V})) \right| \\
         \leq \left(\delta t \omega_{\mathrm{max}}\right)^P \frac{\Delta O}{2} \prod_{p=1}^P \left(\sum_{\mu_p=1}^{M_c} \sum_{\nu_p=1}^{N_c}   |r_{\mu_p \nu_p \mu'_p \nu'_p}| \right).
    \end{multline*}
    If then we require the condition
    \[ \forall \mu',\nu'\ \  \sum_{\mu=1}^{M_c} \sum_{\nu=1}^{N_c}   |r_{\mu \nu \mu' \nu'} | = 1,\]
    we have the same bound as the one used in the proof of Lemma \ref{lemma:lipschitz2}, that is
    \[ \left|\prod_{p=1}^P \tilde{\partial}_{\mu'_p \nu'_p} \tilde{J}(\boldsymbol{V})\right| \leq \left(\delta t \omega_{\mathrm{max}}\right)^P \frac{\Delta O}{2} \]
    Therefore, using trivially the same reasoning as the one presented in that proof we can obtain the claimed result about Lipschitz continuity with the same upper bound for $K$. By setting $M=1$ and $\Delta O=1$ we recover the result from the main text.
\end{proof}

The following result is also a consequence of the bounds from Lemma \ref{lemma:derivatives2}, and follows by applying a well known approximation result from basic calculus (see e.g. \cite{DeMarco}).
\begin{relemma}{lemma:taylor}[Taylor approximation error in the Lagrange form]
\[  |J(\boldsymbol{U}) - J_P(\boldsymbol{U})| \leq \frac{\Delta O}{2}\frac{(\delta t \omega_{\mathrm{max}})^{P+1}}{(P+1)!}||\boldsymbol{U}-\boldsymbol{U}_0||_1^{P+1} \]
\end{relemma}
\begin{proof}
    Let us define the $P$-order Taylor representation of $J$ as follows:
    \[ J_P(\boldsymbol{U}) = \sum_{p=0}^P \smashoperator[r]{\sum_{\mu_1 \nu_1 \dots \mu_p \nu_p }} \frac{\partial_{\mu_1 \nu_1} \dots \partial_{\mu_p \nu_p}J(\boldsymbol{U}_0)}{p!} \prod_{i=1}^p(\boldsymbol{U}-\boldsymbol{U}_0)_{\mu_i \nu_i}. \]
    The result can then be proved using similar tools as the ones used in Lemma \ref{lemma:lipschitz2}. Let us first define the restrictions $\tilde{J}(l), \tilde{J}_P(l)$ as done previously, using the parametrization \[ l \in [0,1] \mapsto \boldsymbol{U}(l) = (\boldsymbol{U} - \boldsymbol{U}_0)l + \boldsymbol{U}_0. \]
    It is easy to check that the derivatives of $\tilde{J}, \tilde{J}_P$ evaluated at $l=0$ are the same up to order $P$:
    \begin{multline*}
    \left.\frac{d^p\tilde{J}_P}{dl^p}\right\vert_{l=0} = \smashoperator[r]{\sum_{\mu_1 \nu_1 \dots \mu_p \nu_p }} \partial_{\mu_1 \nu_1} \dots \partial_{\mu_p \nu_p}J(\boldsymbol{U}_0)\prod_{i=1}^p(\boldsymbol{U}-\boldsymbol{U}_0)_{\mu_i \nu_i} \\
    =\left.\frac{d^p\tilde{J}}{dl^p}\right\vert_{l=0} \ \ \forall\  0 \leq p \leq P.
    \end{multline*}
    But then we can apply the mean value theorem to the remainder $r(l) := \tilde{J}(l) - \tilde{J}_P(l)$:
    \[
        \frac{d^{P}r(l)}{dl^{P}}= \left.\frac{d^{P+1}r}{dl^{P+1}}\right\vert_{l=c}l = \left.\frac{d^{P+1}\tilde{J}}{dl^{P+1}}\right\vert_{l=c}l
    \]
    where $0\leq c \leq l$. By integrating both sides $P$ times we get:
    \begin{multline*} \int_0^{l}dl_1 \cdots \int_0^{l_{P-1}}dl_P \frac{d^{P}r(l_P)}{dl^{P}} = r(l)  \\
    = \left.\frac{d^{P+1}\tilde{J}}{dl^{P+1}}\right\vert_{l=c} \int_0^{l}dl_1 \cdots \int_0^{l_{P-1}}dl_P l_P \\= \left.\frac{d^{P+1}\tilde{J}}{dl^{P+1}}\right\vert_{l=c}  \frac{l^{(P+1)}}{(P+1)!}.
    \end{multline*}
    By developing the derivative and evaluating the expression at $l=1$ we get:
    \begin{multline*}
    J(\boldsymbol{U}) - J_P(\boldsymbol{U}) = \\ 
    \sum_{\mathclap{\mu_1 \nu_1 \dots \mu_{P+1} \nu_{P+1} }} \frac{\partial_{\mu_1 \nu_1} \dots \partial_{\mu_{P+1} \nu_{P+1}}J(\boldsymbol{U})}{(P+1)!} \prod_{i=1}^{P+1}(\boldsymbol{U}-\boldsymbol{U}_0)_{\mu_i \nu_i} 
    \end{multline*}
    from which the claimed result follows readily by applying the derivative bound found in Lemma \ref{lemma:derivatives2}.
\end{proof}

Under the same assumptions of $M$ multiple controls and generic observable $\hat{O}$ we can also prove the results concerning the landscape variance.
\begin{relemma}{lemma:variance_bounded}[Variance over bounded controls]
\label{lemma:variance_bounded_2}
\begin{widetext}
Let $\mathcal{C} = \prod_{\mu,\nu=1}^{M,N}[-u^{(\mu)}_{\mathrm{max}}, u^{(\mu)}_{\mathrm{max}}]$. Then for any for any $n \geq 1$, $P\geq 0$ and $1 \leq \mu_1,\dots,\mu_P \leq M$, $1 \leq \nu_1,\dots,\nu_P \leq N$ we have:
\begin{multline*} 
\Var_{\boldsymbol{U} \in \mathcal{C}} \left[ \left( \prod_{p=1}^P \partial_{\mu_p \nu_p} \right) J_n \right] = \delta t^{2P} \sum_{\mathclap{\boldsymbol{\Omega}, \boldsymbol{\Omega'} \in (\mathcal{S}^\Delta_n)^N\backslash \{ \boldsymbol{0}\}}} c^*_{\boldsymbol{\Omega}}c_{\boldsymbol{\Omega'}}(\prod_{p=1}^P \omega_{\mu_p \nu_p}\omega_{\mu_p \nu_p}') (\tilde{\kappa}(\boldsymbol{\Omega} - \boldsymbol{\Omega'}) - \tilde{\kappa}(\boldsymbol{\Omega})\tilde{\kappa}(\boldsymbol{\Omega'}) ),\ \ \tilde{\kappa}(\boldsymbol{\Omega}) = \prod_{\mu, \nu =1}^{M,N} \frac{\sin{\delta t \omega_{\mu \nu} u^{(\mu)}_{\mathrm{max}}}}{\delta t \omega_{\mu \nu} u^{(\mu)}_{\mathrm{max}}}
\end{multline*}
Moreover, the following bounds hold:
\[
\Var_{\boldsymbol{U} \in \mathcal{C}}\left[\left(\prod_{p=1}^P \partial_{\mu_p \nu_p}\right) J_n,J\right] 
\leq  \min \left[ \frac{L^{2P} \Delta O^2}{4N^{2P}} , \frac{M(\Delta O L^{P+1} u_{\mathrm{max}})^2}{3N^{2P+1}} + o\left(\frac{L^{2P+2}}{N^{2P+1}}\right) \right] .
\]
\end{widetext}
\end{relemma}
\begin{proof}
This proof is made up of two parts: first, we show how to derive the representation of the variance in the claim, and then we use it to prove the upper bounds.

We remind that given a (real) function $F$, we define its variance over a set $\mathcal{C} \subset \mathbb{R}^{M \times N}$ as
\begin{equation}
\label{eq:var_def}
\Var_{\boldsymbol{U} \in \mathcal{C}}[ F ] = \mathbb{E}_{\boldsymbol{U} \in \mathcal{C}}[ |F|^2 ] - |\mathbb{E}_{\boldsymbol{U} \in \mathcal{C}}[F]|^2,
\end{equation}
where
\[ \mathbb{E}_{\boldsymbol{U} \in \mathcal{C}}[ F ] = \int_{\mathcal{C}} \frac{d^{MN}\boldsymbol{U}}{\Vol(\mathcal{C})}\ F(\boldsymbol{U}) \]
Let us start by computing the expectation value of the landscape approximants $J_n$ and of their derivatives:
\begin{multline}
\label{eq:expect}
    \mathbb{E}_{\boldsymbol{U} \in \mathcal{C}} \left[ \left(\prod_{p=1}^P \partial_{\mu_p \nu_p} J_n \right) \right] =  \int_{\mathcal{C}} \frac{d^{MN}\boldsymbol{U}}{\Vol(\mathcal{C})}\ \left(\prod_{p=1}^P \partial_{\mu_p \nu_p} \right) J_n(\boldsymbol{U}) \\
    =\smashoperator[lr]{\sum_{\boldsymbol{\Omega} \in (\mathcal{S}^\Delta_n)^N}} c_{\boldsymbol{\Omega}} \left( \prod_{p=1}^P -i \delta t\omega_{\mu_p \nu_p} \right) \prod_{\mu, \nu =1}^{M,N}\frac{1}{2u_{\mathrm{max}}^{(\mu)}} \smashoperator[lr]{\int_{-u_{\mathrm{max}}^{(\mu)}}^{u_{\mathrm{max}}^{(\mu)}}}du\ e^{-i \delta t \omega_{\mu \nu} u} \\
    = \smashoperator[lr]{\sum_{\boldsymbol{\Omega} \in (\mathcal{S}^\Delta_n)^N}} c_{\boldsymbol{\Omega}}\left(\prod_{p=1}^P -i \delta t \omega_{\mu_p \nu_p} \right) \prod_{\mu, \nu =1}^{M,N} \frac{\sin{\delta t \omega_{\mu \nu} u_{\mathrm{max}}^{(\mu)}}}{\delta t \omega_{\mu \nu} u_{\mathrm{max}}^{(\mu)}}.
\end{multline}
We can then do the same for the square modulus:
\begin{multline*}
    \mathbb{E}_{\boldsymbol{U} \in \mathcal{C}}\left[ \left| \left( \prod_{p=1}^P \partial_{\mu_p \nu_p}\right) J_n\right|^2 \right] \\
    = \int_{\mathcal{C}} \frac{d^{MN}\boldsymbol{U}}{\Vol(\mathcal{C})}\ \left( \prod_{p=1}^P \partial_{\mu_p \nu_p} \right) J^*_n(\boldsymbol{U}) \left( \prod_{p=1}^P \partial_{\mu_p \nu_p}\right) J_n(\boldsymbol{U}) \\
    = \delta t^{2P} \sum_{\boldsymbol{\Omega}, \boldsymbol{\Omega'} \in (\mathcal{S}^\Delta_n)^N} c^*_{\boldsymbol{\Omega}}c_{\boldsymbol{\Omega'}}\left(\prod_{p=1}^P \omega_{\mu_p \nu_p}  \omega_{\mu_p \nu_p}'\right)\tilde{\kappa}(\boldsymbol{\Omega}'-\boldsymbol{\Omega}).
\end{multline*}
By substituting these results inside Eq.~\eqref{eq:var_def}
and by noticing that 
\[ \tilde{\kappa}(\boldsymbol{0} - \boldsymbol{\Omega'}) - \tilde{\kappa}(\boldsymbol{0})\tilde{\kappa}(\boldsymbol{\Omega'}) = 0 \]
we can restrict all sums to non-zero frequencies in $(\mathcal{S}^\Delta_n)^N\backslash \{ \boldsymbol{0} \}$,
which gives us the result in the claim.
Once again, we obtain the case treated in the main text just by setting $M=1$.

We now briefly comment on the convergence of expectation values and variances over bounded sets. Given a sequence $F_n$ which is uniformly bounded $\forall \boldsymbol{U} \in \mathcal{C}\ |F_n(\boldsymbol{U})| \leq F_{\mathrm{max}}$ and converges uniformly to $F$ over a bounded set $\mathcal{C}$ (with non-zero volume), we have that also its expectation value over $\mathcal{C}$ will converge:
\begin{multline*} 
|\mathbb{E}_{\boldsymbol{U} \in \mathcal{C}}[ F_n ] - \mathbb{E}_{\boldsymbol{U} \in \mathcal{C}}[ F ]| = \\
= \frac{1}{\Vol(\mathcal{C})} \left| \int_{\mathcal{C}}d^{MN}\boldsymbol{U} F_n(\boldsymbol{U}) - F(\boldsymbol{U}) \right| \\
\leq \frac{1}{\Vol(\mathcal{C})}  \int_{\mathcal{C}}d^{MN}\boldsymbol{U} |F_n(\boldsymbol{U}) - F(\boldsymbol{U})| \\
\leq \sup_{\boldsymbol{U} \in \mathcal{C}} |F_n(\boldsymbol{U}) - F(\boldsymbol{U})| \xrightarrow[]{n \to \infty} 0.
\end{multline*}
Moreover, if we have two such bounded converging sequences $F_n,G_n$ then also their product converges uniformly and is bounded $|F_nG_n| \leq F_{\mathrm{max}}G_{\mathrm{max}}$, as can be seen from the following inequality:
\begin{multline*}
|F_nG_n - FG| = |F_nG_n - FG_n + FG_n -FG| \\
\leq |F_n -F||G_n| + |G_n - G||F|.
\end{multline*}
Because of this fact it is easy to see that also the variance (being defined in terms of expectation values of powers of $F_n$) converges as expected:
\begin{multline*}
    \lim_{n \to \infty}\Var_{\boldsymbol{U} \in \mathcal{C}}[ F_n ] = \lim_{n \to \infty}\mathbb{E}_{\boldsymbol{U} \in \mathcal{C}}[ F_n^2 ] - \lim_{n \to \infty}|\mathbb{E}_{\boldsymbol{U} \in \mathcal{C}}[F_n]|^2 \\
    =\mathbb{E}_{\boldsymbol{U} \in \mathcal{C}}[ F^2 ] - |\mathbb{E}_{\boldsymbol{U} \in \mathcal{C}}[F]|^2 = \Var_{\boldsymbol{U} \in \mathcal{C}}[ F ]
\end{multline*}
This shows that the variance and expectation values of the Lie-Fourier representation of the landscape $J_n$ and of its derivatives over $\mathcal{C}$ converge to the corresponding quantities computed for the true landscape $J$.
It should be stressed that this convergence result does not necessarily hold for unbounded controls $\boldsymbol{U} \in \mathbb{R}^{M\times N}$, since in that case we do not have any guarantee of uniform convergence of $J_n$ to $J$.

We now show how to obtain the upper bound in the claim.
The first part of the bound is an obvious consequence of the bound of the derivatives from Lemma~\ref{lemma:derivatives2}, which holds for both $F = J, J_n$. As already discussed there, by defining $\tilde{F}(\boldsymbol{U}) = F(\boldsymbol{U}) - \bar{O}$ we have $\forall \boldsymbol{U} \in \mathbb{R}^{M \times N},\  \tilde{F}(\boldsymbol{U}) \in [-\Delta O/2,\Delta O/2]$, so that the bound
\[ |\partial_{\mu \nu}\tilde{F}| \leq \frac{(\delta t \omega_{\mathrm{max}})^P\Delta O}{2}\]
holds also for $P=0$. Moreover, the variance is not affected by uniform shifts, as
\begin{multline*} 
\Var_{\boldsymbol{U} \in \mathcal{C}}[ \tilde{F} ] = \mathbb{E}_{\boldsymbol{U} \in \mathcal{C}}[ (F - \bar{O} - \mathbb{E}_{\boldsymbol{U} \in \mathcal{C}}[F] + \mathbb{E}_{\boldsymbol{U} \in \mathcal{C}}[\bar{O}])^2 ] \\ 
= \mathbb{E}_{\boldsymbol{U} \in \mathcal{C}}[ (F - \mathbb{E}_{\boldsymbol{U} \in \mathcal{C}}[F] )^2 ] = \Var_{\boldsymbol{U} \in \mathcal{C}}[ F ].
\end{multline*}
But then we have
\begin{multline*}
    \Var_{\boldsymbol{U} \in \mathcal{C}}\left[ \left( \prod_{p=1}^P \partial_{\mu_p \nu_p} \tilde{F} \right) \right] \leq \mathbb{E}_{\boldsymbol{U} \in \mathcal{C}} \left[ \left| \left(\prod_{p=1}^P \partial_{\mu_p \nu_p} \tilde{F} \right) \right|^2 \right] \\
     = \int_{\mathcal{C}}  \frac{d^{MN}\boldsymbol{U}}{\Vol(\mathcal{C})} \ \left| \left(\prod_{p=1}^P \partial_{\mu_p \nu_p} \tilde{F}(\boldsymbol{U}) \right) \right|^2\leq \left( \frac{(T \omega_{\mathrm{max}})^P\Delta O}{2N^P} \right)^2.
\end{multline*}
Concerning the second part of the bound, we first prove the result for finite $n$, starting by finding the asymptotic expansion of the kernel $\tilde{\kappa}(\boldsymbol{\Omega})$ for $\delta t \sim 0$:
    \begin{multline}
        \log(\tilde{\kappa}(\boldsymbol{\Omega})) = \sum_{\mu, \nu =1}^{M,N} \log \left(   \frac{\sin{\delta t \omega_{\mu \nu} u^{(\mu)}_{\mathrm{max}}}}{\delta t \omega_{\mu \nu} u^{(\mu)}_{\mathrm{max}}}\right) =\\
        \sum_{\mu, \nu =1}^{M,N} \left(\frac{\sin{\delta t \omega_{\mu \nu} u^{(\mu)}_{\mathrm{max}}}}{\delta t \omega_{\mu \nu} u^{(\mu)}_{\mathrm{max}}} -1 \right) + o\left(\frac{\sin{\delta t \omega_{\mu \nu} u^{(\mu)}_{\mathrm{max}}}}{\delta t \omega_{\mu \nu} u^{(\mu)}_{\mathrm{max}}} -1 \right) \\
        = \sum_{\mu, \nu =1}^{M,N} 1 - \frac{(\delta t \omega_{\mu \nu} u^{(\mu)}_{\mathrm{max}})^2}{6} -1 + o((\delta t \omega_{\mu \nu} u^{(\mu)}_{\mathrm{max}})^2) \\
        = \sum_{\mu, \nu =1}^{M,N} - \frac{(\delta t \omega_{\mu \nu} u^{(\mu)}_{\mathrm{max}})^2}{6} + o((\delta t \omega_{\mu \nu} u^{(\mu)}_{\mathrm{max}})^2)
    \end{multline}
    Now we take the exponential of both sides and obtain the result:
    \begin{multline*}
        \tilde{\kappa}(\boldsymbol{\Omega}) = \exp\left( \sum_{\mu, \nu =1}^{M,N} - \frac{(\delta t \omega_{\mu \nu} u^{(\mu)}_{\mathrm{max}})^2}{6} + o((\delta t \omega_{\mu \nu} u^{(\mu)}_{\mathrm{max}})^2\right) \\
        = 1 - \frac{1}{6}\sum_{\mu, \nu =1}^{M,N} (\delta t \omega_{\mu \nu} u^{(\mu)}_{\mathrm{max}})^2 + o((\delta t \omega_{\mu \nu} u^{(\mu)}_{\mathrm{max}})^2).
    \end{multline*}
    In order to obtain the claim, we make use of the representation of the variance we found previously and substitute the asymptotic expansion:
    \begin{widetext}
    \begin{multline*}
        \Var_{\boldsymbol{U} \in \mathcal{C}}\left[ \left( \prod_{p=1}^P \partial_{\mu_p \nu_p} \right) J_n \right] = \delta t^{2P} \sum_{\mathclap{\boldsymbol{\Omega}, \boldsymbol{\Omega'} \in (\mathcal{S}^\Delta_n)^N\backslash \{ \boldsymbol{0}\}}} c^*_{\boldsymbol{\Omega}}c_{\boldsymbol{\Omega'}} \left(\prod_{p=1}^P \omega_{\mu_p \nu_p}\omega'_{\mu_p \nu_p} \right)(\tilde{\kappa}(\boldsymbol{\Omega} - \boldsymbol{\Omega'}) - \tilde{\kappa}(\boldsymbol{\Omega})\tilde{\kappa}(\boldsymbol{\Omega'}))\\ 
        =  \sum_{\mathclap{\boldsymbol{\Omega}, \boldsymbol{\Omega'} \in (\mathcal{S}^\Delta_n)^N\backslash \{ \boldsymbol{0}\}}} c^*_{\boldsymbol{\Omega}}c_{\boldsymbol{\Omega'}} \left(\delta t^{2P}\prod_{p=1}^P \omega_{\mu_p \nu_p}\omega'_{\mu_p \nu_p} \right) \sum_{\mu, \nu =1}^{M,N} [ \frac{(\delta t u^{(\mu)}_{\mathrm{max}})^2}{6} (\omega^2_{\mu \nu} + \omega'^2_{\mu \nu} - (\omega_{\mu \nu} - \omega'_{\mu \nu} )^2) + \\
        + o((\delta t \omega_{\mu \nu} u^{(\mu)}_{\mathrm{max}})^2) + o((\delta t \omega'_{\mu \nu} u^{(\mu)}_{\mathrm{max}})^2) + o((\delta t (\omega_{\mu \nu}-\omega'_{\mu \nu}) u^{(\mu)}_{\mathrm{max}})^2)) ] \\
        \leq \left| {\sum_{\boldsymbol{\Omega} \in (\mathcal{S}^\Delta_n)^N\backslash \{ \boldsymbol{0}\}}} c^*_{\boldsymbol{\Omega}} \right| \left|{\sum_{\boldsymbol{\Omega'} \in (\mathcal{S}^\Delta_n)^N\backslash \{ \boldsymbol{0}\}}} c_{\boldsymbol{\Omega'}} \right| MN \left(\frac{((\delta t \omega_{\mathrm{max}})^{P+1} u_{\mathrm{max}})^2}{3}+ o(((\delta t \omega_{\mathrm{max}})^{P+1} u_{\mathrm{max}})^2) \right). 
    \end{multline*}
    \end{widetext}
    In order to upper bound the term containing the sum of Lie-Fourier coefficients, we first notice that
    \[ \mathbb{E}_{\boldsymbol{U} \in \mathbb{R}^{M\times N}}[ J_n ] = c_{\boldsymbol{0}}.\]
    To show that, we fix $\forall \mu\ u_{\mathrm{max}}^{(\mu)} = u_{\mathrm{max}}$ inside Eq.~\eqref{eq:expect} and take the limit $u_{\mathrm{max}} \to \infty$:
\begin{multline*}
    \lim_{u_{\mathrm{max}} \to \infty} \mathbb{E}_{\boldsymbol{U} \in \mathcal{C}_{u_{\mathrm{max}}}}[ \prod_{p=1}^P \partial_{\mu_p \nu_p} J_n ] \\
    =\lim_{u_{\mathrm{max}} \to \infty} \sum_{\boldsymbol{\Omega} \in (\mathcal{S}^\Delta_n)^N} c_{\boldsymbol{\Omega}}\left(\prod_{p=1}^P -i \delta t \omega_{\mu_p \nu_p}\right) \prod_{\mu, \nu =1}^{M,N} \delta_{\omega_{\mu \nu},0} \\
    = 
    \begin{cases} 
    c_{\boldsymbol{0}}, & P=0, \\
    0, & P>0.
    \end{cases}
\end{multline*}
    We also recall from the discussion in Lemma \ref{lemma:derivatives} that
    \[ \forall \boldsymbol{U},n\ \  \min_{\ket{\psi} \in \mathcal{H}} \bra{\psi} \hat{O} \ket{\psi} \leq J_n(\boldsymbol{U}) \leq \max_{\ket{\psi} \in \mathcal{H}} \bra{\psi} \hat{O} \ket{\psi}.\]
    By linearity of the expectation value this implies that
    \begin{multline*}
    \min_{\ket{\psi} \in \mathcal{H}} \bra{\psi} \hat{O} \ket{\psi} \leq \mathbb{E}_{\boldsymbol{U} \in \mathbb{R}^{M\times N}}[ J_n ] \leq \max_{\ket{\psi} \in \mathcal{H}} \bra{\psi} \hat{O} \ket{\psi}
    \end{multline*}
    from which easily follows that
    \begin{multline*} 
    \left| {\sum_{\boldsymbol{\Omega} \in (\mathcal{S}^\Delta_n)^N\backslash \{ \boldsymbol{0}\}} }c_{\boldsymbol{\Omega}} \right| = |J_n(0) - \mathbb{E}_{\boldsymbol{U} \in \mathbb{R}^{M\times N}}[ J_n ]| \\
    \leq \max_{\ket{\psi} \in \mathcal{H}} \bra{\psi} \hat{O} \ket{\psi} - \min_{\ket{\psi} \in \mathcal{H}} \bra{\psi} \hat{O} \ket{\psi} = \Delta O.
    \end{multline*}
    By substituting this result in the chain of inequalities for the variance we obtain the claim. By setting $\Delta O=1$ and $M=1$ we obtain the result from the main text. The result for $J$ comes trivially from taking the limit on both sides, which is possible thanks to the uniform convergence of variance and expectation values we discussed previously in this proof.
\end{proof} 

\begin{relemma}{lemma:variance_unbounded}[Variance over unbounded controls]
\label{lemma:variance_unbounded_2}
\begin{widetext}
For any $n \geq 1$, $P\geq 0$ and $1 \leq \mu_1,\dots,\mu_P \leq M$, $1 \leq \nu_1,\dots,\nu_P \leq N$ we have:
\[
\Var_{\boldsymbol{U} \in \mathbb{R}^{M \times N}}\left[ \left( \prod_{p=1}^P \partial_{\mu_p \nu_p} \right) J_n \right] = \delta t^{2P} \sum_{{\boldsymbol{\Omega} \in (\mathcal{S}^\Delta_n)^N\backslash \{ \boldsymbol{0}\}} } |c_{\boldsymbol{\Omega}}|^2 \prod_{p=1}^P \omega_{\mu_p \nu_p}^2
\]
Moreover, the following bounds hold:
\begin{multline*} 
\Var_{\boldsymbol{U} \in \mathbb{R}^{M \times N}}\left[\left(\prod_{p=1}^P \partial_{\mu_p \nu_p}\right) J_n, J\right] 
\leq \left( \frac{(\delta t \omega_{\mathrm{max}})^P \Delta O}{2} \right)^2,\ \ 
\smashoperator[r]{\sum_{\mu_1, \nu_1, \dots, \mu_P, \nu_P}} \Var_{\boldsymbol{U} \in \mathbb{R}^{M \times N}} \left[ \left( \prod_{p=1}^P \partial_{\mu_p \nu_p} \right) J_n \right] \geq \frac{\Delta J^2 \delta t^{2P}}{4\sum_{\boldsymbol{\Omega} \in (\mathcal{S}^\Delta_n)^N\backslash \{ \boldsymbol{0} \}} \frac{1}{||\boldsymbol{\Omega}||^{2P}_2}} 
\end{multline*}
where we defined $\Delta J$ as the maximum variation in $J_n$
\[\Delta J := \sup_{\boldsymbol{U} \in \mathbb{R}^{M \times N}} J_n(\boldsymbol{U}) -\inf_{\boldsymbol{U} \in \mathbb{R}^{M \times N}} J_n(\boldsymbol{U}).\]
\end{widetext}
\end{relemma}
\begin{proof}
    Let us define $\mathcal{C}_{u_{\mathrm{max}}} := [-u_{\mathrm{max}},u_{\mathrm{max}}]^{M \times N}$. The integrals that define the expectation values of a function $F$ over $\mathbb{R}^{M \times N}$ can then be obtained as limits for $u_{\mathrm{max}} \to \infty$ of integrals over $\mathcal{C}_{u_{\mathrm{max}}}$:
    \begin{equation}
    \label{eq:var_def_unb}
    \Var_{\boldsymbol{U} \in \mathbb{R}^{M \times N}}[ F ] = \lim_{u_{\mathrm{max}} \to \infty} \mathbb{E}_{\boldsymbol{U} \in \mathcal{C}}[ |F|^2 ] - \lim_{u_{\mathrm{max}} \to \infty}|\mathbb{E}_{\boldsymbol{U} \in \mathcal{C}}[F]|^2.
    \end{equation}
    Following the steps already take in the proof of Lemma \ref{lemma:variance_bounded_2}, we start with the expectation value of $J_n$, which was there shown to give rise to
    \[
    \lim_{u_{\mathrm{max}} \to \infty} \mathbb{E}_{\boldsymbol{U} \in \mathcal{C}_{u_{\mathrm{max}}}}\left[ \left( \prod_{p=1}^P \partial_{\mu_p \nu_p} \right) J_n \right] 
  =
    \begin{cases} 
    c_{\boldsymbol{0}}, & P=0, \\
    0, & P>0.
    \end{cases}
    \]
We can then do the same for the square modulus:
\begin{multline*}
    \lim_{u_{\mathrm{max}} \to \infty} \mathbb{E}_{\boldsymbol{U} \in \mathcal{C}_{u_{\mathrm{max}}}}\left[ \left| \left( \prod_{p=1}^P \partial_{\mu_p \nu_p} \right) J_n \right|^2  \right] \\
    = \delta t^{2P} \sum_{\boldsymbol{\Omega}, \boldsymbol{\Omega'} \in (\mathcal{S}^\Delta_n)^N} c^*_{\boldsymbol{\Omega}}c_{\boldsymbol{\Omega'}}(\prod_{p=1}^P \omega_{\mu_p \nu_p}  \omega_{\mu_p \nu_p}') \prod_{\mu, \nu =1}^{M,N} \delta_{\omega_{\mu \nu},\omega'_{\mu \nu}} \\
    = \delta t^{2P} \sum_{\boldsymbol{\Omega} \in (\mathcal{S}^\Delta_n)^N} |c_{\boldsymbol{\Omega}}|^2 \prod_{p=1}^P \omega_{\mu_p \nu_p}^2.
\end{multline*}
By substituting these results inside Eq.~\eqref{eq:var_def_unb}
and by noticing that 
\[ \tilde{\kappa}(\boldsymbol{0} - \boldsymbol{\Omega'}) - \tilde{\kappa}(\boldsymbol{0})\tilde{\kappa}(\boldsymbol{\Omega'}) = 0 \]
we can restrict all sums to non-zero frequencies in $(\mathcal{S}^\Delta_n)^N\backslash \{ \boldsymbol{0} \}$,
which gives us the result in the claim.

Concerning the upper bound in the claim, it is a direct consequence of the corresponding bound for finite $u_{\mathrm{max}}$ that we proved in Lemma \ref{lemma:variance_bounded_2}. Since the bound holds for any $u_{\mathrm{max}}$ it also holds for $u_{\mathrm{max}} \to \infty$ because of the continuity property of the limit.

The lower bound is not as trivial to derive, and only holds for $J_n$ (and not $J$, since the limits in $u_{\mathrm{max}}$ and $n$ do not necessarily commute). Using the representation of the variance we derived previously in this Lemma, we have
    \begin{multline*} 
    \sum_{\mu_1, \nu_1, \dots, \mu_P, \nu_P} \Var_{\boldsymbol{U} \in \mathbb{R}^{M \times N}}\left[ \left( \prod_{p=1}^P \partial_{\mu_p \nu_p} \right) J_n \right]  \\
    = \delta t^{2P} \sum_{\boldsymbol{\Omega} \in (\mathcal{S}^\Delta_n)^N\backslash \{ \boldsymbol{0} \}} |c_{\boldsymbol{\Omega}}|^2 \prod_{p=1}^P \sum_{\mu_p, \nu_p=1}^{M,N} \omega_{\mu_p \nu_p}^2  \\
    =\delta t^{2P} \sum_{\boldsymbol{\Omega} \in (\mathcal{S}^\Delta_n)^N\backslash \{ \boldsymbol{0} \}} |c_{\boldsymbol{\Omega}}|^2 ||\boldsymbol{\Omega}||^{2P}_2
    \end{multline*}
    where the $\boldsymbol{\Omega} = \boldsymbol{0}$ term can be omitted from the sum as it does not contribute to the variance.
    We then notice the following:
    \[ |J_n(\boldsymbol{U}) - c_{\boldsymbol{0}}| = \left| \smashoperator[r]{\sum_{\boldsymbol{\Omega} \in (\mathcal{S}^\Delta_n)^N\backslash \{ \boldsymbol{0} \}}} c_{\boldsymbol{\Omega}} e^{-i\delta t \text{Tr}(\boldsymbol{\Omega}^T \boldsymbol{U})} \right| \leq \smashoperator[r]{\sum_{\boldsymbol{\Omega} \in (\mathcal{S}^\Delta_n)^N\backslash \{ \boldsymbol{0} \}}} |c_{\boldsymbol{\Omega}}| \]
    therefore, by considering separately the two inequalities given by the absolute value, namely
    \[ -\sum_{\boldsymbol{\Omega} \in (\mathcal{S}^\Delta_n)^N\backslash \{ \boldsymbol{0} \}} |c_{\boldsymbol{\Omega}}| \leq J_n(\boldsymbol{U}) - c_{\boldsymbol{0}} \leq \sum_{\boldsymbol{\Omega} \in (\mathcal{S}^\Delta_n)^N\backslash \{ \boldsymbol{0} \}} |c_{\boldsymbol{\Omega}}|\]
    and noting they are true also for any $\boldsymbol{U} \in \mathbb{R}^{M \times N}$, they will hold respectively for the infimum and the supremum:
    \begin{multline*}
        \sup_{\boldsymbol{U} \in \mathbb{R}^{M \times N}} J_n(\boldsymbol{U}) - c_{\boldsymbol{0}} \leq \sum_{\boldsymbol{\Omega} \in (\mathcal{S}^\Delta_n)^N\backslash \{ \boldsymbol{0} \}} |c_{\boldsymbol{\Omega}}| \\ 
        \inf_{\boldsymbol{U} \in \mathbb{R}^{M \times N}} J_n(\boldsymbol{U}) - c_{\boldsymbol{0}} \geq -\sum_{\boldsymbol{\Omega} \in (\mathcal{S}^\Delta_n)^N\backslash \{ \boldsymbol{0} \}} |c_{\boldsymbol{\Omega}}|.
    \end{multline*}
    By subtracting the second equation from the first one we obtain that
    \begin{equation} 
    \label{eq:abssum_lbound}
    \sum_{\boldsymbol{\Omega} \in (\mathcal{S}^\Delta_n)^N\backslash \{ \boldsymbol{0} \}} |c_{\boldsymbol{\Omega}}| \geq \frac{\Delta J}{2}
    \end{equation}
    where we defined the maximum variation $\Delta J$ as in the claim. The lower bound can then be obtained as the global minimum of the function
    \[V(c_{\boldsymbol{\Omega}}, c^*_{\boldsymbol{\Omega}}) := \delta t^{2P} \sum_{\boldsymbol{\Omega} \in (\mathcal{S}^\Delta_n)^N\backslash \{ \boldsymbol{0} \}} |c_{\boldsymbol{\Omega}}|^2 ||\boldsymbol{\Omega}||^{2P}_2\] 
    over $c_{\boldsymbol{\Omega}}$ subject to the constraint given by Eq.~\eqref{eq:abssum_lbound}. 
    
    Let us first study the corresponding problem with equality constraints
    \begin{equation} 
    \label{eq:abssum_eq}
    \sum_{\boldsymbol{\Omega} \in (\mathcal{S}^\Delta_n)^N\backslash \{ \boldsymbol{0} \}} |c_{\boldsymbol{\Omega}}| = \frac{\Delta J}{2}
    \end{equation}
    We can use the method of Lagrange multipliers by finding the stationary point of the cost function $\mathcal{L}$
    \begin{equation*} 
    \mathcal{L}(c_{\boldsymbol{\Omega}}, c^*_{\boldsymbol{\Omega}}) = \sum_{\boldsymbol{\Omega} \in (\mathcal{S}^\Delta_n)^N\backslash \{ \boldsymbol{0} \}} \delta t^{2P} |c_{\boldsymbol{\Omega}}|^2||\boldsymbol{\Omega}||^{2P}_2 - \lambda |c_{\boldsymbol{\Omega}}|.
    \end{equation*}
    Since $\mathcal{L}$ is not differentiable whenever any $c_{\boldsymbol{\Omega}}=0$, we will also have to examine these irregular points along with the stationary points. Thus, the candidate minima are given by
    \begin{multline*}
    \forall \boldsymbol{\Omega} \in (\mathcal{S}^\Delta_n)^N\backslash \{ \boldsymbol{0} \}\   \frac{\partial \mathcal{L}}{\partial c^*_{\boldsymbol{\Omega}} } = c_{\boldsymbol{\Omega}} \left(\delta t^{2P}||\boldsymbol{\Omega}||^{2P}_2 -  \frac{\lambda}{2 |c_{\boldsymbol{\Omega}}|} \right)  = 0 \\
    \implies \forall \boldsymbol{\Omega} \in (\mathcal{S}^\Delta_n)^N\backslash \{ \boldsymbol{0} \}\ c_{\boldsymbol{\Omega}}=0 \lor |c_{\boldsymbol{\Omega}}| = \frac{\lambda}{2 \delta t^{2P}||\boldsymbol{\Omega}||^{2P}_2 }.
    \end{multline*}
    While almost any subset of frequencies can have zero coefficients, not all of them can be be zero, otherwise Eq.~\eqref{eq:abssum_eq} would not be satisfied. We can then define $\mathcal{F} \neq \emptyset$ as the set of frequencies with non-zero coefficients $c_{\boldsymbol{\Omega}} \neq 0$, and fix $\lambda$ using the equality Eq.~\eqref{eq:abssum_eq} as
    \begin{equation} \label{eq:lambda_eq}
    \lambda = \frac{\Delta J \delta t^{2P}}{\sum_{\boldsymbol{\Omega} \in \mathcal{F}\backslash \{ \boldsymbol{0} \}} \frac{1}{||\boldsymbol{\Omega}||^{2P}_2}}.
    \end{equation}
    We can now substitute this result in the expression for the variance:
    \begin{multline*}
        \sum_{\mu_1, \nu_1, \dots, \mu_P, \nu_P} \Var_{\boldsymbol{U} \in \mathbb{R}^{M \times N}} \left[ \left( \prod_{p=1}^P \partial_{\mu_p \nu_p} \right) J_n \right] \\
        =\delta t^{2P} \sum_{\boldsymbol{\Omega} \in \mathcal{F}\backslash \{ \boldsymbol{0} \}} |c_{\boldsymbol{\Omega}}|^2 ||\boldsymbol{\Omega}||^{2P}_2 \\
        = \frac{\lambda^2}{4\delta t^{2P}} \sum_{\boldsymbol{\Omega} \in \mathcal{F}\backslash \{ \boldsymbol{0} \}} \frac{1}{||\boldsymbol{\Omega}||^{2P}_2} = \frac{\Delta J^2 \delta t^{2P}}{4\sum_{\boldsymbol{\Omega} \in \mathcal{F}\backslash \{ \boldsymbol{0} \}} \frac{1}{||\boldsymbol{\Omega}||^{2P}_2}},
    \end{multline*}
    where in the last step we made use of Eq.~\eqref{eq:lambda_eq}. 
    As the denominator is a sum of positive quantities, it is easy to see that the minimum over the equality constraint Eq.~\eqref{eq:abssum_eq} is obtained when all frequencies are non-zero $\mathcal{F}=(\mathcal{S}^\Delta_n)^N$. On the other hand the maximum is obtained when all frequencies are zero except the largest one, for which $||\boldsymbol{\Omega}||^{2}_2 = \omega_{\mathrm{max}}^{2MN}$, and $V_{\mathrm{max}}=\Delta J^2 \delta t^{2P}\omega_{\mathrm{max}}^{2MNP}/4$.
    
    In order to see that the minimum over the equality constraint Eq.~\eqref{eq:abssum_eq} is also the minimum over the inequality constraint Eq.~\eqref{eq:abssum_lbound}, we can reason as follows. Let us consider the following sets
    \begin{multline*}
    \mathcal{E}_{-} = \{ c_{\boldsymbol{\Omega}} \in \mathbb{C}^{n_\Delta^N-1}\ | \  \smashoperator[lr]{\sum_{\boldsymbol{\Omega} \in (\mathcal{S}^\Delta_n)^N\backslash \{ \boldsymbol{0} \}}} |c_{\boldsymbol{\Omega}}| \geq \frac{\Delta J}{2} \}, \\
    \mathcal{E}_{+} = \{ c_{\boldsymbol{\Omega}} \in \mathbb{C}^{n_\Delta^N-1}\ | \ \delta t^{2P}\smashoperator[lr]{ \sum_{\boldsymbol{\Omega} \in (\mathcal{S}^\Delta_n)^N\backslash \{ \boldsymbol{0} \}}} |c_{\boldsymbol{\Omega}}|^2 ||\boldsymbol{\Omega}||^{2P}_2 \leq A \}.
    \end{multline*}
    where $A>0$ and we used the notation $n_{\Delta} = \# \mathcal{S}^\Delta_n$. We will now show that $V$ assumes a global minimum inside $\mathcal{E}_{-}$. Since the two sets are solutions to continuous inequalities they are both closed, and $\mathcal{E}_{+}$ is bounded as $\mathcal{E}_{+} \subseteq \mathcal{E}_{b}$, with
    \[\mathcal{E}_{b} = \{ c_{\boldsymbol{\Omega}} \in \mathbb{C}^{n_\Delta^N-1}\ | \smashoperator[lr]{ \sum_{\boldsymbol{\Omega} \in (\mathcal{S}^\Delta_n)^N\backslash \{ \boldsymbol{0} \}}} |c_{\boldsymbol{\Omega}}|^2  \leq \frac{A}{(\delta t \omega_{\mathrm{min}})^{2P}} \}. \]
    It follows that $\mathcal{E} := \mathcal{E}_{+} \cap \mathcal{E}_{-}$ is a closed and bounded set, and it is straightforward to see that by choosing $A > V_{\mathrm{max}}$ we have $\mathcal{E} \neq \emptyset$. But then the continuous function $V$ admits admits a global minimum (and maximum) over $\mathcal{E}$ because of the extreme value theorem (see for instance Corollary 2.16.2 of \cite{DeMarco}). Now, we have
    \[\mathcal{E}_{-} = (\mathcal{E}_{-} \cap \mathcal{E}_{+}) \cup (\mathcal{E}_{-} \cap (\mathbb{C}^{n_\Delta^N-1} \setminus \mathcal{E}_{+})) =: \mathcal{E} \cup \mathcal{E}'. \]
    Since by definition $V>A$ over $\mathcal{E}'$, $V$ cannot assume its minimum value inside $\mathcal{E}'$ (unless $\mathcal{E}= \emptyset$, which we already excluded), since it assumes smaller values $V \leq A$ inside $\mathcal{E}$. But then the minimum of $V$ over $\mathcal{E}$ is also the minimum over $\mathcal{E}_-$.
    
    In order to conclude, we now show by exclusion that the minimum must lie on the boundary defined by Eq.~\eqref{eq:abssum_eq}.
    The set $\mathcal{E}$ can be further decomposed into $\mathcal{E} = \mathcal{E}_{int} \cup \partial \mathcal{E}_{-} \cup \partial \mathcal{E}_{+}$, with
    \begin{multline*}
    \partial \mathcal{E}_{-} = \{ c_{\boldsymbol{\Omega}} \in \mathbb{C}^{n_\Delta^N-1}\ | \  \smashoperator[lr]{\sum_{\boldsymbol{\Omega} \in (\mathcal{S}^\Delta_n)^N\backslash \{ \boldsymbol{0} \}}} |c_{\boldsymbol{\Omega}}| = \frac{\Delta J}{2} \} \\
    \partial \mathcal{E}_{+} = \{ c_{\boldsymbol{\Omega}} \in \mathbb{C}^{n_\Delta^N-1}\ | \ \delta t^{2P}\smashoperator[lr]{ \sum_{\boldsymbol{\Omega} \in (\mathcal{S}^\Delta_n)^N\backslash \{ \boldsymbol{0} \}}} |c_{\boldsymbol{\Omega}}|^2 ||\boldsymbol{\Omega}||^{2P}_2 = A \} \\
    \mathcal{E}_{int} = \{ c_{\boldsymbol{\Omega}} \in \mathbb{C}^{n_\Delta^N-1}\ | \  \smashoperator[lr]{\sum_{\boldsymbol{\Omega} \in (\mathcal{S}^\Delta_n)^N\backslash \{ \boldsymbol{0} \}}} |c_{\boldsymbol{\Omega}}| > \frac{\Delta J}{2} \land \\
    \ \delta t^{2P}\smashoperator[lr]{ \sum_{\boldsymbol{\Omega} \in (\mathcal{S}^\Delta_n)^N\backslash \{ \boldsymbol{0} \}}} |c_{\boldsymbol{\Omega}}|^2 ||\boldsymbol{\Omega}||^{2P}_2 < A \}
    \end{multline*}
    If $A>V_{\mathrm{max}}$, $V$ cannot assume its minimum value inside $\partial \mathcal{E}_{+}$, as then there would be points on $\partial \mathcal{E}_{-}$ where $V$ would be smaller. That cannot even happen at an internal point belonging to the open set $\mathcal{E}_{int}$, as then it would be a stationary point $\grad V=\boldsymbol{0}$, while the only stationary point of $V$ is $c_{\boldsymbol{\Omega}}=0$, which does not belong to $\mathcal{E}$. The only remaining option is that $V$ assumes its minimum value on $\partial \mathcal{E}_{-}$. This implies that the stationary point we have found with the method of Lagrange multipliers is the minimum of $V$ over $\mathcal{E}_{-}$, that is the minimum of the variance subject to Eq.~\eqref{eq:abssum_lbound}.
    The results cited in the main text follow by fixing $M=1$, $\Delta O=1$.
\end{proof}
\section{Proofs valid only for one control M=1}\label{appendix_d}
\begin{relemma}{lemma:coefficients_l1}[Boundedness of the coefficients]
    \[ \exists r \in \mathbb{R},\ s.t.\ \   \forall n,\  \sum_{\boldsymbol{\omega} \in (\mathcal{S}_n^\Delta)^N} |c_{\boldsymbol{\omega}}(n, \delta t)| \leq r \]
\end{relemma}
\begin{proof}
We start by noticing that the following inequality holds:
\[
\forall i,j\ \  0 \leq |e^{\frac{\hat{X}}{n}}-\hat{I}|_{ij} \leq ||e^{\frac{\hat{X}}{n}}-\hat{I}||_{\infty}  \leq  \frac{||\hat{X}||_{\infty}}{n} e^{\frac{||\hat{X}||_{\infty}}{n}}, 
\]
where we employed Eq.~\eqref{eq:exp_ineq1}. So, since $||\hat{V}||_{\infty}=1$, by defining 
\[ R(n) :=  \frac{||\delta t \hat{H}_d||_{\infty}}{n} e^{\frac{||\delta t \hat{H}_d||_{\infty}}{n}}, \] 
we can write an inequality for $W_{ij}$:
\[
| W_{ij}| = |\delta_{ij} - \delta_{ij} + W_{ij}| \leq \delta_{ij} + |\hat{W} - \hat{I}|_{ij} \leq \delta_{ij} + R(n).
\]
Then, we plug this inequality into the sum of the moduli of the Fourier coefficients:
\begin{multline*}
\sum_{ik\boldsymbol{j}} |\tilde{A}^{\boldsymbol{j}}_{ik}| \leq \sum_{ik\boldsymbol{j}} | W_{i j_{n}}| \cdots |W_{j_2 j_1}||\delta_{j_1 k}| \\ 
\leq \sum_{ij_1 \dots j_{n}} (\delta_{ij_n} + R(n))\cdots(\delta_{j_2 j_{1}}+R(n)) = \\
= D \sum_{m=0}^{n} \binom{n}{m} R(n)^m D^{m} \\
= D \sum_{m=0}^{n} \binom{n}{m}  \left(\frac{D \delta t ||\hat{H}_d|| _{\infty}}{n}\right)^m e^{\frac{m \delta t ||\hat{H}_d||_{\infty}}{n}} = \\
D\left( 1 + \frac{D \delta t ||\hat{H}_d||_{\infty}}{n}e^{\frac{\delta t ||\hat{H}_d||_{\infty}}{n}} \right)^n \xrightarrow{n\to \infty} D e^{D \delta t ||\hat{H}_d||_{\infty}}.
\end{multline*}
The passage from the sum over $\boldsymbol{j}$ to the sum over $m$ is performed by noticing that the expansion of the product in the second line gives rise to terms which consists in $m+1$ independent chains of Kronecker deltas (which can also be trivial), each one of which gives rise to a $D$ factor when summed up over its indices. 
    
Because of this inequalities, the positive sequence $\sum_{ik\boldsymbol{j}} |\tilde{A}^{\boldsymbol{j}}_{ik}|$ is upper-bounded by a (positive) converging sequence. But a positive converging sequence has a finite maximum $A$.
This result implies the boundedness of the $\tilde{B}^{\omega}_{ik}$ coefficients
\begin{multline*} 
\sum_{ik} \sum_{\omega \in \mathcal{S}_n} |\tilde{B}^{\omega}_{ik}| = \sum_{ik} \sum_{\omega \in \mathcal{S}_n} |\sum_{\boldsymbol{j}\in \mathcal{J}_{\omega}} \tilde{A}^{\boldsymbol{j}}_{ik}| \leq \\
\leq \sum_{ik} \sum_{\omega \in \mathcal{S}_n} \sum_{\boldsymbol{j}\in \mathcal{J}_{\omega}} |\tilde{A}^{\boldsymbol{j}}_{ik}| = \sum_{ik\boldsymbol{j}}  |\tilde{A}^{\boldsymbol{j}}_{ik}| \leq A,
\end{multline*} 
and of the coefficients of the fidelity itself:
\begin{widetext}
\begin{multline*} 
\sum_{\boldsymbol{\omega} \in (\mathcal{S}_n^\Delta)^N} |c_{\boldsymbol{\omega}}| = \sum_{\boldsymbol{\omega} \in (\mathcal{S}_n^\Delta)^N} \left| \sum_{\boldsymbol{\omega'} \boldsymbol{\omega''} \in \mathcal{S}_n^N} \delta_{\omega, \omega'-\omega''} \bra{\tilde{\psi}} \tilde{B}^{\boldsymbol{\omega''} \dagger} \ket{\tilde{\chi}}\bra{\tilde{\chi}}   \tilde{B}^{\boldsymbol{\omega'}} \ket{\tilde{\psi}} \right|   \leq \sum_{\boldsymbol{\omega} \in (\mathcal{S}_n^\Delta)^N} \sum_{\boldsymbol{\omega'} \boldsymbol{\omega''} \in \mathcal{S}_n^N} \delta_{\omega, \omega'-\omega''} |\bra{\tilde{\psi}} \tilde{B}^{\boldsymbol{\omega''} \dagger} \ket{\tilde{\chi}}\bra{\tilde{\chi}}   \tilde{B}^{\boldsymbol{\omega'}} \ket{\tilde{\psi}}| \\
=  \left( \sum_{\boldsymbol{\omega'} \in \mathcal{S}_n^N} | \bra{\tilde{\chi}} \tilde{B}^{\boldsymbol{\omega'}} \ket{\tilde{\psi}}| \right)^2 \leq \left( \underbrace{\max_{ij} |\tilde{\chi_i}||\tilde{\psi_j}|}_{\leq 1} \sum_{ij} \sum_{\boldsymbol{\omega'} \in \mathcal{S}_n^N} |\tilde{B}^{\boldsymbol{\omega'}}|_{ij} \right)^2 \leq \left( \sum_{ij} \sum_{\omega^{(1)} \dots \omega^{(N)}}  |\tilde{B}^{{\omega}^{(N)}} \cdots \tilde{B}^{{\omega}^{(1)}}|_{ij} \right)^2   \leq \\
\leq \left( \sum_{ij} \sum_{\omega \in \mathcal{S}_n} |\tilde{B}^{\omega}_{ij}| \right)^{2N} \leq A^{2N}.
\end{multline*}
\end{widetext}
which concludes the proof with $r=A^{2N}$.
\end{proof}

\begin{relemma}{lemma:symmetries}[Symmetries and selection rules]
See main text.
\end{relemma}
\begin{proof}
    Since $[\hat{H}_d,\hat{\Gamma}]=[\hat{H}_c,\hat{\Gamma}]=0$, we can choose $\hat{V}$ in Eq.~\eqref{eq:timestep_unitary} so that both $\hat{\Lambda}=\hat{V}\hat{H}_c\hat{V}^{\dagger}$ and $\hat{\tilde{\Gamma}}=\hat{V} \hat{\Gamma} \hat{V}^{\dagger} $ have diagonal matrix representations:
    \[ \bra{i}\hat{\tilde{\Gamma}}\ket{k}  = \delta_{ik} \gamma_{g(i)}\]
    where $g(i)$ maps the index $i\in\{1,\dots,D\}$ to the corresponding symmetry sector index $g\in\{1,\dots,G\}$. But since also $[\hat{\tilde{H}}_d,\hat{\tilde{\Gamma}}]=0$, then 
    \[ \bra{i}\hat{\tilde{H}}_d\ket{k} = \delta_{g(i) g(k)} [\hat{\tilde{H}}^{(g(i))}_d]_{ik}\]
    must be block diagonal, with blocks $\hat{\tilde{H}}^{(g)}_d$ corresponding to the degenerate eigenspaces of $\hat{\tilde{\Gamma}}$. We name $\mathcal{I}_g$ the set of index values belonging to the $g$-th block:
    \[ \mathcal{I}_g = \{ i\in\{1,\dots,D\}\ |\ g(i)=g\} \]
    This in turn implies that $\hat{W} = \hat{V}e^{-\frac{i\delta t}{n}\hat{H}_d}\hat{V}^{\dagger} = e^{-\frac{i\delta t}{n}\hat{\tilde{H}}_d}$ has block diagonal matrix elements in the same fashion. But then 
    \[
    \tilde{A}^{\boldsymbol{j}}_{ik} = W_{i j_{n}} \cdots W_{j_2 j_1} \delta_{j_1 k} \neq 0
    \]
    only when all the indices $i,k,j_1,\dots,j_n \in \mathcal{I}_g$ belong to the same block. 
    Correspondingly, if we consider products like
    \[
    \hat{\tilde{A}}^{\boldsymbol{J}} := \hat{\tilde{A}}^{\boldsymbol{j}^{(N)}} \cdots \hat{\tilde{A}}^{\boldsymbol{j}^{(1)}},
    \]
    where $\boldsymbol{J}$ is the multiindex 
    \[ \boldsymbol{J} = \begin{pmatrix}
    \boldsymbol{j}^{(1)} & \dots & \boldsymbol{j}^{(N)} \\
    \end{pmatrix} ,
    \]
    we have $\tilde{A}^{\boldsymbol{J}}_{ik} \neq 0$
    only when $i,k \in \mathcal{I}_g$ and $\boldsymbol{j}^{(1)},\dots,\boldsymbol{j}^{(N)} \in \mathcal{I}^n_g$ belong to the same block. 
    This determines selection rules for the frequencies in the spectrum $\mathcal{S}_n$ of the time step operator $\hat{U}_n(u)$
    \[ 
    [\hat{\tilde{U}}_n(u)]_{ik} = \sum_{\boldsymbol{j}\in [D]^n} e^{-i\delta t u \omega_{\boldsymbol{j}}} \tilde{A}^{\boldsymbol{j}}_{ik} = \sum_{\omega \in \mathcal{S}_n}  e^{-i\delta t u \omega} \tilde{B}^{\omega}_{ik},
    \]
    since then the only combinations of eigenvalues that are allowed are the ones of eigenvalues within the same symmetry sector, so that the spectrum becomes
    \begin{multline*}
    \ce{^{(\Gamma)}\mathcal{S}_n} = \bigcup_{g=1}^G \mathcal{S}^{(g)}_n \\
    \mathcal{S}^{(g)}_n = \{ \omega_{\boldsymbol{j}} = \frac{1}{n}(\lambda_{j_1}+\dots+\lambda_{j_n})\  |\  \boldsymbol{j}\in \mathcal{I}_g^n \}.
    \end{multline*}
    Concerning $\tilde{B}^{\omega}_{ik}$, in general we can only say that they are block diagonal in the indices $i,k$ in the same way as $W_{ik}$ and $\tilde{A}^{\boldsymbol{j}}_{ik}$, but the frequency $\omega$ does not select uniquely the block. 
    
    For similar reasons, when it comes to the full time evolution operator $\tilde{U}_n(\boldsymbol{u})$ 
    \begin{multline*}
    \hat{\tilde{U}}_n(\boldsymbol{u})
    = \sum_{\boldsymbol{j}^{(1)} \cdots \boldsymbol{j}^{(N)}} e^{-i\delta t \sum_{ \nu} u_{\nu}\omega_{\boldsymbol{j}^{(\nu)}}} \hat{\tilde{A}}^{\boldsymbol{j}^{(N)}} \cdots \hat{\tilde{A}}^{\boldsymbol{j}^{(1)}}\\
    = \sum_{\boldsymbol{\omega}\in \ce{^{(\Gamma,N)}\mathcal{S}_n}}  e^{-i\delta t \boldsymbol{\omega} \cdot \boldsymbol{u}} \hat{B}^{\boldsymbol{\omega}}
    \end{multline*}
    the only combinations of frequencies for which the coefficients are non-zero are in the set
    \[
        \ce{^{(\Gamma,N)}\mathcal{S}_n} = \bigcup_{g=1}^G (\mathcal{S}^{(g)}_n)^N
    \]
    We can now use this information to write a decomposition for the Lie-Fourier representation of the matrix element 
    \begin{multline*}
    \bra{\chi} \hat{U}_n(\boldsymbol{u}) \ket{\psi} = \bra{\tilde{\chi}} \hat{\tilde{U}}_n(\boldsymbol{u}) \ket{\tilde{\psi}} \\
    =\sum_{i,k}\sum_{\boldsymbol{J}} e^{-i\delta t \sum_{ \nu} u_{\nu}\omega_{\boldsymbol{j}^{(\nu)}}} \tilde{\chi}_{i} \tilde{A}^{\boldsymbol{J} }_{ik} \tilde{\psi}_{k}\\
    =\sum_{g=1}^G\sum_{i,k \in \mathcal{I}_g}\sum_{\boldsymbol{J} \in \mathcal{I}^{n \times N}_g} e^{-i\delta t \sum_{ \nu} u_{\nu}\omega_{\boldsymbol{j}^{(\nu)}}} \tilde{\chi}_{i} \tilde{A}^{\boldsymbol{J} }_{ik} \tilde{\psi}_{k}\\
    =\sum_{g=1}^G\sum_{i,k \in \mathcal{I}_g}\sum_{\boldsymbol{\omega} \in (\mathcal{S}^{(g)}_n)^N} e^{-i\delta t \boldsymbol{\omega} \cdot \boldsymbol{u}} \tilde{\chi}_{i} \tilde{B}^{\boldsymbol{\omega}}_{ik} \tilde{\psi}_{k} \\
    = \sum_{g=1}^G \sum_{\boldsymbol{\omega} \in (\mathcal{S}^{(g)}_n)^N} e^{-i\delta t \boldsymbol{\omega} \cdot \boldsymbol{u}} b^{(g)}_{\boldsymbol{\omega}} \\
    \end{multline*}
    where $b^{(g)}_{\boldsymbol{\omega}} \neq 0$ only if $g \in \mathcal{G}$, defined as  
    \begin{equation}
    \mathcal{G} = \{g =1,\dots,G\ |\ \exists i,k \in \mathcal{I}_g, \tilde{\chi}_{i}\tilde{\psi}_{k} \neq 0\}, 
    \end{equation}
    that is, $b^{(g)}_{\boldsymbol{\omega}}$ is non-zero only if both initial and target state have a non-zero overlap with the eigenstates generating the $g$-th symmetry sector. In fact we have 
    \begin{multline*}
    \tilde{\psi}_k = \bra{k}\hat{V} \ket{\psi} = \bra{k} \sum_{i=1}^D \ket{i}\braket{\gamma_{g(i)},\lambda^{(g(i))}_i|\psi} = \\
    = \braket{\gamma_{g(k)},\lambda^{(g(k))}_k|\psi} \neq 0
    \end{multline*}
    and similarly for $\tilde{\chi}$. But this is equivalent to asking $\hat{P}_g \ket{\psi} \neq 0$ and $\hat{P}_g \ket{\chi} \neq 0$ as in the claim.
    Therefore, we can conclude
    \[\bra{\chi} \hat{U}_n(\boldsymbol{u}) \ket{\psi} = \sum_{g \in \mathcal{G}} \smashoperator[r]{\sum_{\boldsymbol{\omega} \in (\mathcal{S}^{(g)}_n)^N}} e^{-i\delta t \boldsymbol{\omega} \cdot \boldsymbol{u}} b^{(g)}_{\boldsymbol{\omega}} = \smashoperator[lr]{\sum_{\boldsymbol{\omega}\in \ce{^{(\Gamma,N)}\mathcal{S}_n}}}  e^{-i\delta t \boldsymbol{\omega} \cdot \boldsymbol{u}} b_{\boldsymbol{\omega}}. \]
    The spectrum of the fidelity is constructed as usual by taking all possible frequency differences, also across different symmetry sectors, by means of Eq.~\eqref{eq:fid_diff}, giving rise to the frequency set $\ce{^{(\Gamma,N)}\mathcal{S}^\Delta_n}$.
\end{proof}
\section{Details about Ising model landscapes}\label{appendix_e}
\subsection{Computing Lie-Fourier coefficients with the Discrete Fourier Transform (DFT)}
When the single-timestep spectrum $\mathcal{S}^\Delta_n$ of the Lie-Fourier representation $J_n$ of the landscape only contains evenly spaced frequencies, that is
\[ \mathcal{S}^\Delta_n = \{ \omega = \omega_{\mathrm{max}}\frac{k}{k_{\mathrm{max}}}\ |\ k = -k_{\mathrm{max}}, \dots, k_{\mathrm{max}} \}, \]
its coefficients can be computed numerically by means of the DFT. In fact, we can write
\begin{equation}
    J_n(\boldsymbol{u}) = \sum_{\boldsymbol{k} = -\boldsymbol{k}_{\mathrm{max}}}^{\boldsymbol{k}_{\mathrm{max}}} c_{\boldsymbol{k}} e^{-i \frac{\delta t \omega_{\mathrm{max}}}{ k_{\mathrm{max}} } \boldsymbol{k} \cdot \boldsymbol{u}},
    \label{eq:dft_landscape}
\end{equation}
where we used the shorthand $c_{\boldsymbol{k}}= c_{\boldsymbol{\omega}(\boldsymbol{k})}$ and the multiindex notation for $\boldsymbol{k} \in \mathbb{Z}^N$.
Let us now multiply both sides of Eq.~\eqref{eq:dft_landscape} by $\exp(-i\varphi(\boldsymbol{u}))$, with
\begin{equation*}
    \varphi(\boldsymbol{u}) = \delta t \omega_{\mathrm{max}}(\boldsymbol{u} \cdot \boldsymbol{1})
\end{equation*}
and evaluate the expression on a square hyperlattice given by
\begin{equation}
\boldsymbol{u_{j}} = \frac{2\pi\boldsymbol{j}}{\delta t \omega_{\mathrm{max}}} \frac{k_{\mathrm{max}}}{n_\Delta},
\label{eq:dft_sample}
\end{equation}
where $\boldsymbol{j} \in \mathbb{Z}^N$ is an integer multiindex whose elements each range from $0$ to $n_\Delta-1$. 
Then, we obtain the left hand side expressed as the $N$-dimensional DFT of the coefficients
\begin{equation*}
    e^{-i\varphi(\boldsymbol{u_{j}})}J_n(\boldsymbol{u_{j}}) = \sum_{\boldsymbol{k}=\boldsymbol{0}}^{\boldsymbol{n}_\Delta-\boldsymbol{1}} c_{\boldsymbol{k}-\boldsymbol{k}_{\mathrm{max}}} e^{-i \frac{2\pi}{n_\Delta} \boldsymbol{k} \cdot \boldsymbol{j}}
\end{equation*}
which means that the coefficients themselves can be obtained using the inverse DFT:
\begin{equation*}
    c_{\boldsymbol{k}-\boldsymbol{k}_{\mathrm{max}}} = \frac{1}{n_\Delta^N} \sum_{\boldsymbol{j}=\boldsymbol{0}}^{\boldsymbol{n}_\Delta-\boldsymbol{1}} e^{-i\varphi(\boldsymbol{u_{j}})} J_n(\boldsymbol{u_{j}}) e^{i \frac{2\pi}{n_\Delta} \boldsymbol{k} \cdot \boldsymbol{j}}.
\end{equation*}
We note that all the Fourier components, and hence the functions $J_n$, are invariant under the transformation $\boldsymbol{u} \mapsto \boldsymbol{u} + 2 \pi k_{\mathrm{max}} (\delta t \omega_{\mathrm{max}})^{-1} \boldsymbol{m}$, with $\boldsymbol{m} \in \mathbb{Z}^N$. This means that they are all periodic with period $T_n = 2 \pi k_{\mathrm{max}} (\delta t \omega_{\mathrm{max}})^{-1}$ (or an integer multiple thereof) along each dimension.
Because of the orthogonality of these Fourier components in the $L_2[0,T_n]^N$ sense, we can be sure that the coefficients given by the inverse DFT and the ones from the Lie-Fourier expansion are the same.

The DFT and its inverse can be computed efficiently for a sample of $N_s$ points using the Fast Fourier Transform algorithm in $O(N_s \log N_s)$ flops, but the sample defined by Eq.~\eqref{eq:dft_sample} contains $N_s=n_\Delta^N$ points, which prevents us from pushing the computation to large numbers of time steps $N$.

\subsection{Time symmetry in the real case}
Whenever both the Hamiltonian $\hat{H}$ and the states $\ket{\psi}, \ket{\chi}$ defining the state transfer problem only have real matrix elements and overlaps in the same basis, the landscape 
\[ J(\delta t; \boldsymbol{u}) = |\bra{\chi} \hat{U}(\delta t; u_N) \cdots \hat{U}(\delta t; u_1) \ket{\psi}|^2,\]
where $\hat{U}(\delta t; u) = \exp( -i\delta t \hat{H}(u) )$ is symmetric under change of sign of time
\[ J(\delta t) = J(-\delta t).\]
In order to see this, let us first express the overlap in index form:
\begin{multline*}
     \bra{\chi} \hat{U}(\delta t; u_N) \cdots \hat{U}(\delta t; u_1) \ket{\psi} \\ 
     = \sum_{\boldsymbol{j}} \chi_{j_1} U_{j_1 j_2}(\delta t; u_1) \cdots U_{j_{N} j_{N+1}}(\delta t; u_1) \psi_{j_{N+1}}
\end{multline*}
While working in this basis, the Hamiltonians at the different timesteps are going to be real and symmetric matrices, hence they can be diagonalized by means of real orthogonal matrices $V \in \mathbb{R}^{D \times D}, V^TV = \mathds{1}$:
\[ U_{i k}(\delta t; u) = V_{ij}(u) e^{-i \delta t \Lambda_j(u)} V_{jk}^T(u) .\]
By expressing this way every unitary appearing in the overlap, it is easy to see that it can be expressed as a real combination of Fourier components (which are not the same as the ones in the Lie-Fourier representation in the main text):
\begin{multline*}
     \bra{\chi} \hat{U}(\delta t; u_N) \cdots \hat{U}(\delta t; u_1) \ket{\psi} = \\
     = \sum_{\boldsymbol{j}, \boldsymbol{k}} \chi_{j_1} V_{j_1k_1}(u) e^{-i \delta t \Lambda_{k_1}(u)} V_{k_1j_2}^T(u) \cdots \\ \cdots V_{j_N k_N}(u) e^{-i \delta t \Lambda_{k_N}(u)} V_{k_{N}j_{N+1}}^T(u) \psi_{j_{N+1}}
     = \sum_{\boldsymbol{j}, \boldsymbol{k}} r_{\boldsymbol{j}\boldsymbol{k}} e^{-i \delta t \omega_{\boldsymbol{k}}}.
\end{multline*}
But then the square absolute value is time symmetric as expected:
\begin{multline*}
J(\delta t) = \left| \sum_{\boldsymbol{j}, \boldsymbol{k}} r_{\boldsymbol{j}\boldsymbol{k}} e^{-i \delta t \omega_{\boldsymbol{k}}} \right|^2 = \left| \left( \sum_{\boldsymbol{j}, \boldsymbol{k}} r_{\boldsymbol{j}\boldsymbol{k}} e^{-i \delta t \omega_{\boldsymbol{k}}} \right)^* \right|^2 \\
\left| \sum_{\boldsymbol{j}, \boldsymbol{k}} r_{\boldsymbol{j}\boldsymbol{k}} e^{+i \delta t \omega_{\boldsymbol{k}}}  \right|^2 = J(-\delta t).
\end{multline*}

\bibliography{bibliography}

\end{document}